\documentclass[11pt,english,american,fleqn]{article}

\usepackage{amsmath,amsfonts,amssymb,amsthm,mathrsfs}
\usepackage{mathabx}

\usepackage{bbm}

\usepackage{setspace}

\usepackage{enumerate}
\usepackage{bm}
\usepackage{bbm}
\usepackage{verbatim}
\usepackage{xargs} 
\usepackage{authblk}
\usepackage[square,sort,comma,numbers]{natbib}
\usepackage{chngcntr}
\usepackage[font=small]{caption}
\usepackage[hypertexnames=false,hidelinks]{hyperref}

\theoremstyle{plain}\newtheorem{theorem}{Theorem}[section]
\theoremstyle{plain}\newtheorem{lemma}[theorem]{Lemma}
\theoremstyle{plain}\newtheorem{corollary}[theorem]{Corollary}
\theoremstyle{plain}
\theoremstyle{plain}\newtheorem{proposition}[theorem]{Proposition}
\theoremstyle{definition}
\theoremstyle{remark}\newtheorem{remark}{Remark}
\theoremstyle{definition}\newtheorem{def:and:lemma}[theorem]{Definition and Lemma}
\theoremstyle{plain}

\numberwithin{equation}{section}
\counterwithin{remark}{section}
\counterwithin{figure}{section}

\newcounter{remarks}
\newcommand{\D}{\textnormal{d}}
\newcommand{\lsp}{\big \langle }
\newcommand{\rsp}{\big \rangle }

\newcommand{\im}{\operatorname{Im}}
\newcommand{\re}{\operatorname{Re}}

\newcommand{\sno}{\vert\hspace{-1pt}\vert}

\newcommand{\I}{\big|}
\newcommand{\tc}{\textcolor{white}{.}}

\newcommand{\op}{\text{\tiny{$\rm op$}}}
\newcommand{\HS}{\text{\tiny{$\rm HS$}}}

\newcommand{\2}{\text{\tiny{$L^2$}}}
\newcommand{\1}{\text{\tiny{$L^1$}}}
\newcommand{\su}{\text{\tiny{$L^\infty$}}}
\newcommand{\Fock}{\text{\tiny{$\mathcal F$}}}
\newcommand{\h}{\text{\tiny{$\mathscr H$}}}
\newcommand{\PP}{P_\psi}
\newcommand{\QQ}{Q_\psi}

\usepackage[colorinlistoftodos,prependcaption,textsize=tiny]{todonotes}

\usepackage{xcolor}

\begin{document}

\bibliographystyle{alpha}

\title{\Huge Ubiquity of bound states for the\\ strongly coupled polaron}

\author{
David Mitrouskas\thanks{Institute of Science and Technology Austria (ISTA), Am Campus 1, 3400 Klosterneuburg, Austria.\\ Email: \texttt{david.mitrouskas@ist.ac.at}}
\phantom{i}and Robert Seiringer\thanks{Institute of Science and Technology Austria (ISTA), Am Campus 1, 3400 Klosterneuburg, Austria.\\ Email: \texttt{robert.seiringer@ist.ac.at}}
}

%\date{August 15, 2022}
\date{\today}

\maketitle

\frenchspacing

\begin{spacing}{1.15} 
 
\begin{abstract}
We study the spectrum of the Fr\"ohlich Hamiltonian for the polaron  at fixed total momentum. We prove the existence of excited eigenvalues between the ground state energy and the essential spectrum at strong coupling. In fact, our main result shows that the number of excited energy bands  diverges in the strong coupling limit. To prove this we derive upper bounds for the min-max values of the corresponding fiber Hamiltonians and compare them with the bottom of the essential spectrum, a lower bound on which was recently obtained in \cite{Brooks22}. The upper bounds are given in terms of the ground state energy band shifted by momentum-independent excitation energies determined by an effective Hamiltonian of Bogoliubov-type.
\end{abstract}
%\vspace{0mm}
%\noindent \text{MSC class:} 
%\\[0.5mm]
%\noindent \text{Keywords:}

\allowdisplaybreaks

\tableofcontents

\section{Introduction and Main Result}

In this article we study the energy-momentum spectrum of the strongly coupled Fr\"ohlich polaron. The Fr\"ohlich polaron is a translation invariant model for a charged particle interacting with a polar crystal that is described by a continuous non-relativistic quantum field. A particular feature of this model is the assumption that the excitations of the quantum field  have a constant dispersion.

The energy-momentum spectrum is defined as the energy spectrum as a function of the conserved total momentum.  For the Fr\"ohlich polaron,  the essential part of the energy spectrum does not depend on the momentum. For all momenta $P$, it is given by the interval $[E_\alpha(0)+\alpha^{-2},\infty)$ where $E_\alpha(0)$ is the ground state energy at $P=0$ and $\alpha>0$ denotes the coupling constant \cite{Moeller2006,LMM22}. Moreover, for a certain range of momenta around zero the ground state energy lies strictly below the essential spectrum \cite{Miyao2010,Froehlich1974,Spohn1988,Polzer22,MMS22}. These observations naturally pose the question if further bound states exist between the ground state energy and the bottom of the essential spectrum. The aim of this work is to prove that such bound states do exist for large enough values of $\alpha$. In fact, our result shows that the number of energy bands below the essential spectrum diverges in the limit of strong coupling. To our knowledge, this is the first result about the existence of excited bound states for the polaron. In numerical investigations of the problem \cite{Prokofev2000}, they have so far not been seen, which may be due to the necessary restriction to relatively small $\alpha$. 

From a technical point of view, our analysis can be viewed as an extension of \cite{MMS22} where we gave an upper bound on  the ground state energy as a function of the total momentum. A crucial ingredient in our proof is the lower bound on the absolute ground state energy $E_\alpha(0)$ recently obtained in \cite{Brooks22}, which implies a corresponding lower bound on the edge of the essential spectrum.

The Hilbert space of the polaron is
\begin{align}
\mathscr{H} \, =\,  L^2(\mathbb R^3,\D x) \otimes \mathcal F
\end{align}
with $\mathcal F$ the bosonic Fock space over $L^2(\mathbb R^3)$, and the Fr\"ohlich Hamiltonian reads
\begin{align}\label{eq: Froehlich Hamiltonian}
H_\alpha \,  =\,  -\Delta_x + \alpha^{-2} \mathbb N + \alpha^{-1}\phi(h_x),
\end{align}
where $-\Delta_x = (-i \nabla_x)^2$, with $x\in \mathbb R^3$ and $-i\nabla_x$ the position and momentum of the electron, and $\mathbb N = \int a_y^\dagger a_y \, \D y$ the number operator on Fock space. The interaction between the electron and the quantum field is described by the linear field operator
\begin{align}
\phi(h_x) =a^\dagger(h_x) + a(h_x) = \int h_x(y)\, (a_y^\dagger + a_y)\, \D y
\end{align}
with $x$-dependent form factor
\begin{align}
\label{eq: def of h_x(y)}
h_x(y) \,  =\,  - \frac{1}{ 2\pi^2 \vert x - y\vert^2 }.
\end{align}
The creation and annihilation operators satisfy the bosonic canonical commutation relations
\begin{align}
\big[ a(f  ), a^\dagger( g ) \big] \, =\,  \lsp f | g  \rsp_\2, \quad  \big[ a( f ), a( g ) \big]\,  =\,  0 
\end{align}
and we shall adopt the usual notation $a(f) = \int \D y \overline{f(y)} a_y$.

After setting $\hbar = 1$ and $m_{\rm el} = 1/2$ (the mass of the electron) the polaron model depends on a single dimensionless parameter $\alpha>0$. By rescaling all lengths by a factor $1/\alpha$, the operator $\alpha^2 H_\alpha$ is unitarily equivalent to the Hamiltonian $ H_{\alpha}^{\rm {Pol}} \,  =\, -\Delta_x + \mathbb N + \sqrt \alpha \phi(h_x)$ which is more common in the physics literature. Let us further note that \eqref{eq: Froehlich Hamiltonian} is to be understood in the sense of quadratic forms, since $ h_x \notin L^2(\mathbb R^3)$. Since $(1+|\cdot |)^{-1}\widehat h_x \in L^2(\mathbb R^3)$ the quadratic form is bounded from below \cite{Lieb1997,Lieb1958}.

An important property of the Fr\"ohlich Hamiltonian is that it commutes with the total momentum operator, $[H_\alpha, -i\nabla_x + P_f] \, = \, 0$, where 
\begin{align}
P_f %=\D \Gamma(-i\nabla) 
= \int \D y\, a^\dagger_y (-\nabla_y) a_y
\end{align}
denotes the momentum operator of the phonons. It follows that the total momentum and the energy are simultaneously diagonalizable, which is best implemented by the Lee–Low–Pines \cite{LeeLowPines} transformation $S:\mathscr H \to \mathcal F$, $S = F\circ e^{i P_f x } $ where $F$ indicates the Fourier transformation w.r.t. the electron coordinate. A straightforward computation reveals that $S (-i\nabla_x + P_f) S^\dagger = \int_{\mathbb R^3}^\oplus P \, \D P$ and $S H_\alpha S^\dagger = \int^{\oplus}_{\mathbb R^3} H_\alpha(P)\, \D P$ with fiber Hamiltonians
\begin{align}
H_\alpha (P) & \, =\,  (P_f-P)^2 + \alpha^{-2} \mathbb N + \alpha^{-1} \phi (h_0).
\end{align}
The operator $H_\alpha (P)$ acts on the Fock space $\mathcal F$ and describes the system at total momentum $P\in \mathbb R^3$. Using this fiber decomposition we can introduce the energy-momentum spectrum
\begin{align}
\Sigma_\alpha & : = \big\{ (E,P) \in \mathbb R\times \mathbb R^3 \, | \, E \in \sigma (H_\alpha(P)) \big\}.
\end{align}
By rotation invariance of the Hamiltonian, we have $(E,P) \in \Sigma_\alpha \Leftrightarrow (E, RP)\in \Sigma_\alpha$ for every $R \in O(3)$. The function $P \mapsto E_\alpha(P) : = 
 \inf \sigma (H_\alpha(P)) $ is called the energy-momentum relation or ground state energy band. It has a unique global minimum at $P=0$ \cite{LMM22,Polzer22}, thus in particular $E_\alpha( 0) = \inf \sigma (H_\alpha)$. It is known that $E_\alpha(P)$ is strictly increasing below the essential spectrum \cite{Polzer22} and that the set $\{P : E_\alpha(P) < \inf\sigma_{\rm ess}(H_\alpha(P)) \}$ is non-empty \cite{Miyao2010,Froehlich1974,Spohn1988,Polzer22,MMS22}. Moreover, by a Perron--Frobenius type argument the ground state of $H_\alpha(P)$, if it exists, is non-degenerate \cite{Miyao2010}. In this work we are interested in the existence of other discrete eigenvalues in $\Sigma_\alpha$. As already mentioned,  the essential spectrum of $H_\alpha(P)$ is known to be $P$-independent and continuous. For regular polaron models (for instance the Fr\"ohlich Hamiltonian with UV cutoff) the precise location of the essential spectrum has been studied via the HVZ theorem in \cite{Moeller2006} (see also \cite{Froehlich1974}). This can be extended to the Fr\"ohlich Hamiltonian via a suitable approximation argument \cite{GriesemerW2016}. Combining both results, one can prove the following statement \cite{LMM22}.
\begin{lemma}\label{lem:essential:spectrum} For every $\alpha > 0$ and $P\in \mathbb R^3$ the essential spectrum of $H_\alpha(P)$ is given by
\begin{align}\label{eq:essential:spectrum}
\sigma_{\rm ess}(H_\alpha(P)) = [ E_\alpha(0)+\alpha^{-2}, \infty).
\end{align}
\end{lemma}

Before we continue with the main result, let us define the critical momentum (at  strong-coupling)
\begin{align}\label{eq:characteristic:mom}
P_{\rm c}(\alpha) : =\sqrt{2 M^{\rm LP}}\alpha
\end{align}
where $M^{\rm LP} = \frac{2}{3} \sno \nabla\varphi \sno_\2^2$ is the ($\alpha$-independent) semi-classical effective mass \cite{Landau1948,FeliciangeliRS21} (the classical field $\varphi$ is defined in \eqref{eq: optimal phonon mode}). The momentum $P_{\rm c}(\alpha)$ separates the ground state energy band $P\mapsto E_\alpha(P)$ for large $\alpha$ into two different regimes. On one side, there is the quasi-particle regime $|P| \lesssim P_{\rm c}(\alpha)$ where $E_\alpha(P)$ is approximately a quadratic function of $P$ and on the other side, the radiation regime $|P|\gtrsim P_{\rm c}(\alpha)$ where $E_\alpha(P)$ is expected to coincide with the bottom of the essential spectrum; see Figure \ref{fig:my-label} for a sketch of $E_\alpha(P)$. %On the quasi-particle side, 
This picture has been mathematically substantiated in recent works \cite{MMS22,Brooks22,Brooks22b}. 

In the present work we shall  focus on the quasi-particle regime. 
Our main result states that the number of energy bands below the essential spectrum diverges in the limit $\alpha \to \infty$.
\begin{theorem}\label{thm:main:result:1} Let $g :[0,\infty) \to [0,\infty)$ be a function with $\lim_{x\to\infty} g(x) = 0$. %such that $g(x) \xrightarrow{x\to \infty} 0 $. 
For all $|P|\le g(\alpha) P_{\rm c}(\alpha)$, we have $| \sigma_{\rm disc}(H_\alpha(P)) | \to \infty$ as $\alpha \to \infty$, where $\sigma_{\rm disc}(H_\alpha(P))$ is the discrete spectrum of $H_\alpha(P)$. 
\end{theorem}

We are not aware of any previous result about the existence of excited bound states or energy bands for the polaron. A complementary result, however, was recently obtained in \cite{Seiringer22}, proving the absence of excited eigenvalues at $P=0$ for sufficiently weak coupling. (In a similarly fashion \cite{JonasPHD17} showed that for sufficiently weak coupling and small momentum, there are at most two eigenvalues below the essential spectrum). This suggests, somewhat analogous to the case of Schr\"odinger operators with a binding potential, that the number of eigenvalues below the essential spectrum is an increasing function of the coupling constant.

To prove Theorem \ref{thm:main:result:1} we derive upper bounds for the min-max values of $H_\alpha(P)$ and show that for large $\alpha$ they lie strictly below the essential spectrum. By the min-max principle, this implies that the min-max values correspond to eigenvalues of $H_\alpha(P)$. To compare the upper bounds for the min-max values with $\inf \sigma_{\rm ess}(H_\alpha(P))$, we will  use Lemma \ref{lem:essential:spectrum} and a recently obtained two-term lower bound for $E_\alpha(0)$ \cite{Brooks22}. More concretely, our upper bounds are given in terms of the ground state energy band at large coupling shifted by momentum-independent excitation energies. For $|P| \lesssim  P_{\rm c}(\alpha)$ and $\alpha \to \infty$, the ground state energy band is a quadratic function whose coefficients are determined by the semi-classical approximation \cite{MMS22,Brooks22b}. While the leading term in $E_\alpha(P)$ is of order one, the excitation energies are proportional to $\alpha^{-2}$, hence they are comparable to the energy of a free phonon which characterizes the bottom of the essential spectrum. In fact the excitation energies correspond to excitations of the quantum field around its classical value and they  are explicitly determined by the eigenvalues of a quadratic Bogoliubov-type Hamiltonian. 

A qualitative picture of the bounds for the min-max values is shown in Figure \ref{fig:my-label}. The precise statement and its proof require a certain amount of preparations. In the next section we recall the semi-classical theory of the polaron and define the Bogoliubov Hamiltonian that describes the fluctuations of the quantum field. In Section \ref{sec:upper:bounds} we state the upper bounds for the min-max values, which are then used to prove Theorem \ref{thm:main:result:1}. The more technical part concerning the derivation of the upper bounds is postponed to Section \ref{sec:proof:main:proposition}.

\begin{figure}[b!] 
        \center{ \includegraphics[width=0.8\textwidth,height=4.75cm]
        {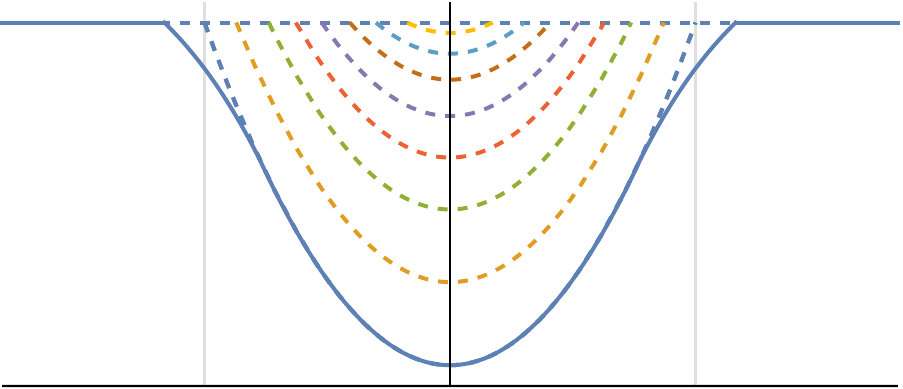} } 
        \caption{\label{fig:my-label}\small{Schematic plot of the ground state energy band $P\mapsto E_\alpha(P)$ (lowest curve) and upper bounds for the excited energy bands (dashed curves). The horizontal line represents the bottom of the essential spectrum and the vertical lines indicate the critical momentum $P_{\rm c}(\alpha) = \sqrt{2 M^{\rm LP}} \alpha$ that separates the parabolic quasi-particle regime from the flat radiation regime. In the limit $\alpha \to \infty$ the eigenvalues accumulate at the edge of the essential spectrum for all $|P|\ll P_c(\alpha)$.
 \label{figure 1}}} 
\end{figure} 

\section{Upper Bounds for the Min-Max Values}

\subsection{Pekar functionals and Bogoliubov Hamiltonian\label{sec:Pekar:Bogoliubov}}

The semi-classical theory of the polaron \cite{Pekar54} arises naturally in the context of strong coupling and corresponds to minimizing the Fröhlich Hamiltonian over product states of the form $ \Psi_{u,v} \,  =\, u \otimes e^{a^\dagger (\alpha v ) - a (\alpha v) } \Omega $,  where $u\in H^1(\mathbb R^3)$ is a normalized electron wave function, $\Omega = (1,0,0,\ldots)$ the Fock space vacuum and $ e^{a^\dagger (\alpha v) - a (\alpha v) } \Omega$ the coherent state associated with the classical field $ v\in L^2(\mathbb R^3)$. A simple computation leads to the Pekar energy 
\begin{align}
\mathcal G(u,v) \, =\,  \frac{\lsp \Psi_{u,v} | H_\alpha   \Psi_{u,v} \rsp_\h }{\lsp \Psi_{u,v} | \Psi_{u,v} \rsp_\h } \, =\,   \lsp  u | (-\Delta + V^{v} ) u \rsp_\2 + \sno v \sno^2_\2
\end{align}
where the polarization potential is given by
\begin{align}\label{eq: def of effective potential}
V^{v}(x) \,  =\,   2 \re \lsp v | h_x \rsp_\2 \,  =\,  - \re \int \frac{ v (y)}{\pi ^2\vert x-y\vert^2 } \D y.
\end{align}
Minimizing over the field variable one obtains the Pekar functional for the electron
\begin{align}\label{eq: electronic pekar functional}
\mathcal E(u) & \,  =\, \inf_{v \in L^2} \mathcal G(u,v) \,  =\, \int \vert \nabla u (x) \vert^2 \D x - \frac{1}{4\pi}\iint \frac{ \vert u(x)\vert^2 \vert u(y)\vert^2 }{ \vert x-y\vert } \D x \D y ,
\end{align}
which is known \cite{Lieb1977} to admit a unique rotation invariant minimizer $\psi >  0$ (the minimizer is unique up to translations and multiplication by a constant phase). Alternatively, one can minimize $\mathcal G$ w.r.t. the electron wave function first. This leads to the Pekar functional for the classical field
\begin{align}
\mathcal F(v) & \,  =\, \inf_{\sno u\sno_\2 = 1} \mathcal G(u,v) \,  =\, \inf \sigma (-\Delta + V^{v}) + \sno v\sno^2_\2 \label{eq: phonon pekar functional}
\end{align}
where $\inf \sigma (-\Delta + V^{v})$ is the ground state energy of the Schr\"odinger operator $-\Delta + V^{v}$ acting on $L^2(\mathbb R^3)$. The unique rotation invariant minimizer of $\mathcal F(v)$ is given by
\begin{align}\label{eq: optimal phonon mode}
\varphi(z)  \,  =\, -  \lsp \psi \I h_\cdot (z)  \psi \rsp_\2 \,  =\,  \int \frac{\vert \psi (y)\vert ^2}{2\pi^2 \vert z-y\vert^2}\D y.
\end{align}
Note that $\psi$ and $\varphi$ are both real-valued (in fact positive) functions.

The semi-classical ground state energy is called the Pekar energy
\begin{align}\label{eq: Pekar energy}
e^{\rm Pek} \,  =\, \mathcal E(\psi) \,  =\, \mathcal F(\varphi),\quad e^{\rm Pek}<0,
\end{align}
and by the variational principle $\inf \sigma (H_\alpha) \le e^{\rm Pek}$. The mathematical validity of Pekar's ansatz was verified by Donsker and Varadhan \cite{Donsker1983} who proved that $ \lim_{\alpha \to \infty} \inf \sigma(H_\alpha) =  e^{\rm Pek}$ and by Lieb and Thomas \cite{Lieb1997} who quantified the error by showing that $\inf \sigma (H_\alpha) \, \ge\,  e^{\rm Pek} + O(\alpha^{ - 1 / 5 })$.

For the potential $V^\varphi$ we consider the Schr\"odinger operator for the electron 
\begin{align}\label{eq: hPek definition}
h^{\rm Pek} \,  =\, -\Delta + V^\varphi - \lambda^{\rm pek}, \quad  \lambda^{\rm Pek} = \inf \sigma (-\Delta + V^{\varphi}),
\end{align}
with $\psi>0$ the unique ground state of $h^{\rm Pek}$ and $\lambda^{\rm Pek} \,  =\,  e^{\rm Pek} - \sno \varphi\sno_\2^2$. It follows from general arguments for Schr\"odinger operators that $h^{\rm Pek}$ has a finite spectral gap above zero and thus the reduced resolvent
\begin{align}\label{eq: def of resolvent}
 R \,  =\,\QQ (h^{\rm Pek})^{-1} \QQ \quad \text{with} \quad \QQ \,  =\, 1-\PP,\quad \PP \,  =\, |\psi \rangle \langle \psi |  
\end{align}
defines a bounded operator ($\PP$ denotes the orthogonal projection onto $\text{Span}\{\psi\}$). Note that $R$ has a real-valued kernel.

Next we introduce the operator $H^{\rm Pek} :L^2(\mathbb R^3) \to L^2(\mathbb R^3) $ with (real-valued) integral kernel
\begin{align}\label{eq: Hessian kernel}
H^{\rm Pek}(y,z) \,  =\,  \delta(y-z) -4   \lsp \psi \I h_\cdot(y) R  h_\cdot(z) \psi \rsp_\2,
\end{align}
which arises as the Hessian of the energy functional \eqref{eq: phonon pekar functional} at its minimizer $\varphi$ (see \cite{MMS22,FeliciangeliRS20,FeliciangeliRS21}). The following lemma was proved in \cite[Lemma 2.1]{MMS22}.
\begin{lemma}\label{lem: Hessian} 
The operator $H^{\rm Pek}$ satisfies the following properties.
\begin{itemize}
\item[(i)] $0\le H^{\rm Pek} \le  1$
\item[(ii)]  $\textnormal{Ker} H^{\rm Pek} = \textnormal{Span}\{\partial_i \varphi \, : \, i = 1,2,3 \}$ with $\varphi \in H^1(\mathbb R^3)$ defined in \eqref{eq: optimal phonon mode}
\item[(iii)] $H^{\rm Pek} \ge \tau >0$ when restricted to $(\textnormal{Ker} H^{\rm Pek})^\perp$ \item[(iv)] ${\rm Tr}_{L^2}(1- \sqrt{H^{\rm Pek}}) < \infty $.
%\item[(iv)] $ 1 - H^{\rm Pek}$ is compact and $\inf \sigma_{\rm ess}(H^{\rm Pek}) =\{1\}$.
\end{itemize}
\end{lemma}
Let $\Pi_0$ and $\Pi_1$ be the orthogonal projections defined by
\begin{align}\label{eq: def of Pi_i}
\text{Ran}(\Pi_0) \, = \, \text{Ker}H^{\rm Pek}, \quad  
\text{Ran}(\Pi_1) \, = \, (\text{Ker}H^{\rm Pek})^\perp.
\end{align}
The decomposition $ L^2 (\mathbb R^3) \, =\, \text{Ran}(\Pi_0) \oplus \text{Ran}(\Pi_1) $
implies the factorization
\begin{align}\label{eq: Fock space factorization}
\mathcal F = \mathcal F_0 \otimes \mathcal F_1 \quad \textnormal{with} \quad \mathcal F_0 = \mathcal F( \text{Ran}(\Pi_0)) \quad \textnormal{and} \quad \mathcal F_1 = \mathcal F( \text{Ran}(\Pi_1)).
\end{align}
Through the isomorphisms 
\begin{subequations}
\begin{align}
\mathcal F_0 & \cong \text{Span}\big\{ (a^\dagger(\partial_1 \varphi))^{j_1} (a^\dagger (\partial_2 \varphi))^{j_2} ( a^\dagger(\partial_3 \varphi)) ^{j_3} \Omega\, | \, (j_1,j_2,j_3) \in \mathbb N^3_0  \big\} \subset \mathcal F \label{eq:iso:0}\\[1mm]
\mathcal F_1 & \cong \text{Ker}(a(\partial_1 \varphi)) \cap \text{Ker}(a(\partial_2 \varphi)) \cap \text{Ker}(a(\partial_3 \varphi)) \subset \mathcal F \label{eq:iso:1}
\end{align}
\end{subequations}
one can view $\mathcal F_0$ and $\mathcal F_1$ also as subspaces of $\mathcal F$.

The Bogoliubov Hamiltonian describing the fluctuations of the quantum field acts non-trivially only in $\mathcal F_1$. For a simple notation, however, we define it directly as an operator on the full Fock space $\mathcal F$. Moreover, for technical reasons, we need to introduce the Bogoliubov Hamiltonian $\mathbb H_K $ with a momentum cutoff $K\in (0,\infty]$. More concretely we set 
\begin{align}\label{eq: Bogoliubov Hamiltonian maintext}
\mathbb  H_K = 1_{\mathcal F_0} \otimes \tilde{\mathbb H}_K 
\end{align}
with the non-trivial part $\tilde {\mathbb H}_K : \mathcal F_1 \to \mathcal F_1$ defined by
\begin{align}
\tilde {\mathbb H}_K \,  =\, \mathbb N_1  - \lsp \psi \I  \phi( h_{K,\cdot}^1 ) R \phi( h_{K,\cdot}^1 )  \psi \rsp_\2  \quad \text{with}\quad \mathbb N_1 = \int \D y \D z\, \Pi_1(y,z) a^\dagger_y a_z ,
\end{align}
where $\mathbb N_1$ corresponds to the number operator in $\mathcal F_1$ and the new form factor is defined by\footnote{Our definition of $\mathbb H_K$ is different from the Bogoliubov Hamiltonian used in \cite{LeopoldMRSS2020,Mitrouskas21} to describe the effective dynamics of the polaron. On the one hand, we have the cutoff $K$, which is needed for technical reasons in the proof. On the other hand, more importantly, there is the additional restriction to $\mathcal F_1$, which is related to the fact that we study the spectrum of the fiber Hamiltonian $H_\alpha(P)$.}
\begin{align}\label{def: cut off coupling function}
h_{K,x}^1(y) \,  =\, \int \D z\,  \Pi_1(y,z) h_{K,x}(z) \quad \text{with}\quad h_{K,x}(y) \,  =\, \frac{1}{(2\pi)^3} \int_{|k|\le K} \frac{e^{ik(x-y)}}{|k|} \D k.
\end{align}
Note that the second term in $\mathbb H _K$ defines the quadratic operator
\begin{align}
&  \lsp \psi \I \phi( h^1_{K,\cdot} ) R \phi( h^1_{K,\cdot} )  \psi \rsp_\2  \notag \\
& \quad \quad    = \iint \D y  \D z \,  \lsp \psi \I (h^1_{K,\cdot}) (y) R ( h^1_{K,\cdot})(z)  \psi \rsp_\2 (a_y^\dagger + a_y) (a_z^\dagger + a_z) .
\end{align}

In Lemma \ref{prop: diagonalization of HBog} below we shall show that $\mathbb H_K$ is bounded from below and diagonalizable by a unitary Bogoliubov transformation. For the precise statement, we need to introduce the operator $H^{\rm Pek}_K$, $K\in (0,\infty]$, defined by
\begin{align}
H^{\rm Pek}_K \restriction  \textnormal{Ran} (\Pi_0) & \, =\,  0 ,  \quad  H^{\rm Pek}_K \restriction  \textnormal{Ran} (\Pi_1)  \, =\,  \Pi_1 - 4 T_K\label{eq: def of H:Pek:K:1}
\end{align}
where the $\Pi_i$ are given by \eqref{eq: def of Pi_i} (note that they do not depend on $K$) and $T_K$ is defined by the integral kernel
\begin{align}\label{eq: Hessian kernel with cutoff}
T_K(y,z)&  \, = \,  \lsp \psi \I  h^1_{K,\cdot}(y) R  h^1_{K,\cdot}(z)  \psi \rsp_\2.
\end{align}
Since $ H^{\rm Pek}  \Pi_1  = \Pi_1  - 4 T_\infty $ it follows that $H^{\rm Pek}_\infty = H^{\rm Pek}$. With $\Theta_K = (H^{\rm Pek}_K)^{1/4}$ we further consider $A_K,B_K:L^2(\mathbb R^3) \to L^2(\mathbb R^3)$ given by 
\begin{subequations}
\begin{align}
A_K \restriction \text{Ran}(\Pi_0)  & \, =\,  \Pi_0 && \hspace{-2cm} B_K \restriction \text{Ran}(\Pi_0)  \, =\,  0  \\
A_K \restriction \text{Ran}(\Pi_1)  & \, =\, \frac{\Theta_K^{-1} + \Theta_K }{2} && \hspace{-2cm} B_K \restriction \text{Ran}(\Pi_1) \, = \, \frac{\Theta_K^{-1} - \Theta_K }{2} . \label{eq: A and B on Pi0}
\end{align}
\end{subequations}
The next lemma, whose proof is given in Section \ref{Sec: Remaining Proofs}, collects useful properties of these operators.

\begin{lemma}\label{lem: regularized Hessian} For $K_0$ large enough and $K\in (K_0,\infty]$ let $H_K^{\rm Pek}:L^2(\mathbb R^3) \to L^2(\mathbb R^3)$ be defined by \eqref{eq: def of H:Pek:K:1}. There exist constants $\beta \in (0,1)$ and $C>0$ such that for all $K\in (K_0,\infty]$
\begin{itemize}
\item[(i)] $0 \le H^{\rm Pek}_K \le 1$ and $ (H^{\rm Pek}_K - \beta) \restriction  \textnormal{Ran} (\Pi_1)  \ge 0$
\item[(ii)] $(B_K)^2 \le C( 1- H_K^{\rm Pek})$
\item[(iii)] $ {\rm{Tr}}_{L^2}(1- H_K^{\rm Pek}) \le C  $
\item[(iv)] ${\rm{Tr}}_{L^2}((-i\nabla)( 1- H_K^{\rm Pek} )(-i\nabla))  \le C K $
\item[(v)] $|\sigma_{\rm disc} (H^{\rm Pek}_K) \cap [0,1)| = \infty$ with only accumulation point at one.
\end{itemize}
Let $(\lambda^{(n)}_K)_{n\in \mathbb N }\subset (0,1)$ be the non-zero eigenvalues of $H_K^{\rm Pek}$ below one  (numbered increasingly counting multiplicity) and denote the corresponding normalized eigenfunctions $\mathfrak u_K^{(n)} \in \textnormal{Ran}(\Pi_1)$. Then
\begin{itemize}
\item[(vi)] $| \lambda_K^{(n)} - \lambda^{(n)}_\infty |  \le C K^{-1/2}$ for all $n\in \mathbb N $, $K\ge K_0$
\item[(vii)] $\sno \nabla \mathfrak u_K^{(n)} \sno_\2 \le C \sqrt{K} (1-\lambda_K^{(n)})^{-1/2} $ for all $n\in \mathbb N $, $K\ge K_0$.
\end{itemize} 
\end{lemma}

%As a remark, if we write $\mathbb H_K = 1_{\mathcal F_0} \otimes \widetilde {\mathbb H}_K$, the operator $\widetilde {\mathbb H}_K$ is the second quantization (up to normal ordering) of the block operator
%\begin{align}
%\mathcal A : \text{Ran}(\Pi_1)^{\otimes 2} \to \text{Ran}(\Pi_1)^{\otimes 2},\quad \mathcal A =  
%\tfrac{1}{4}\begin{pmatrix} 
%2(1 + H_K^{\rm Pek}) &    H_K^{\rm Pek} - 1  \\[1mm]
%H_K^{\rm Pek} - 1  &  2( 1+H_K^{\rm Pek} )
%\end{pmatrix}.
%\end{align}
Next we introduce the transformation
\begin{align} \label{eq: def of U}
\mathbb U^{\textcolor{white}{.}}_K a(f) \mathbb U^\dagger_K \,  & = \,  a( A_K f ) + a^\dagger (  B_{K} \overline{ f} ) \quad \text{for all}\ f\in L^2(\mathbb R^3)
\end{align}
which, for $K$ large enough, defines a unitary  operator $\mathbb U_K :\mathcal F\to \mathcal F$. This follows from the Shale--Stinespring condition (see \cite{Ruijsenaars77,Shale62,JPS2007}) because $A_K$ is bounded and $B_K$ is Hilbert--Schmidt by Lemma \ref{lem: regularized Hessian}. Note that $\mathbb U_K$ acts as the identity in the first component in $\mathcal F = \mathcal F_0 \otimes \mathcal F_1$ and does not mix the two components.

\begin{lemma}\label{prop: diagonalization of HBog} 
For $K\in (K_0,\infty]$ with $K_0$ large enough and $\mathbb U_K$ defined by \eqref{eq: def of U}, we have
\begin{subequations}
\begin{align}\label{eq: diagonalization of HBog} 
\mathbb U^{\textcolor{white}.}_K \mathbb H^{\textcolor{white}.}_K \mathbb U^\dagger_K  \, & = \, \D \Gamma ((H_K^{\rm Pek})^{1/2})  +   \frac{1}{2}\textnormal{Tr}_{ L^2  } \big(  (H_K^{\rm Pek})^{1/2} - 1 \big) + \frac{3}{2}  \\
\label{rem: Bog ground state energy} 
\inf \sigma( \mathbb H_K) \, &  = \,   \frac{1}{2}\textnormal{Tr}_{ L^2  } \big(  (H_K^{\rm Pek})^{1/2} - 1 \big) + \frac{3}{2}.
\end{align}
\end{subequations}
\end{lemma}
The proof is obtained by an explicit computation; for details we refer to \cite[Lemma 2.3]{MMS22}. From Lemma \ref{lem: regularized Hessian}, we have $\inf \sigma(\mathbb H_K) > -\infty$. Note that by \eqref{eq: Bogoliubov Hamiltonian maintext} the statement holds trivially also as an identity for operators on $\mathcal F_1$ if we replace $\mathbb H_K$ by $\tilde {\mathbb H}_K$ (and $\mathbb U_K$ and $H_K^{\rm Pek}$ by their restrictions to $\mathcal F_1$ and $\text{Ran}(\Pi_1)$).

From now on we shall always assume $K\ge K_0$ with $K_0$ large enough such that Lemmas \ref{lem: regularized Hessian} and \ref{prop: diagonalization of HBog} are applicable.

\subsection{Excitation spectrum of the Bogoliubov Hamiltonian\label{sec:Bog:exc:spectrum}}

To introduce the Bogoliubov spectrum we consider the operator $\tilde{ \mathbb H}_K : \mathcal F_1 \to \mathcal F_1$ for $K\in (K_0,\infty]$. We are interested in the excitation spectrum of $\tilde {\mathbb H}_K$ below $\inf \sigma( \tilde {\mathbb H}_K)+1$. This part of the spectrum is purely discrete since it is given by the free bosonic excitation spectrum of particles with energies determined by the non-zero eigenvalues of $H_K^{\rm Pek}$. More concretely, one can write the excitation spectrum below one as
\begin{align}
\sigma \big(\tilde {\mathbb H}_K - \inf \sigma (\tilde {\mathbb H}_K) \big)  \cap [0,1) = \bigcup_{n=0}^\infty \big\{ \Lambda_K^{(n)} \big\}
\end{align}
with the sequence of eigenvalues (counting multiplicities) 
\begin{align}
0= \Lambda_K^{(0)} < \Lambda_K^{(1)} \le \Lambda_K^{(2)} \le \ldots < 1, \quad \Lambda^{(n)}_K =  \sum_{j \in \mathfrak J_K^{(n)}} ( \lambda^{(j)}_K  )^{1/2} \label{eq:Lambda:eigenvalue}
\end{align}
where $\mathfrak J_K^{(n)}$ are suitable (finite) integer sets\footnote{For instance
\begin{align*}
\mathfrak J_K^{(0)}=\varnothing, \quad  \mathfrak J_K^{(1)}=\{1\} \quad \text{and }\quad \mathfrak J_K^{(2)} = \begin{cases}
\{2\}   \quad  & \text{if}\quad \lambda_K^{(2)} < 2 \lambda_K^{(1)}\\
\{1,1\}  & \text{if}  \quad \lambda_K^{(2)} > 2 \lambda_K^{(1)}
\end{cases}.
\end{align*}
} and $0< \lambda^{(1)}_K\le \lambda^{(2)}_K \le \ldots ... < 1$ the sequence of non-zero eigenvalues of $H_K^{\rm Pek}$ below one, cf. Lemma \ref{lem: regularized Hessian}. We shall write $\Lambda^{(n)} = \Lambda_\infty^{(n)}$ for $K=\infty$. The eigenvalues $(\Lambda_K^{(n)})_{n\in \mathbb N_0}$ only accumulate at one and they satisfy
\begin{align}\label{eq:Lambda:lambda:bound}
\Lambda^{(n)}_K\le (  \lambda^{(n)}_K  )^{1/2} \quad \forall n\in \mathbb N.
\end{align}
Generally $\mathfrak J_K^{(n)}$ is a subset of integers drawn from $\{ 1 , \ldots , n \}$ with repetitions allowed. Since $\lambda_K^{(1)} \ge \beta > 0$ (with $\beta $ the constant from Lemma \ref{lem: regularized Hessian}), the length of $\mathfrak J_K^{(n)}$ is bounded by
\begin{align}\label{eq:max:length}
|\mathfrak J_K^{(n)}| \le \mathfrak m\quad \text{with} \quad \mathfrak m : = \big\lceil \, 1/\sqrt \beta \, \big\rceil = \min \big\{ j \in \mathbb N : j \ge 1/\sqrt \beta \big\}.
\end{align}
We emphasize that $\mathfrak m $ does not depend on $K$ and $n$.

\subsection{Upper bounds for the min--max values and proof of Theorem \ref{thm:main:result:1}\label{sec:upper:bounds}}

We define the min-max values of $H_\alpha(P)$ by
\begin{align}\label{def:minmax:values}
\mu_n(H_\alpha(P))  = \inf_{\mathcal V_{n+1} \subset \mathcal F } \, \sup \Big\{ \lsp \Psi |  H_\alpha(P) \Psi \rsp_\Fock \, | \, \Psi \in \mathcal V_{n+1},\ \sno \Psi \sno_\Fock = 1 \Big\}, \quad n \in \mathbb N_0,
\end{align}
where the infimum is taken over all $(n+1)$-dimensional subspaces $\mathcal V_{n+1} \subset \mathcal  Q(H_\alpha(P))$ with $\mathcal Q(H_\alpha(P)) = \mathcal Q(P_f^2 + \mathbb N)$ the form domain of $ H_\alpha(P) $.

\begin{theorem} \label{theorem: main estimate} Let $M^{\rm LP} = \frac{2}{3} \sno \nabla \varphi \sno^2_\2$ with $\varphi$ defined in \eqref{eq: optimal phonon mode}, $e^{\rm Pek}$ the Pekar energy \eqref{eq: Pekar energy} and $H^{\rm Pek}$ defined by \eqref{eq: Hessian kernel}. Moreover let $(\Lambda^{(n)})_{n\in \mathbb N_0}$ be the sequence of eigenvalues introduced in \eqref{eq:Lambda:eigenvalue} for $K=\infty$. For all $\varepsilon>0$ and $n\in \mathbb N_0$ there exist constants $C (n,\varepsilon) > 0$ and $\alpha(n) \ge 1$ such that
\begin{align}\label{eq: main bound}
& \mu_n (H_\alpha(P))  \le   e^{\rm Pek}  + \frac{1}{\alpha^2} \Big( \textnormal{Tr}_{L^2}(\sqrt{H^{\rm Pek}} - 1 )  + \Lambda^{(n)}   \Big) + \frac{P^2}{2 \alpha^4 M^{\rm LP}} +   C(n,\varepsilon)  \alpha^{-\frac{5}{2}+\varepsilon}
\end{align}
for all $P\in \mathbb R^3$ and $\alpha \ge  \alpha(n)$.
\end{theorem}

\begin{remark}\label{rem:main:statement:1} For $n=0$, since $E_\alpha(P) = \mu_0(H_\alpha(P))$ and $\Lambda ^{(n)}=0$, this coincides with \cite[Theorem 2.1]{MMS22}. A compatible lower bound on $\mu_0(H_\alpha(P))$ was obtained in \cite{Brooks22b}. 
\end{remark}

\begin{remark} We expect that the min-max values $\mu_n(H_\alpha(P))$ are in fact asymptotically described, as $\alpha \to \infty$, by the r.h.s. in \eqref{eq: main bound}. In other words, that the excited energy bands are given by the ground state energy band $P\mapsto E_\alpha(P)$ shifted by the momentum-independent excitation energies $\Lambda^{(n)}$. The derivation of a compatible lower bound, however, seems out of reach with current methods.
\end{remark}

Using Theorem \ref{theorem: main estimate} in combination with Lemma \ref{lem:essential:spectrum} and the lower bound for $E_\alpha(0)$ obtained in \cite{Brooks22}, we can prove our main result.

\begin{proof}[Proof of Theorem \ref{thm:main:result:1}] By the min-max theorem the inequality $\mu_n(H_\alpha(P)) < \inf\sigma_{\rm ess}(H_\alpha(P))$ implies that $\mu_n(H_\alpha(P))$ corresponds to the $(n+1)$th discrete eigenvalue of $H_\alpha(P)$ when counted with multiplicity starting from the lowest eigenvalue.

The main result in \cite{Brooks22} is the following lower bound for the ground state energy of \eqref{eq: Froehlich Hamiltonian}
\begin{align}\label{eq:lower:bound:Brooks}
\inf\sigma (H_\alpha) \ge e^{\rm Pek} + \frac{1}{2 \alpha^2} \text{Tr}_{L^2}(\sqrt{H^{\rm Pek}} -1 )   - \frac{1}{\alpha^{2+ s }}
\end{align}
for all $0<s< \tfrac{1}{29}$ and $\alpha\ge \alpha(s)$ for some constant $\alpha(s)$. Since $E_\alpha(0) = \inf\sigma (H_\alpha)$, we can combine \eqref{eq: main bound} for $ \varepsilon = \tfrac{7}{15}$ and \eqref{eq:lower:bound:Brooks} for $s=\tfrac{1}{30}$ to
\begin{align}
& \mu_n (H_\alpha(P)) \, \le \, E_\alpha(0) + \frac{1}{\alpha^2} \Lambda^{(n)} + \frac{P^2}{2\alpha^4 M^{\rm LP}} + C(n) \frac{1}{\alpha^{2+1/30}}
\end{align}
for all $\alpha \ge \alpha(n)$. Invoking Lemma \ref{lem:essential:spectrum} and \eqref{eq:Lambda:lambda:bound} leads to
\begin{align}\label{eq:condition:essential}
& \mu_n (H_\alpha(P)) - \inf \sigma_{\rm ess}(H_\alpha(P)) \le \frac{1}{\alpha^2} \bigg[ (\lambda_\infty^{(n)})^{1/2}  - 1 + \frac{P^2}{2 \alpha^2 M^{\rm LP}}  +   C(n)   \frac{1}{\alpha^{1/30}} \bigg]
\end{align}
and thus by the min-max theorem
\begin{align}
& |  \sigma_{\rm disc}(H_\alpha(P)) |  \ge \sup \bigg\{n+1 \in \mathbb N \, : \, \underbrace{ \frac{P^2}{2 \alpha^2 M^{\rm LP}} +   C (n)  \frac{1}{\alpha^{1/30}}  < 1 -   (\lambda^{(n)}_\infty)^{1/2} }_{(\star)} \, \bigg\}.
\end{align}
Since $\lambda^{(n)}_\infty \xrightarrow{n\to \infty}1$ but $\lambda^{(n)}_\infty < 1 $ for all $n \in \mathbb N$ by Lemma \ref{lem: regularized Hessian}, the r.h.s. of ($\star $) is always larger than some small $\delta(n) > 0$. On the l.h.s. the second term is smaller than $\delta(n)/4$ for $\alpha \ge \alpha(n)$ with $\alpha(n)$ sufficiently large. For $ |P| \le g(\alpha) \sqrt{2 M^{\rm LP} } \alpha$ with $g(\alpha) \xrightarrow{\alpha\to \infty}0$, also the first term is smaller than $\delta(n)/4$ for $\alpha\ge \alpha(n)$ with $\alpha(n)$ large enough. This implies the statement of Theorem \ref{thm:main:result:1}.
\end{proof}

In the next section we prove Theorem \ref{theorem: main estimate}. To this end we first introduce a suitable family of trial subspaces and state the corresponding variational estimates in Proposition \ref{theorem: main estimate 2}. The proof of Theorem \ref{theorem: main estimate} is obtained as a corollary of Proposition \ref{theorem: main estimate 2}.

\subsection{Variational estimates and proof of Theorem \ref{theorem: main estimate}\label{Sec: trial state}} 

To estimate the min-max values, we need to introduce a suitable family of trial subspaces. We first define the linear  operator $\mathcal S_P : \mathcal F \to \mathcal F$, for $P\in \mathbb R^3$, given by
\begin{align}\label{eq:SP:map}
\Gamma \mapsto \mathcal S_P\Gamma =  \, \int \D x \,  e^{ i ( P_f - P ) x }\, e^{a^\dagger(\alpha \varphi_P) - a(\alpha \varphi_P)} \big( \mathscr G^{0}(x) - \tfrac{1}{\alpha} \mathscr G^{1}(x) \big)  \Gamma,
\end{align}
where
\begin{align}\label{eq: def of varphi_P}
\varphi_P \, =\,  \varphi + \xi_P  \quad \text{with} \quad \xi_P \,  =\, \frac{i}{\alpha^2 M^{\rm LP}} (P  \nabla ) \varphi, \quad  M^{\rm LP} = \frac{2}{3} \sno \nabla \varphi \sno^2_\2,
\end{align}
and
\begin{align}
\mathscr G^{0}(x) & \,  =\,  \psi (x) , \quad \mathscr G^{1}(x) = u_\alpha(x) (R \phi( h^1_{K,\cdot} ) \psi )(x)
\end{align}
with form factor $h_{K,\cdot}^1$ defined in \eqref{def: cut off coupling function}, and $u_\alpha :\mathbb R^3 \to [0,1]$ a radial function satisfying
\begin{align}\label{eq: properties of ualpha}
u_\alpha (x) \,  =\,
\begin{cases}
  1 & \forall \ |x| \le \alpha  \\
  0 & \forall\ |x| \ge  2\alpha
\end{cases} \quad \quad\text{and}\quad \quad \sno \nabla u_\alpha \sno_\su + \sno \Delta u_\alpha \sno_\su  \, \le \,  C \alpha^{-1}
\end{align}
for some constant $C>0$. For completeness, we recall that $\psi > 0$ and $\varphi > 0 $ are the unique rotation invariant minimizers of the Pekar functionals \eqref{eq: electronic pekar functional} and \eqref {eq: phonon pekar functional}, respectively. The operator $\mathcal S_P$ depends of course also on $\alpha$, $K$ and $u_\alpha$, but we omit this in our notation. 

\begin{remark} \label{remark: HS isomorphism}
Note that $\mathscr G^{i}(x)\Gamma$ are elements of $L^2(\mathbb R^3, \mathcal F)$ and
\begin{align}
(R \phi( h^1_{K,\cdot} ) \psi )(x) \,  =\, \iint \D z \D y \, R(x,y) h_{K,y}^1 (z)  \psi(y)\,  \big( a^\dagger_z + a_z \big)
\end{align}
defines an $x$-dependent Fock space operator. Via $L^2(\mathbb R^3, \mathcal F) \simeq \mathscr H$, we can view these states also as elements in $ \mathscr H$. In this case we shall use the notation
\begin{align}\label{eq: alternative definition of G}
\mathscr G^{0}_{\Gamma} \,  =\,  \psi \otimes  \Gamma , \quad \mathscr G_{\Gamma}^{1} \,  =\,  u_\alpha R \phi(h^1_{K,\cdot}) \psi \otimes \Gamma .
\end{align}
\end{remark}

%\begin{lemma}\textcolor{blue}{$\mathcal S_P \mathcal Q(P_f^2 + \mathbb N ) \subset \mathcal Q(P_f^2 + \mathbb N) $ for the form domain of $P_f^2 +\mathbb N$.}
%\end{lemma}

To specify the choice of the state $\Gamma$, let us recall that the eigenfunctions of $H_K^{\rm Pek}$ corresponding to non-zero eigenvalues below one are denoted by $\mathfrak u_K^{(n)} \in \text{Ran}(\Pi_1)$, $n\in \mathbb N$, cf. Lemma \ref{lem: regularized Hessian}. We define
\begin{align}
\label{eq:excited:states:Gamma}
\Gamma^{(n)}_K & = \mathbb U^\dagger_K \gamma_K^{(n)}  \quad \text{with}\quad  \gamma_K^{(n)} = \prod_{j \in \mathfrak J_K^{(n)}} a^\dagger(\mathfrak u_K^{(j)}) \Omega, \quad  n\in \mathbb N_0
\end{align}
where $\Omega \in \mathcal F$ is the Fock space vacuum,  $\mathbb U_K$ the unitary operator \eqref{eq: def of U} and $(\mathfrak J_K^{(n)})_{n\in \mathbb N_0}$ the sequence of integer sets defined by \eqref{eq:Lambda:eigenvalue}. Since $\mathfrak J_K^{(0)} = \varnothing$, $\gamma_K^{(0)} = \Omega$. For $n\ge 1$, we have $|\mathfrak J_K^{(n)}|\le \mathfrak m$ and $\mathfrak u_K^{(n)} \in \text{Ran}(\Pi_1)$, and thus
\begin{align}\label{eq:Fockspace:m:1}
\gamma_K^{(n)} \in \mathcal F^{(\le \mathfrak m)}_1 \quad \forall n,K \quad \text{with} \quad \mathcal F_1^{(\le \mathfrak m)} : =  \mathcal F_1 \cap \mathcal F^{(\le \mathfrak m)} \subset \mathcal F,
\end{align}
where $\mathcal F_1 \subset \mathcal F$ was introduced in \eqref{eq: Fock space factorization} and $\mathcal F^{(\le \mathfrak m)}$ denotes the truncated Fock space
\begin{align}\label{eq:Fock:space:m:particles}
\mathcal F^{(\le \mathfrak m)} = \bigoplus_{j=0}^{\mathfrak m} \mathcal F^{(j)},\quad \mathcal F^{(j)} =   L^2(\mathbb R^3 )^{\otimes_{\rm{sym}} j } .
\end{align} 
Let us now consider the linear subspace
\begin{align}\label{eq:Bogo:subspace}
\mathcal V_K^{(n+1)}  = \textnormal{Span}\big\{ \Gamma^{(0)}_K,\ldots , \Gamma_K^{(n)} \big\}\subset \mathcal F_1 .
\end{align}
With $\langle \Gamma_K^{(i)} | \Gamma_K^{(j)} \rangle_\Fock = \delta_{ij}$ it follows from Lemma \ref{prop: diagonalization of HBog} that
\begin{align}\label{eq:sup:Bog:ev}
\sup \big\{ \lsp \Gamma | \mathbb H_K \Gamma \rsp_\Fock \, | \, \Gamma\in \mathcal  V_K^{(n+1)},\ \sno \Gamma\sno_\Fock = 1 \big\} \, =\, \inf \sigma(\mathbb H_K) + \Lambda_K^{(n)}.
\end{align}

The next proposition provides a variational bound on the energy for states of the form $\Psi_P = \mathcal S_P \Gamma$ with $\mathcal S_P$ defined in \eqref{eq:SP:map} and $\Gamma\in \mathcal V_K^{(n+1)} $. 

%\textcolor{blue}{Note that by Lemmas \ref{lem: bounds for the number operator} and \ref{lem: bounds for P_f:new} below, it follows that $ \mathcal V_K^{(n+1)} \subset \mathcal Q(P_f^2 + \mathbb N)$ for all finite $K\ge K_0$.}

\begin{proposition}\label{theorem: main estimate 2} Let $M^{\rm LP} = \frac{2}{3} \sno \nabla \varphi \sno^2_\2$ with $\varphi$ defined in \eqref{eq: optimal phonon mode}, $e^{\rm Pek}$ the Pekar energy \eqref{eq: Pekar energy} and $H_K^{\rm Pek}$ defined by \eqref{eq: Hessian kernel}. For all $\varepsilon >0$ and $n\in \mathbb N_0$ there exist constants $C(n,\varepsilon)> 0$ and $\alpha(n)\ge 1$ such that
\begin{align}\label{eq: main estimate 2}
& \Bigg| \frac{ \lsp \mathcal S_P \Gamma | H_\alpha(P) \mathcal S_P \Gamma \rsp_{\Fock}}{ \lsp \mathcal S_P \Gamma | \mathcal S_P \Gamma \rsp_\Fock}    - \bigg( e^{\rm Pek} + \frac{1}{\alpha^2}\Big( \lsp \Gamma | \mathbb H_K \Gamma \rsp_\Fock  - \tfrac{3}{2}  \Big) + \frac{P^2}{2 \alpha^4 M^{\rm LP}} \bigg) \Bigg| \notag\\[1mm]
& \hspace{7cm} \le C(n,\varepsilon) \alpha^\varepsilon \big( K^{-1/2} \alpha^{-2} + K^{1/2} \alpha^{-3} \big)  
\end{align}
for all normalized $\Gamma \in \mathcal V_K^{(n+1)}$, $|P|  \le \sqrt{2 M^{\rm LP}} \alpha $, $ K_0 \le K \le \alpha $ and $\alpha \ge \alpha(n)$.
\end{proposition}

To use Proposition \ref{theorem: main estimate 2} to estimate the min-max values, we have to show that the subspace $\mathcal S_P \mathcal V_K^{(n+1)}$ has dimension $n+1$. That this holds for large enough $\alpha$ is the content of the following lemma, whose proof is given in Section \ref{Sec: Remaining Proofs}. 

\begin{lemma}\label{lem:linear:independence} For every $n\in \mathbb N$ there exists a constant $\alpha(n)\ge 1 $ such that  
\begin{align}
\textnormal{dim} \big( \mathcal S_P \mathcal V_K^{(n+1)} \big)  = n+1
\end{align}
for all $|P| \le \sqrt{2M^{\rm LP}} \alpha $, $K_0\le K\le \alpha$ and $\alpha \ge \alpha(n)$. 
\end{lemma}
We can now combine Proposition \ref{theorem: main estimate 2} and Lemmas \ref{lem: regularized Hessian} and \ref{lem:linear:independence} to prove Theorem \ref{theorem: main estimate}.

\begin{proof}[Proof of Theorem \ref{theorem: main estimate}] For $|P|> \sqrt{2M^{\rm LP} } \alpha$, we use $\mu_n(H_\alpha(P)) \le \sigma_{\rm ess}(H_\alpha(P)) \le \mu_n(H_\alpha(0)) + \alpha^{-2}$ by definition of the min-max values and Lemma \ref{lem:essential:spectrum}. Hence it is sufficient to consider $|P|\le \sqrt{2 M^{\rm LP}} \alpha$. 

With the aid of Lemma \ref{lem:linear:independence} we can choose $\mathcal V_{n+1} = \mathcal S_P \mathcal V_K^{(n+1)}  $ in \eqref{def:minmax:values}. Applying Proposition \ref{theorem: main estimate 2} together with \eqref{eq:sup:Bog:ev}, we obtain
\begin{align}\label{eq: main estimate 3}
\mu_n (H_\alpha(P))  & \, \le \,  e^{\rm Pek} + \frac{1}{\alpha^2}\Big( \inf \sigma (\mathbb H_K) + \Lambda_K^{(n)} - \tfrac{3}{2} \Big) \notag\\
& \hspace{2.25cm} + \frac{P^2}{2 \alpha^4 M^{\rm LP}}  + C (n,\varepsilon ) \alpha^{\varepsilon}  \big( K^{-1/2}\alpha^{-2} + \sqrt K \alpha^{-3} \big)
\end{align}
for all $|P| \le \sqrt {2 M^{\rm LP} } \alpha $, $K_0 \le K\le \alpha $ and $\alpha \ge \alpha(n)$. Next we use $| \inf \sigma (\mathbb H_K ) - \inf \sigma (\mathbb H_\infty)  | \le C K^{-1/2}$ (see \cite[Section 3.9]{MMS22}) and $ |  \Lambda^{(n)}_K  -  \Lambda^{(n)}   | \le C K^{-1/2} $. The second inequality is a consequence of \eqref{eq:Lambda:eigenvalue} together with the uniform bound \eqref{eq:max:length} and $|(\lambda_K^{(n)})^{1/2} - (\lambda_\infty^{(n)})^{1/2} | \le C K^{-1/2} $. The latter is obtained from Lemma \ref{lem: regularized Hessian} (vi) and $\lambda_K^{(n)}\ge \beta >0$. The proof of Theorem \ref{theorem: main estimate} now follows by using \eqref{rem: Bog ground state energy} for $K=\infty$, and choosing $K=\alpha$ in the term quantifying the error.
\end{proof}

The remainder of the paper is dedicated to the proof of Proposition \ref{theorem: main estimate 2}. After reviewing some general relations, we start in Section \ref{sec:total:energy} with a suitable formula for the expectation value of $H_\alpha(P)$ in states of the form $\mathcal S_P \Gamma$. Sections \ref{sec:Gaussian:lemma} -- \ref{sec:technical:lemmas} collect further preparations for the proof of Proposition \ref{theorem: main estimate 2}. Some of the statements in Sections \ref{sec:Gaussian:lemma} and \ref{sec:further:preliminary} were already used in \cite{MMS22}; for the reader's convenience we provide the statements and refer to the proofs in \cite{MMS22}.  In Sections \ref{sec: norm} -- \ref{sec: energy K} we derive estimates for the norm of our trial states and for the different terms in the energy formula that was obtained in Section \ref{sec:total:energy}. A heuristic sketch of this part of the proof is given in Section \ref{Sec: proof guide}. Some of the terms can be estimated similarly as the corresponding ones in \cite{MMS22}; for these terms some details will be skipped. The main contributions, however, require substantial modifications and thus we present them in more detail. The results from Sections \ref{sec: norm} -- \ref{sec: energy K} are combined in Section \ref{Sec: concluding the proof} to complete the proof of Proposition \ref{theorem: main estimate 2}. All remaining proofs of auxiliary lemmas are postponed until Section \ref{Sec: Remaining Proofs}.

Throughout the remainder we will abbreviate constants by the letter $C$ (and $c$) and write $C(\tau)$ or $C_\tau$ whenever we want to specify that it depends on a parameter $\tau$. As usual, the value of a constant may change from one line to the next.

\section{Proof of Proposition \ref{theorem: main estimate 2}\label{sec:proof:main:proposition}}

We recall the definition of the field operators
\begin{align}
\phi(f) = a^\dagger(f) + a(f) , \quad \pi(f) = \phi(i f )
\end{align}
%the standard estimates for creation and annihilation operators,
%\begin{align}\label{eq: bounds for a and a*}
%\sno a(f) \xi \sno_{\mathcal F} \le \sno f\sno_{2} \, \sno \mathbb N^{1/2} \xi \sno_{\mathcal F} , \quad  \sno a^\dagger(f) \xi \sno_{\mathcal F} \le \sno f\sno_{2}\, \sno (\mathbb N +1)^{1/2} \xi \sno_{\mathcal F}.
%\end{align}
and the Weyl operator
\begin{align}\label{eq: BCH for Weyl}
W(f) \,  =\, e^{-i\pi(f) } \,  =\, e^{a^\dagger(f) - a(f)} \,  =\, e^{a^\dagger(f)} e^{-a(f)} e^{-\frac{1}{2} \sno f\sno^2_\2 }.
\end{align}
The Weyl operator is unitary and satisfies
\begin{align}\label{eq: Weyl identities}
W^\dagger(f) \,  =\,  W(-f) , \quad W(f) W(g) \,  = \,    W(f + g ) e^{ i \im \langle g | f \rangle_\2 }.
\end{align}
For later use we recall that the Weyl operator shifts the creation and annihilation operators by complex numbers,
\begin{align}
W(g)^\dagger a^\dagger (f) W(g) \,  =\, a^\dagger(f) + \langle g| f\rangle_\2 , \quad W(g)^\dagger a (f) W(g) \,  =\, a(f) + \overline{\langle g| f\rangle_\2},
\end{align}
and thus
\begin{subequations}
\begin{align}
W(g)^\dagger \phi(f) W(g) & \,  =\,  \phi(f) +  2 \re \lsp f | g\rsp_\2  \label{eq: W phi W},\\[0.5mm]
W(g)^\dagger \mathbb N W (g) & \,  =\, \mathbb N +  \phi(g) +  \sno g \sno^2_\2, \label{eq: W N W}\\[0.5mm]
W(g)^\dagger P_f W(g) & \,  =\, P_f - a^\dagger( i \nabla  g) - a( i \nabla  g) -  \lsp g | i \nabla   g \rsp_\2 .
\end{align}
\end{subequations}

\subsection{The total energy\label{sec:total:energy}}

The proof of Proposition \ref{theorem: main estimate 2} starts with a convenient formula for the expectation value of the fiber Hamiltonian $H_\alpha(P)$ in a state $\Psi_P = \mathcal S_P \Gamma$ for general $\Gamma \in \mathcal F$. The precise statement requires some more notation. 

We introduce the $y$-dependent $L^2(\mathbb R^3)$-function
\begin{align}\label{eq:def of w}
w_{P,y} \,  =\, (1-e^{-y\nabla})\varphi_P 
\end{align}
with $\varphi_P$ defined in \eqref{eq: def of varphi_P}, and the $y$-dependent Fock space operator
\begin{align}\label{eq: definition AP}
A_{P,y} \,  =\, i P_f y + i g_{P}(y)  , \quad g_{P}(y) \,  =\,  - \frac{2}{M^{\rm LP}} \int_0^1 \D s\, \langle \varphi | e^{-sy  \nabla} (y \nabla)^3 (P \nabla) \varphi \rangle_\2 .
\end{align}
Since $g_{P}(y)$ is real-valued, it satisfies $(A_{P,y})^\dagger = - A_{P,y}$. 

We further introduce the Hamiltonian (acting on $L^2(\mathbb R^3) \otimes \mathcal F$)
\begin{subequations}
\begin{align}
\widetilde H_{\alpha,P} \,  & = \, h^{\rm Pek} + \alpha^{-2} \mathbb N + \alpha^{-1} \phi(h_{x} + \varphi_P )
\end{align} 
\end{subequations}
with $h^{\rm Pek}$ defined in \eqref{eq: hPek definition}, and the shift operator (that will exclusively act on the electron coordinate $x$)
\begin{align}\label{eq:shift:operator}
T_y:L^2(\mathbb R^3) \to L^2(\mathbb R^3), \quad (T_y f)(x) = f(x+y) \quad \text{for all}\quad f\in L^2(\mathbb R^3).
\end{align}

\begin{lemma} \label{lem: energy identity} For every $P\in \mathbb R^3$ and $\Gamma \in \mathcal V_K^{(n+1)}$ we have
\begin{align}
\lsp \mathcal S_P \Gamma |  H_\alpha(P)\mathcal S_P \Gamma \rsp_\Fock & \,  =\, \Big( e^{\rm Pek} + \frac{P^2}{2 \alpha^4 M^{\rm LP}} \Big)\, \sno \mathcal S_P \Gamma  \sno^2_\Fock  + \mathcal E^\Gamma +  \mathcal G^\Gamma + \mathcal K^\Gamma
\end{align}
where
\begin{subequations}
\begin{align}
\mathcal E^\Gamma & \,  =\, \int \D y\,  \lsp \mathscr G_{\Gamma}^{0} | \widetilde H_{\alpha,P} T_y  e^{A_{P,y}} W(\alpha w_{P,y}) \mathscr G_{\Gamma}^{ 0} \rsp_\h \label{eq: E}\\
\mathcal G^\Gamma & \,  =\, - \frac{2}{\alpha }  \int \D y\, \re \lsp \mathscr G_{\Gamma}^{ 0} | \widetilde H_{\alpha,P} T_y e^{A_{P,y}} W(\alpha w_{P,y})\mathscr G_{\Gamma}^{ 1}  \rsp_\h \label{eq: G}\\
\mathcal K^\Gamma & \,  =\, \frac{1}{\alpha^2} \int \D y\,  \lsp \mathscr G_{\Gamma}^{1} | \widetilde H_{\alpha,P} T_y e^{A_{P,y}} W(\alpha w_{P,y}) \mathscr G_{\Gamma}^{1} \rsp_\h \label{eq: K}
\end{align}
\end{subequations}
with $\mathscr G_\Gamma^{i}$ defined in \eqref{eq: alternative definition of G}.
\end{lemma}

We omit the proof since it follows verbatim the proof of \cite[Lemma 3.1]{MMS22}.

\subsection{The Gaussian lemma\label{sec:Gaussian:lemma}}

With \eqref{eq: def of Pi_i}, \eqref{eq:def of w} and $\Theta_K =  (H^{\rm Pek}_K)^{1/4} $ we introduce
\begin{subequations}
\begin{align}
w_{P,y}^{0} & \, : = \,   \Pi_0 w_{P,y} \in \text{Ker}H^{\rm Pek} \label{eq: def of tilde w:0} \\[1mm]
w_{P,y}^{ 1 } & \, : =\,  \Pi_1 w_{P,y} \in ( \text{Ker}H^{\rm Pek})^\perp %\label{eq: def of tilde w:1} 
\\[1mm]
\widetilde w_{P,y}^1&  \, : =\,   \Theta_K \re( w_{P,y}^{1} )  + i \Theta_K^{-1} \im (w_{P,y}^{1} ) \label{eq: def of tilde w:1}\\[1mm]
\widetilde w_{P,y}&  \, : =\,   w_{P,y}^0 + \widetilde w_{P,y}^1   \label{eq: def of tilde w}.
\end{align}
\end{subequations}
\begin{remark}\label{rem: symmetries_of_w} For later use, let us note some symmetries of the above functions. As explained in \cite[Remark 3.1]{MMS22} $(y,z)\mapsto \re(w_{P,y}^{i})(z)$, $i=0,1$, are even as functions on $\mathbb{R}^6$, while $\im (w_{P,y}^i)(z)$, $i=0,1$, are odd on the same space. Moreover, it follows by rotation invariance of $\varphi$ that $ \re(w^0_{P,-y})(z) = -\re(w^0_{P,y})(z)$ for all $y,z\in \mathbb R^3$. (Note that $\re(w_{P,y}^i) = w_{0,y}^i$).
\end{remark} 

The following lemma is proved in \cite[Lemma 3.3]{MMS22}.

\begin{lemma} \label{lem: bound for w1 and w0} Let $\lambda = \frac{1}{6}\sno \nabla \varphi \sno^2_\2$ and $K_0>0$ large enough. For every $c>0$ there exists a constant $C>0$ such that
\begin{subequations}
\begin{align}
 \sno w_{P,y}^1 \sno^2_\2  +  \sno \widetilde w_{P,y}^1 \sno^2_\2  & \, \le\,  C \big(\alpha^{-2} y^2 + y^4  \big) \label{eq: bound for w perp a} \\[1mm]
\big| \sno w_{P,y}^0 \sno^2_\2 - 2 \lambda y^2 \big|  & \, \le \, C \big( \alpha^{-2} y^2  + y^4 + y^6  \big) \label{eq: bound for w 0 a} 
\end{align}
\end{subequations}
for all $y\in \mathbb R^3$, $|P|/\alpha \le c$, $K \in (K_0,\infty] $ and $\alpha>0$.
\end{lemma}

For $0\le \delta <1$ and $\eta >0$ we introduce the weight function
\begin{align}\label{eq: definition of F}
n_{\delta , \eta  }(y) & \, =\,  \exp\bigg( - \frac{ \eta \alpha^{2(1-\delta)}  \sno \widetilde w_{P,y} \sno^2_\2}{2} \bigg) 
\end{align}
where we omit the dependence on $\alpha$, $P$ and $K$. By Remark \ref{rem: symmetries_of_w}, it follows that $n_{\delta , \eta }(y)$ is even as a function of $y$. In the limit of large $\alpha$ the dominant part of the weight function when integrated against suitably decaying functions comes from the term in the exponent that is quadratic in $y$. This is a crucial ingredient in our proofs and the content of the next lemma.

\begin{lemma}\label{lem: Gaussian lemma}
Let $\eta_0 > 0$, $c>0$, $\lambda = \frac{1}{6}\sno \nabla \varphi \sno^2_\2$ and $n_{\delta,\eta }$ defined in \eqref{eq: definition of F}. For every $n\in \mathbb N_0$ there exist constants $ d, C(n)>0$ such that
\begin{align}\label{eq: Gaussian estimate}
\int |y|^n  g(y)  \, \Big| n_{\delta,\eta }(y) - e^{ - \eta \lambda \alpha^{2(1-\delta)} y^2} \Big| \D y \, \le\,  C(n)  \frac{ \sno g \sno_\su  } {\alpha^{(4+n)(1-\delta) + \delta }}  + e^{- d   \alpha^{-2\delta+1} }\sno |\cdot |^n g \sno_\1
\end{align}
for all non-negative functions $g\in L^\infty(\mathbb R^3) \cap L^1(\mathbb R^3)$, $\eta\ge \eta_0$, $\delta \in [0,1) $, $|P|/\alpha \le c $ and $K,\alpha$ large enough.
\end{lemma}

%At first reading, one should think of $n=0$, $\delta = 0$, $\eta = 1$ and $g$ a suitable $\alpha$-independent non-negative function. In this case the integral involving the Gaussian is of order $\alpha^{-3}$ whereas the term on the right hand side is of order $\alpha^{-4}$ and thus contributing a subleading error. 
For the proof see \cite[Lemma 3.3]{MMS22}. As a direct consequence that will be useful to estimate error terms, we find

\begin{corollary}\label{cor: Gaussian for errors} Given the same assumptions as in Lemma \ref{lem: Gaussian lemma}, for every $n \in \mathbb N_0$ there exist constants $d, C(n) >0$ such that
\begin{align}
 \int  |y|^n g (y) n_{\delta,\eta } (y) \D y \, \le\, C(n) \frac{ \sno g \sno_\su }{  \alpha^{(3+n)(1-\delta)} } + e^{- d \alpha^{-2\delta+1} } \sno |\cdot|^n g \sno_\1
\end{align}
for all non-negative functions $g\in L^\infty(\mathbb R^3) \cap L^1(\mathbb R^3)$, $\eta \ge \eta_0 $, $\delta \in [0,1) $, $|P|/\alpha \le c$ and $K,\alpha$ large enough.
\end{corollary}

\begin{proof}[Proof of Corollary \ref{cor: Gaussian for errors}] Since
\begin{align}
\int \D y\, |y|^n  e^{ - \eta \lambda \alpha^{2(1-\delta)} y^2} \, =\,  ( \eta \lambda \alpha^{2(1-\delta)} )^{- \frac{3+n}{2}} \int \D y\, |y|^n e^{-y^2} \, =\,  C(n) \alpha^{-(3+n)(1-\delta)},
\end{align}
the statement follows immediately from Lemma \ref{lem: Gaussian lemma}.
\end{proof}

\subsection{Further preliminaries\label{sec:further:preliminary}}

In this section we summarize further helpful results. The proofs of Lemmas \ref{lemma: props_peks}--\ref{lem: bounds for the number operator}, Corollary \ref{cor: X h Y bounds} and Lemma \ref{lem: exp N bound} can be found in \cite[Section 3]{MMS22}.

\subsubsection{Estimates involving the Pekar minimizers}

\begin{lemma}\label{lemma: props_peks}
Let $\psi>0$ be the unique rotation invariant minimizer of the Pekar functional \eqref{eq: electronic pekar functional}, and let
\begin{align}\label{eq: definition of H}
  H(x)\, :=\, \lsp \psi|T_x\psi\rsp_\2 \, =\, (\psi\ast \psi)(x).
\end{align} 
We have that $\psi$, $|\nabla \psi|$ and $H$ are $L^p(\mathbb{R}^3,(1+|x|^n)\D x)$ functions for all $1\leq p \leq \infty $ and all $n\geq 0$. Moreover, there exists a constant $C>0$ such that for all $x$ we have 
\begin{align}\label{eq: quadratic_bd_on_H}
  |H(x)-1|\, \leq\, Cx^2.
\end{align}
\end{lemma}
The next lemma contains bounds for the potential $V^\varphi$ and the resolvent $R$ introduced in \eqref{eq: def of effective potential}, \eqref{eq: optimal phonon mode} and \eqref{eq: def of resolvent}.

\begin{lemma}\label{lem: bound for R} There is a constant $C>0$ such that
\begin{align}\label{eq: bound potential well}
(V^\varphi)^2 \, \le\, C (1-\Delta), \quad  \pm V^\varphi \, \le\, \frac{1}{2} (-\Delta) + C \quad \text{and} \quad \sno \nabla R^{1/2} \sno_{\op} \le C .
\end{align}
\end{lemma}
\begin{comment}
\begin{proof}
For the proof of the first two inequalities, we refer to \cite[Lemma III.2]{LeopoldRSS2019}. The bound for the resolvent is obtained through
\begin{align}
0\, \le\,   R^\frac{1}{2} (-\Delta ) R^\frac{1}{2} & \, \le\,  R^\frac{1}{2} h^{\rm Pek}R^\frac{1}{2} - R^\frac{1}{2} (V^\varphi -\lambda^{\rm Pek} ) R^\frac{1}{2} \, \le\,   C R +  \frac{1}{2}R^\frac{1}{2} (-\Delta ) R^\frac{1}{2},
\end{align}
where we made use of the second inequality in \eqref{eq: bound potential well}.
\end{proof}
\end{comment}

\subsubsection{The commutator method}

In our proof we are faced with bounding field operators like $\phi(h_x)$. The standard estimates for creation and annihilation operators,
\begin{align}\label{eq: standard estimates for a and a*}
\sno a(f) \Psi \sno_\h \, \le\, \sno f \sno_\2  \sno \mathbb N^{1/2} \Psi \sno_\h, \ \sno a^\dagger(f) \Psi \sno_\h \, \le\, \sno f \sno_\2 \sno (\mathbb N+1)^{1/2} \Psi \sno_\h, \ \Psi \in  \mathscr H,
\end{align}
are not sufficient since $h_{0}(y)$ is not square-integrable. For this purpose we shall use the following lemma, which is a version of the commutator method by Lieb and Yamazaki \cite{Lieb1958}.

\begin{lemma}\label{lem: LY CM} 
Let $h_{K,\cdot}$ for $K\in (1,\infty]$ as defined in \eqref{def: cut off coupling function}, let $A$ denote a bounded operator in $L^2(\mathbb R^3)$ \textnormal{(}acting on the field variable\textnormal{)} and $a^\bullet \in \{ a , a^\dagger \}$. Further let $X,Y$ be bounded symmetric operators in $L^2(\mathbb R^3)$ \textnormal{(}acting on the electron variable\textnormal{)} that satisfy $D_0 := \sno X \sno_{\op} \sno Y \sno_{\op} + \sno \nabla X \sno_{\op}\sno Y \sno_{\op} + \sno X \sno_{\op} \sno \nabla Y \sno_{\op}< \infty $. There exists a constant $C>0$ such that
\begin{subequations}
\begin{align}
\sno X a^\bullet(A h_{K,\cdot + y }) Y \Psi \sno_\h & \, \le\, C D_0 \sno (\mathbb N+1)^{1/2} \Psi \sno_\h \\ \sno X a^\bullet ( A h_{\Lambda ,\cdot+y} - A h_{K,\cdot+y} )  Y \Psi \sno_\h & \, \le\, \frac{C D_0}{\sqrt {K}} \sno (\mathbb N+1)^{1/2} \Psi \sno_\h
\end{align}
\end{subequations}
for all $\Psi\in \mathscr H$, $y\in \mathbb R^3$ and $1 < K < \Lambda \le \infty$.
\end{lemma}
\begin{remark} \label{eq: remark Ah} Note that $A h_{K,\cdot+y}= T_y (Ah_{K,\cdot}) $ and in case that $A$ has an integral kernel,
\begin{align}
(A h_{K,x})(z) = \int \D u\, A(z,u)h_{K,x}(u).
\end{align}  
\end{remark}
A simple but useful corollary is given by
\begin{corollary}\label{cor: X h Y bounds}
Under the same conditions as in Lemma \ref{lem: LY CM}, with the additional assumption that $Y$ is a rank-one operator, there exists a constant $C>0$ such that
\begin{subequations}
\begin{align}
\int \D z\, \sno X ( A h_{K,\cdot+y})(z) Y  \sno^2_{\op} & \, \le \, C D_0^2 \label{eq: X h Y bound}\\
\int \D z\, \sno X  \big( (A h_{K,\cdot+y})(z) -  (A h_{\Lambda,\cdot + y })(z) \big) Y  \sno_{\op}^2 & \,  \le \, \frac{C D_0^2}{K}  \label{eq: X h Y difference bound}
\end{align}
for all $y\in \mathbb R^3$ and $1 < K < \Lambda \le \infty$.
\end{subequations}
\end{corollary}
\begin{comment}
\begin{proof} Since $Y$ has rank one, we can use
\begin{align}
\int \D z\, \sno X ( A h_{K,\cdot+y})(z) w \sno^2_\2 & \, = \, \sno X a^\dagger(Ah_{K,\cdot+y}) w\otimes \Omega\sno^2_\h ,
\end{align}
for any $w\in L^2(\mathbb R^3)$, and similarly for \eqref{eq: X h Y difference bound}, and apply Lemma \ref{lem: LY CM}.
\end{proof}
\end{comment}

\subsubsection{Transformation properties of $\mathbb U_K$ \label{sec: transf prop UK}}

The next lemma collects relations for the Bogoliubov transformation $\mathbb U_K$ defined in \eqref{eq: def of U}. Its proof follows directly from this definition and the fact that $\Theta_K= (H^{\rm Pek}_K)^{1/4}$ is real-valued.

\begin{lemma} \label{lem: U W U transformation}
Let $f \in L^2(\mathbb R^3)$, $f^0 = \Pi_0 f$, $f^1 = \Pi_1 f$ with $\Pi_i$ defined in \eqref{eq: def of Pi_i} and set
\begin{subequations}
\begin{align}
 \underline f \, & : =\,   f^0 + \Theta^{-1}_{K} \re (f^1) + i \Theta_K \im(f^1) \label{eq: underline notation} \\[1mm]
 \widetilde {f} \, & : =\,  f^0  + \Theta^{\tc}_K \re(f^1) + i \Theta_K^{-1} \im (f^1 ) . \label{eq:tilde notation}
\end{align}
\end{subequations}
The unitary operator $\mathbb U_K$ defined in \eqref{eq: def of U} satisfies the relations
\begin{subequations}
\begin{align}
\mathbb U_K a(f) \mathbb U_K^\dagger & \, =\,  a( f^0 ) + a (A_K f^1 ) + a^\dagger (B_K \overline{  f^1 } ) \label{eq: U again} \\[1.5mm]
\mathbb U ^\dagger_K a(f) \mathbb U^{\tc}_K & \, =\,  a( f^0 ) + a (A_K f^1 ) - a^\dagger (B_K \overline{  f^1 } ) \label{eq: U reverse}\\[1.5mm]
\label{eq: transformation of phi}
\mathbb U^{\tc}_K \phi(f) \mathbb U^\dagger_K & \, =\,  \phi(\underline f) , \quad 
 \mathbb U^{\tc}_K \pi(f) \mathbb U^\dagger_K  \, =\,  \pi(\widetilde{f}) \\[1.5mm]
\mathbb U_K W(f) \mathbb U_K^\dagger & \, = \,  W( \widetilde f).\label{eq: transformation of Weyl}
\end{align}
\end{subequations}
\end{lemma}

Note that \eqref{eq:tilde notation} is consistent with the notation introduced in \eqref{eq: def of tilde w}.  
The following statements provide helpful bounds on the number operator when transformed with the Bogoliubov transformation.

\begin{lemma}\label{lem: bounds for the number operator} There exists a constant $b>0$ such that 
\begin{align}
\mathbb U^{\tc}_K (\mathbb N +1 )^n \mathbb U^{\dagger}_K \, \le \, b^n  n^n  (\mathbb N+1)^n ,\quad \, \mathbb U_K^\dagger (\mathbb N +1 )^n \mathbb U^{\tc}_K \, \le \, b^n  n^n  (\mathbb N+1)^n
\end{align}
for all $n\in \mathbb N $ and $K\in (K_0,\infty]$.
\end{lemma} 
In the next two statements we denote by $\mathbbm{1}(\mathbb N > c)$ (resp. $\mathbbm{1}(\mathbb N \le c)$) the orthogonal projection in $\mathcal F$ onto all states with phonon number larger than (resp. less or equal to) $c$.

\begin{corollary}\label{Cor: Upsilon estimates} Let $\mathcal V_K^{(n+1)} \subset \mathcal F_1$ as in \eqref{eq:Bogo:subspace} and $\mathfrak m \in \mathbb N$ defined by \eqref{eq:max:length}. Further set $\Gamma^{>} :=  \mathbbm{1}(\mathbb N > \alpha^\delta )  \Gamma$ for $\delta > 0$. There exist constants $b,C(\delta,j) > 0$ such that
\begin{subequations}
\begin{align}
\lsp \Gamma | (\mathbb N+1)^j \Gamma \rsp_\Fock & \, \le\,  b^j j^j ( \mathfrak m + 1 )^j \label{eq: bound for Y<} \\[1mm]
\lsp \Gamma^> | (\mathbb N+1)^j \Gamma^> \rsp_\Fock  &\,  \le\, C(\delta,j)\, \alpha^{- 20 }.\label{eq: exp bound for tail}
\end{align}
for all normalized $\Gamma \in \mathcal V_K^{(n+1)}$ and all $n,j\in \mathbb N_0$ and $K\in (K_0,\infty]$ with $K_0$ large enough.
\end{subequations}
\end{corollary} 
\begin{proof} The first bound follows from Lemma \ref{lem: bounds for the number operator} together with \eqref{eq:max:length}. The second one is obtained from
\begin{align}
\lsp \Gamma^> | (\mathbb N+1)^j \Gamma^> \rsp_\Fock \, & \le \, \sno \mathbb N^{k} (\mathbb N+1)^{j} \Gamma^{ >} \sno_\Fock  \,  \sno \mathbb N^{-k}\Gamma^{>} \sno_\Fock \notag\\[1mm]
& \,  \le\, \sno (\mathbb N + 1) ^{j+k} \Gamma  \sno_\Fock \, \alpha^{-k \delta }   \, \le \, (2  (j+k) b( \mathfrak m + 1 ))^{j+k} \alpha^{- k \delta }
\end{align}
with $k \ge 20 / \delta$.
\end{proof}

\begin{lemma}\label{lem: exp N bound} For $\delta>0$ and $\kappa = 1/ (16 e b \alpha^\delta)$ with $b > 0$ the constant from Lemma \ref{lem: bounds for the number operator}, the operator inequality
\begin{align}
\mathbbm{1}(\mathbb N \le 2 \alpha^\delta) \mathbb U_K^\dagger \exp( 2\kappa \mathbb N ) \mathbb U_K^{{\color{white}{\dagger}}} \mathbbm{1}(\mathbb N \le 2 \alpha^\delta) \, \le \, 2 
\end{align}
holds for all $K\ge K_0$ and $\alpha $ large enough.
\end{lemma}

The reason for introducing the cutoff $K$ in $\mathbb H_K$ is the following lemma, whose proof is given in Section \ref{Sec: Remaining Proofs}. The statement is a generalization of \cite[Lemma 3.13]{MMS22} to excited eigenstates of the Bogoliubov Hamiltonian.

\begin{lemma}\label{lem: bounds for P_f:new} For every $n\in \mathbb N_0$ there exist constants $C(n),K(n)>0$ such that
\begin{align}
\sno P_f \Gamma \sno_\Fock \, \le\, C(n) \sqrt  K  
\end{align}
%\begin{align}
%\sno P_f \Gamma \sno_\Fock \, \le\, C \sqrt  K \sqrt{ \frac{n+ 1}{1- %\lambda_\infty^{(n)}} } 
%\end{align}
for all normalized $\Gamma \in \mathcal V_K^{(n+1)}$ and $K\ge K(n)$.
\end{lemma}

\subsection{Replacing the Weyl operator by the weight function \label{sec:technical:lemmas}}

The next two lemmas will be useful in order to replace the Weyl operator $W(\alpha \widetilde w_{P,y})$ (when multiplied with different operators) by the weight function \eqref{eq: definition of F}. The proofs are given in Section \ref{Sec: Remaining Proofs}. For the precise statements, recall the notation $\mathscr G_\gamma^0 = \psi \otimes \gamma \in L^2(\mathbb R^3)\otimes \mathcal F$ and definitions \eqref{eq: def of tilde w}, \eqref{eq: definition of F} and \eqref{eq:Fockspace:m:1}.

\begin{lemma} \label{lem:X:gamma:W:chi} Let $\mathbf X,\mathbf Y: L^2(\mathbb R^3) \otimes\mathcal F_1 \to L^2(\mathbb R^3) \otimes \mathcal F_1 $ be densely defined operators whose Fock space components are at most quadratic in creation and annihilation operators.\footnote{Meaning $\mathbf X ( L^2(\mathbb R^3) \otimes \mathcal F^{(n)}_1 ) \subset L^2(\mathbb R^3) \otimes \bigoplus_{m=\max\{0,n-2\}}^{n+2} \mathcal F_1^{(m)}$ $\forall$ $n\in \mathbb N_0$, with $\mathcal F_1^{(m)} = \mathcal F_1 \cap \mathcal F^{(m)}$, cf. \eqref{eq:Fock:space:m:particles}.} Moreover let $\mathfrak p_\alpha(y)  =  1 + (\alpha |y|)^{4\mathfrak m + 8} $  with $\mathfrak m$ defined in \eqref{eq:max:length}. There exists a constant $C>0$ such that
\begin{align}
& \big| \lsp  \mathscr G_\gamma^0 | \mathbf X W(\alpha \widetilde w_{P,y} ) \mathbf Y \mathscr G_\xi^0 \rsp_\h  -  \lsp   \mathscr G_\gamma^0 |\mathbf X \mathbf Y \mathscr G_\xi^0 \rsp_\h n_{0,1}(y) \big| \notag\\[1mm]
& \hspace{3.5cm} \le C \alpha^{-1} \sno \mathbf X^\dagger  \mathscr G^0_\gamma \sno_\h  \sno \mathbf Y \mathscr G_\xi^0  \sno_\h \, \mathfrak p_\alpha(y)\,    n_{0,1}(y)  \quad \forall y\in \mathbb R^3
\end{align}
for all such operators $\mathbf X,\mathbf Y$ and all $\gamma ,\xi \in  \mathcal F_1^{(\le \mathfrak m)}$, $|P|  \le \sqrt{2M^{\rm LP}} \alpha $, $K\ge K_0$ and $\alpha \ge 1 $. 
\end{lemma}

\begin{lemma}\label{lem:error:terms:e-kappaN} Let $c,\delta>0$ and $\kappa = 1/ (16 e b \alpha^\delta)$ with $b > 0$ the constant from Lemma \ref{lem: bounds for the number operator}. Moreover let $ \mathfrak p_\alpha (y) = 1+  ( \alpha |y|) ^{4\mathfrak m + 8  }$ with $\mathfrak m$ defined in \eqref{eq:max:length}. There exist constants $\eta,C>0$ and $\alpha_0\ge 1$  such that
\begin{align}
\sno e^{-\kappa \mathbb N } W(\alpha \widetilde w_{P,y})  \gamma \sno_\Fock \le C \,  \mathfrak p_\alpha (y)\,  n_{\delta,\eta}(y)\quad \forall y \in \mathbb R^3
\end{align}
for all normalized $\gamma \in \mathcal F^{(\le \mathfrak m)}$, $|P| \le \sqrt{2 M^{\rm LP}} \alpha $ and $\alpha \ge \alpha_0$.
\end{lemma}

\subsection{Sketch of the remaining part of the proof \label{Sec: proof guide}}

With Lemma \ref{lem: energy identity}, the task is to show that $(\mathcal E^{\Gamma} + \mathcal G^{\Gamma} + \mathcal K^{\Gamma})/\sno \mathcal S_P \Gamma \sno_\Fock^2$ for $\Gamma \in \mathcal V_K^{(n+1)}$ coincides, up to small errors, with the energy contribution of order $\alpha^{-2}$ in \eqref{eq: main estimate 2}. Since this part of the proof is somewhat technical,  we explain in this section the heuristic ideas behind the different steps. The main point is that the integrands in \eqref{eq: E} -- \eqref{eq: K} are all localized around $y=0$ at the length scale of order $ \alpha^{-1}$. In a first step, this allows us to replace $e^{A_{P,y}}$ by the identity, and keeping only the terms that will contribute to the energy of order $\alpha^{-2}$, this gives
\begin{subequations}
\begin{align}
\mathcal E^\Gamma & \, =\, \frac{1}{\alpha^2} \int \D y\, \lsp \mathscr G_\Gamma^0 | \mathbb N_1 T_y W(\alpha  w_{P,y}) \mathscr G_\Gamma^0  \rsp_\h     \label{eq: heuristic E1} \\
& \quad  + \,  \frac{1}{\alpha}\int \D y\,  \lsp \mathscr G_\Gamma^0 |  \phi( h_\cdot + \varphi ) T_y W(\alpha  w_{P,y})  \mathscr G_\Gamma^0 \rsp_\h  \ + \ \text{Errors} \label{eq: heuristic E2} \\[1mm] 
\mathcal G^\Gamma & \, = \, -  \frac{2}{\alpha^2} \int \D y \re \lsp \mathscr G_\Gamma^0 | \phi(h_\cdot^1 )T_y  W(\alpha  w_{P,y}) \mathscr G_\Gamma^1 \rsp_\h    \ + \ \text{Errors} \label{eq: heuristic G} \\
 \mathcal K^\Gamma & \, = \, \frac{1}{\alpha^2} \int \D y\, \lsp \mathscr G_\Gamma^1   |h^{\rm Pek} T_y W(\alpha \widetilde w_{P,y})  \mathscr G_\Gamma^1 \rsp_\h  \ + \ \text{Errors} \label{eq: heuristic K}.
\end{align}
\end{subequations}
In the first, third and fourth term, we use the transformation property \eqref{eq: def of U}, e.g. 
\begin{align}
\lsp \mathscr G_\Gamma^0 | \mathbb N_1 T_y W(\alpha  w_{P,y}) \mathscr G_\Gamma^0  \rsp_\h  = \lsp \mathscr G_\gamma^0 | \mathbb U_K \mathbb N_1 \mathbb U_K^\dagger T_y W(\alpha   \widetilde w_{P,y}) \mathscr G_\gamma^0 \rsp ,
\end{align}
where we introduced $\gamma = \mathbb U_K \Gamma$. Since in the relevant regime, that is for $|y| = O(\alpha^{-1})$, we have $\widetilde w_{P,y} \approx y\nabla \varphi \in \text{Ran}(\Pi_0)$, the Weyl operator $W(\alpha \widetilde w_{P,y})$ can be effectively replaced by the Gaussian factor $e^{-\lambda \alpha^2 y^2}$ (recall that $\mathscr G_\gamma^0$ is the vacuum in $\mathcal F_0$). Technically we first apply Lemma \ref{lem:X:gamma:W:chi} to replace the Weyl operator by the weight function $n_{0,1}(y)$, and then Lemma \ref{lem: Gaussian lemma} to replace the weight function by $e^{- \lambda \alpha^2 y^2 }$. For that purpose, we need to show that the functions that multiply $n_{0,1}(y)$ are indeed bounded and integrable as needed in Lemma \ref{lem: Gaussian lemma} (e.g. in \eqref{eq: heuristic E1}, we have $g(y) = \langle \psi | T_y \psi \rangle_\2 \langle \gamma| \mathbb U_K \mathbb N_1 \mathbb U_K^\dagger \gamma\rangle_\Fock$, which satisfies these properties by Lemmas \ref{lemma: props_peks} and \ref{lem: bounds for the number operator}). In the leading order terms, we proceed by approximating $T_y \approx 1 $, $h_\cdot \approx h_{K,\cdot}$ and $u_\alpha \approx 1 $. Ignoring all errors resulting from these steps, this will lead to
\begin{align}
 \eqref{eq: heuristic E1} + \eqref{eq: heuristic G} + \eqref{eq: heuristic K} \approx \, \frac{1}{\alpha^2} \lsp \Gamma | \mathbb H_K \Gamma  \rsp_\Fock \int \D y \, e^{-\lambda \alpha^2 y^2 },
\end{align}
and similarly we show that the last factor is approximately given by the norm $\sno \mathcal S_P \Gamma \sno_\Fock^2 \approx (\tfrac{\pi}{\lambda \alpha^2})^{3/2}$. This explains the first part of the energy of order $\alpha^{-2}$ in \eqref{eq: main estimate 2}.

To compute \eqref{eq: heuristic E2}, we use the CCR, the fact that $\mathscr G_\Gamma^0$ coincides with the vacuum in $\mathcal F_0$, together with $(y,z)\mapsto \widetilde w_{P,y}^i(z)$ being symmetric and $w_{P,y}^0 \approx y\nabla \varphi$ for $|y|=O(\alpha^{-1})$. A careful analysis will show that
\begin{align}
\eqref{eq: heuristic E2} \approx  \int \D y\,  \lsp \psi \otimes \gamma  |  \langle h_\cdot | (y\nabla \varphi ) \rangle_\2 (y \nabla) W(\alpha \widetilde w_{P,y})  \psi \otimes \gamma \rsp_\h,
\end{align}
where we can replace again the Weyl operator by $e^{-\lambda \alpha^2 y^2 }$, and using integration by parts and \eqref{eq: optimal phonon mode}, one finds that
\begin{align}
  \eqref{eq: heuristic E2}  & \, \approx \, -\frac{1}{2}\int \D y\,  e^{-\lambda \alpha^2 y^2 } \sno y\nabla \varphi \sno^2_\2=-\frac{3}{2\alpha^2} \int \D y\, e^{-\lambda \alpha^2 y^2 }.\label{eq: -3/2 contribution}
\end{align}
This explains the energy contribution $-\tfrac{3}{2\alpha^2}$ in \eqref{eq: main estimate 2}.

\subsection{Norm of the trial state\label{sec: norm}}

In this section we derive an estimate for the inner product between two trial states. This gives a bound on the norm and provides the almost orthogonality of different states that is used to prove Lemma \ref{lem:linear:independence}.

\begin{proposition}\label{prop:overlap} Let $\mathcal V_K^{(n+1)} $ be defined as in \eqref{eq:Bogo:subspace}. For all $\varepsilon >0$ and $n\in \mathbb N_0$ there exist constants $C(n,\varepsilon) >0$ and $\alpha_0\ge 1$ such that 
\begin{align}
\bigg| \lsp \mathcal S_P \Gamma |\mathcal S_P \Xi \rsp_\Fock - \lsp \Gamma |  \Xi \rsp_\Fock \bigg(\frac{\pi}{\lambda  \alpha^2}\bigg)^{3/2}  \bigg| & \, \le\,   C(n,\varepsilon)  \sqrt K \alpha^{-4 + \varepsilon} 
\end{align}
for all normalized $\Gamma,\Xi \in \mathcal V_K^{(n+1)}$, $|P| \le \sqrt{2M^{\rm LP}} \alpha $, $K\ge K_0$ and $\alpha \ge \alpha_0$.
\end{proposition}
\begin{proof} Let $\mathcal N :=  \lsp \mathcal S_P \Gamma |\mathcal S_P \Xi \rsp_\Fock$. Following the argument of \cite[Proof of Lemma 3.1]{MMS22}, one verifies that $\mathcal N   =  \mathcal N_{0} + \mathcal N_{1a} + \mathcal N_{1b}  + \mathcal N_2$ with
\begin{subequations}
\begin{align}
\mathcal N_0 & \, =\,   \int \D y\, \lsp \mathscr G_{\Gamma}^{0} \I  T_y e^{ A_{P,y}} W(\alpha w_{P,y}) \mathscr  G_{\Xi}^{0} \rsp_\h \label{eq: norm line 1}\\
\mathcal N_{1a} & \, =\,  - \frac{1}{\alpha}  \int \D y\,  \lsp \mathscr  G_{\Gamma }^{0}  \I  T_y e^{ A_{P,y}  } W(\alpha w_{P,y}) \mathscr G_{\Xi}^{1}  \rsp_\h   \label{eq: norm line 2}\\
\mathcal N_{1b} & \, =\,  - \frac{1}{\alpha}  \int \D y\,    \lsp  T_y e^{ A_{P,y}  } W(\alpha w_{P,y}) \mathscr G_{\Gamma}^{1}  \I \mathscr  G_{ \Xi }^{0}   \rsp_\h  \label{eq: norm line 2b}\\
\mathcal N_{2}  & \, =\,   \frac{1}{\alpha^2}  \int \D y\, \lsp \mathscr  G_{\Gamma}^{1}  \I T_y  e^{ A_{P,y}  } W(\alpha w_{P,y})\mathscr   G_{\Xi }^{1} \rsp_\h \label{eq: norm line 3}
\end{align}
\end{subequations}
with $A_{P,y}$ defined in $\eqref{eq: definition AP}$ and $T_y$ denoting the shift operator \eqref{eq:shift:operator}.\medskip

\noindent \underline{Term $\mathcal N_0$.} This part contains the contribution $\langle \Gamma | \Xi \rangle_\Fock (\frac{\pi}{\lambda  \alpha^2})^{3/2}$. With $H$ defined in \eqref{eq: definition of H}, let us write
\begin{align}
 \mathcal N_0 & \, = \, \int \D y\,  H(y) \lsp   \Gamma \I   W(\alpha w_{P,y})    \Xi \rsp_\Fock \notag  \\
&\quad \, +\,  \int \D y\, H(y)  \lsp \Gamma \I (e^{A_{P,y}}-1) W(\alpha w_{P,y})  \Xi \rsp_\Fock = \mathcal N_{01} + \mathcal N_{02}.
\end{align}
In the first term we insert $\Gamma  = \mathbb U_K^\dagger  \gamma $, $\Xi = \mathbb U_K^\dagger \xi$ with $\gamma,\xi \in \mathcal F^{(\le \mathfrak m)}_1$, and apply \eqref{eq: transformation of Weyl} to transform the Weyl operator with the Bogoliubov transformation. This gives
\begin{align}\label{eq: identity for tilde w transformation}
\mathbb U_K W( \alpha w_{P,y} ) \mathbb U_K^\dagger & \, = \,  W( \alpha \widetilde w_{P,y})
\end{align}
with $\widetilde w_{P,y}$ defined in \eqref{eq: def of tilde w}. Applying Lemma \ref{lem:X:gamma:W:chi} for $\mathbf X  = \mathbf Y  = 1$ leads to the bound
\begin{align}\label{eq:bound N01 new}
\bigg|\, \mathcal N_{01} -  \lsp \gamma| \xi
 \rsp_\Fock \int \D y\, H(y)   \, n_{0,1}(y)  \bigg| \le \, \alpha^{-1} \int \D y\, H(y)\, \mathfrak p_\alpha (y) \, n_{0,1}(y)  .
\end{align}
Since $\sno H\sno_\1 + \sno H\sno_\su \le C $, cf. Lemma \ref{lemma: props_peks}, we can apply Lemma \ref{lem: Gaussian lemma} in order to replace the weight function $n_{0,1}(y)$ in the term containing $\langle \gamma| \xi \rangle_\Fock$ by the Gaussian $e^{-\lambda \alpha^2 y^2}$. More precisely
\begin{align}
\bigg | \int \D y\, H(y) n_{0,1}(y) - \int \D y\, H(y) e^{- \lambda \alpha^2 y^2 } \bigg| \, \le\,  C\alpha^{-4}
\end{align}
for all $|P| \le c \alpha $ and $K,\alpha$ large enough. Next we use $| H(y)-1| \le C y^2$ to obtain
\begin{align}
\bigg| \int \D y\, H(y) n_{0,1}(y) - \bigg( \frac{\pi}{\lambda \alpha^2 }\bigg)^{3/2}\bigg|\, \le \, C\alpha^{-4}.
\end{align}
For the remainder term in \eqref{eq:bound N01 new}, we recall $\mathfrak p_\alpha (y) = 1 + (\alpha |y|)^{4\mathfrak m + 8 }  $ and apply  Corollary \ref{cor: Gaussian for errors}. This leads to
\begin{align}
\bigg| \, \mathcal N_{01} - \lsp \gamma| \xi \rsp_\Fock \bigg( \frac{\pi}{\lambda \alpha^2 }\bigg)^{3/2}	\bigg| \le C \alpha^{-4}.
\end{align}

To treat $\mathcal N_{02}$ it is convenient to decompose the state $ \Gamma$ into a part with bounded particle number and a remainder. To this end, we choose $\delta>0$ small (but fixed w.r.t. $\alpha$) and write
\begin{align} \label{eq: decomposition of Upsilon}
\Gamma  \, = \,  \Gamma^{<}  +  \Gamma^{>} \, = \,  \mathbbm{1}(\mathbb N \le \alpha^\delta )  \Gamma  +  \mathbbm{1}(\mathbb N > \alpha^\delta ) \Gamma .
\end{align}
Inserting this into $\mathcal N_{02}$ and using unitarity of $e^{A_{P,y}}$ and $\sno H\sno_\1\le C$, we can estimate 
\begin{align}
| \mathcal N_{02}|\, & \le\, \int \D y\, H(y)   \I \lsp  \Gamma^{ <} \I (e^{A_{P,y}} -1 ) W(\alpha w_{P,y}) \Xi \rsp_\Fock \I + C \sno \Gamma ^>\sno_\Fock.
\end{align}
By Corollary \ref{Cor: Upsilon estimates} for $j=0$, $\sno \Gamma ^>\sno_\Fock \le C_\delta \, \alpha^{-10}$. In the remaining expression we use $\Xi = \mathbb U_K^\dagger \xi$ and \eqref{eq: identity for tilde w transformation},
\begin{align}\label{eq: inserting identity}
\lsp  \Gamma^{ <} \I (e^{A_{P,y}} -1 ) W(\alpha w_{P,y})  \Xi \rsp_\Fock \, = \,  \lsp  \Gamma^< \I (e^{A_{P,y}}-1) \mathbb U_K^\dagger W(\alpha \widetilde w_{P,y}) \xi  \rsp_\Fock ,
\end{align}
and then insert the identity
\begin{align} \label{eq: number op identity}
\mathbbm{1} \, = \, e^{\kappa \mathbb N } e^{- \kappa \mathbb N } \quad \text{with }\quad \kappa \, = \,  \frac{1}{16 e b \alpha^\delta}
\end{align}
on the left of the Weyl operator (here $b>0$ is the constant from Lemma \ref{lem: bounds for the number operator}). With Cauchy--Schwarz, this leads to
\begin{align} \label{eq: bound for the norm error}
& \I \lsp  \Gamma^{ <} \I (e^{A_{P,y}} -1 ) W(\alpha w_{P,y}) \Xi \rsp_\Fock \I  \notag\\[1mm]
& \hspace{2.5cm} \le \, \sno e^{\kappa \mathbb N} \mathbb U_K (e^{-A_{P,y}}-1) \Gamma^< \sno_\Fock  \sno e^{-\kappa \mathbb N } W(\alpha \widetilde w_{P,y})  \xi \sno_\Fock.
\end{align}
In the second factor we can employ Lemma \ref{lem:error:terms:e-kappaN},
\begin{align}
\sno e^{-\kappa \mathbb N } W(\alpha \widetilde w_{P,y})  \xi \sno_\Fock \le C \,  \mathfrak p_\alpha (y)\, n_{\delta,\eta}(y)\quad \forall y \in \mathbb R^3
\end{align}
for some $\alpha$-independent $\eta > 0 $ and $\alpha$ large enough. To estimate the first factor in \eqref{eq: bound for the norm error}, we apply Lemma \ref{lem: exp N bound} (note that $(e^{A_{P,y}} - 1) \Gamma^<  \in \text{Ran}(\mathbbm{1} (\mathbb N\le  \alpha ^\delta ))$) to obtain
\begin{align}
\sno e^{\kappa \mathbb N} \mathbb U_K (e^{ - A_{P,y}}-1)  \Gamma^< \sno_\Fock \le \sqrt 2 \sno (e^{ - A_{P,y}}-1) \Gamma \sno_\Fock.
\end{align}
On the right side we proceed with the functional calculus for self-adjoint operators
\begin{align}
\sno (e^{- A_{P,y} }-1) \Gamma \sno_\Fock  \le \sno A_{P,y} \Gamma \sno_\Fock  & \le |y|\, \sno P_f \Gamma \sno_\Fock +  |g_{P}(y)| \le  C \big( C(n) \sqrt K |y|  +  \alpha|y|^3 \big),
\end{align}
where we applied Lemma \ref{lem: bounds for P_f:new} in the second step and also used
\begin{align}\label{eq: bound for g_P}
|g_{P}(y)|\, \le  \, C \alpha |y|^3,
\end{align}
which is inferred from \eqref{eq: definition AP} using $\sno \Delta \varphi\sno_\2< \infty$. Returning to \eqref{eq: bound for the norm error} we have shown that 
\begin{align}
|\mathcal N_{02} | \, \le\, C \int \D y \, H(y)  \big( C(n) \sqrt{K} |y| +  \alpha |y|^3 \big) \mathfrak p_\alpha(y)\, n_{\delta , \eta} (y) + C_\delta\, \alpha^{-10} ,
\end{align}
and hence we are in a position to apply Corollary \ref{cor: Gaussian for errors}. This implies for all $K,\alpha$ large
\begin{align}\label{eq: example error bound}
|\mathcal N_{02} | \, & \le \,  C (n) \sqrt K \alpha^{-4 +   (12 + 4 \mathfrak m) \delta}  + C_\delta\, \alpha^{-10},
\end{align}
where we used that  $\sno |\cdot |^n H\sno_\1 + \sno H\sno_\su \le C(n) $. Note that the largest error comes from the term $C(n) \sqrt K\int H(y) |y| (\alpha |y|)^{4\mathfrak m + 8} n_{\delta , \eta} (y) $, which explains the factor $\alpha^{-4 + (12 + 4 \mathfrak m)\delta}$.

\medskip
 
\noindent \underline{Terms $\mathcal  N_{1a}$ and $\mathcal N_{1b}$}. We start with $\mathcal N_{1a}$ by inserting \eqref{eq: alternative definition of G} for $\mathscr G_\Gamma^0$, $\mathscr G_\Xi^1$ in \eqref{eq: norm line 2}. Since the Weyl operator commutes with $u_\alpha$, $R$ and $\PP = |\psi \rangle \langle \psi | $, we can apply \eqref{eq: W phi W} to obtain 
\begin{align}\label{eq: identity for WG1}
W(\alpha  w_{P,y} ) \mathscr G_\Xi^1 \, =\,  u_\alpha R \big( \phi(h_{K,\cdot}^1) + 2 \alpha \langle h_{K,\cdot} | \re ( w_{P,y}^1 ) \rangle_\2 \big) \PP W(\alpha  w_{P,y}) \mathscr G_\Xi^0,
\end{align}
where we used that $h_{K,x}$ is real-valued. Note that $\langle h_{K,\cdot} | \re( w_{P,y}^1 ) \rangle_\2 $ is a $y$-dependent multiplication operator in the electron variable. With $( T_y e^{A_{P,y}})^\dagger =  T_{-y} e^{-A_{P,y}}$ and \eqref{eq: decomposition of Upsilon}, we can thus write
\begin{align}
\mathcal N_{1a}   & \, = \, - \frac{2}{\alpha} \int \D y \, \re \lsp R_{1,y} \psi \otimes \big( \Gamma^< + \Gamma^> \big) \I W(\alpha  w_{P,y})  \mathscr G_\Xi^0 \rsp_\h = \mathcal N_{1a}^< + \mathcal N_{1a}^>,\label{eq: N11 plus N12}
\end{align}
where we introduced the operator $R_{1,y} = R_{1,y}^{(1)} + R_{1,y}^{(2)} $ with
\begin{subequations}
\begin{align}\label{eq: R_1y^1}
R_{1,y}^{(1)} & \, = \, \PP \phi(h_{K,\cdot}^1) R u_\alpha T_{-y} \PP e^{-A_{P,y}} ,\\[1mm]
R_{1,y}^{(2)} & \, = \, 2\alpha  \PP \lsp h_{K,\cdot} | \re (w_{P,y}^1) \rsp_\2  R u_\alpha T_{-y} \PP e^{-A_{P,y}} .\label{eq: R_1y^2}
\end{align}
\end{subequations}
As shown in \cite[Eq. (3.85)]{MMS22} this operator satisfies
\begin{align}
\sno R_{1,y} \Psi \sno_\h  & \, \le\, C  \sno u_\alpha T_{-y} \PP  \sno_{\op} (1+ \alpha y^2) \sno   (\mathbb N+1)^{1/2} \Psi \sno_\h ,\quad  \Psi \in \mathscr H,
  \label{eq: bound for R_1}
\end{align}
and since $\psi(x)$ decays exponentially for large $|x|$, the function $f_{\alpha}(y) := \sno u_\alpha T_{-y} \PP \sno_{\op}   $ satisfies
\begin{align}\label{eq: bound for f alpha}
\sno  |\cdot |^n f_{\alpha} \sno_\1 \, \le \,  \int \D y\, |y|^n \bigg( \int \D x\, \psi(x+y)^2 u_\alpha(x)^2 \bigg)^{1/2} \, \le\,  C(n) \alpha^{3+n}\ \ \text{for all}\ n \in \mathbb N_0.
\end{align}
With this at hand we can estimate the part with the tail by invoking Corollary \ref{Cor: Upsilon estimates}
\begin{align}
| \mathcal N_{1a}^> | \, \le \, \frac{C}{\alpha}  \sno (\mathbb N+1)^{1/2} \Gamma^> \sno_\Fock   \int \D y\,  f_{\alpha} (y) (1+\alpha y^2) \, \le \, C_\delta\, \alpha^{-5} .
\end{align}
To estimate the first term in \eqref{eq: N11 plus N12}, we proceed similarly as in the bound for $\mathcal N_{02}$. We insert the identity \eqref{eq: number op identity}, apply Cauchy--Schwarz and employ $\Xi = \mathbb U_K^\dagger \xi$ together with Lemma \ref{lem:error:terms:e-kappaN}. This leads to
\begin{align}
\vert \mathcal N_{1a}^< \vert &\,  \le\,  \frac{1}{\alpha}  \int \D y \,
\sno e^{\kappa \mathbb N} \mathbb U^{\tc}_K ( e^{- A_{P,y}}  R_{1,y} \psi \otimes \Gamma^< )\sno_\h \, \sno e^{- \kappa \mathbb N } W(\alpha \widetilde w_{P,y}) \xi  \sno_{\Fock} \nonumber \\
& \, \le\,  \frac{1}{\alpha}  \int \D y \,  \sno e^{\kappa \mathbb N} \mathbb U^{\tc}_K ( e^{ - A_{P,y}} R_{1,y}  \psi \otimes  \Gamma^< ) \sno_\h \, \mathfrak p_\alpha(y)\, n_{\delta,\eta}(y) .
\end{align}
In the remaining norm we use the fact that $e^{- A_{P,y}} R_{1,y}$ changes the number of phonons at most by one, and thus we can apply Lemma \ref{lem: exp N bound} and \eqref{eq: bound for R_1} together with \eqref{eq: bound for Y<}, to get
\begin{align}
\sno e^{\kappa \mathbb N} \mathbb U^{\tc}_K  (e^{-A_{P,y}}  R_{1,y} \psi \otimes \Gamma^< ) \sno_\Fock \, \le \, \sqrt 2 \sno  R_{1,y} \psi \otimes \Gamma^< \sno_\Fock \,\le \, C f_{\alpha} (y) \big( 1+ \alpha y^2\big) .
\end{align}
With Corollary \ref{cor: Gaussian for errors}, \eqref{eq: bound for f alpha} and $\sno f_\alpha  \sno_\su \leq 1 $, this leads to the bound
\begin{align}
|\mathcal N_{1a}^<  | & \, \le \, \frac{C}{\alpha} \int \D y\,  f_{\alpha} (y) \big( 1 + \alpha y^2 \big) \, \mathfrak p_\alpha(y)\, n_{\delta,\eta}(y)
\,  \le\, C \alpha^{-4 + (11 + 4 \mathfrak m )\delta} . 
\end{align}
The term $\mathcal N_{1b}$ can be estimated in the same way as $\mathcal N_{1a}$, exchanging the roles of $\Gamma$ and $\Xi$.\medskip

\noindent \underline{Term $\mathcal  N_{2}$}. The strategy for this term is similar to the one for $\mathcal N_{1a}$. Proceeding as explained before \eqref{eq: N11 plus N12}, one obtains
\begin{align}
\mathcal  N_2& \, = \, \frac{1}{\alpha^2} \int \D y \, \lsp  R_{2,y} \psi \otimes \big( \Gamma^< + \Gamma^> \big) \I   W(\alpha w_{ P ,y} )  \mathscr G_\Xi^0  \rsp_\h \,  = \,  \mathcal  N_{2}^< + \mathcal  N_{2}^> \label{eq: Norm term N2}
\end{align}
with $R_{2,y} = R_{2,y}^{(1)} + R_{2,y}^{(2)}$,
\begin{subequations}
\begin{align}
 R_{2,y}^{(1)} &\,  =\,  \PP \phi(h^1_{K,\cdot}) R e^{-A_{P,y}} u_\alpha T_{-y} u_\alpha R \phi(h^1_{K,\cdot}) \PP , \\[1mm]
 R_{2,y}^{(2)} & \, = \, 2 \alpha \PP \langle h_{K,\cdot } | \re (w^1_{P,y})\rangle_\2 R e^{-A_{P,y}} u_\alpha T_{-y} u_\alpha R \phi(h^1_{K,\cdot}) \PP .
\end{align}
\end{subequations}
It follows in close analogy as for $R_{1,y}$ in \eqref{eq: R_1y^1}--\eqref{eq: R_1y^2} that
\begin{align}
\sno  R_{2,y}  \Psi \sno_\h & \, \le \, C \sno u_\alpha T_{-y} u_\alpha \sno_{\op} (1+\alpha y^2)  \sno (\mathbb N+1)  \Psi \sno_\h, \quad \Psi \in \mathscr H,
\end{align}
and since $\sno u_\alpha T_{-y} u_\alpha \sno_{\op} \le \mathbbm{1} (|y|\le 4\alpha)$, we can use Corollary \ref{Cor: Upsilon estimates} to estimate
\begin{align}
|\mathcal  N_{2}^> | \, \le \, \frac{C}{\alpha^2} \sno (\mathbb N+1) \Gamma^> \sno_{\Fock} \int \D y \, \mathbbm{1}(|y|\le 4 \alpha) ( 1 + \alpha y^2 ) \,  \le\,  C_\delta\, \alpha^{-6 }.
\end{align}
To bound the first term in \eqref{eq: Norm term N2} we proceed similarly as for $\mathcal N_{01}$, that is
\begin{align}
|\mathcal  N_{2}^< | & \le  \alpha^{-2} \int \D y\, \sno e^{\kappa \mathbb N } \mathbb U_K (R_{2,y} \psi \otimes \Gamma^< ) \sno_\Fock \, \sno e^{-\kappa \mathbb N } W(\alpha \widetilde w_{P,y}) \xi \sno_\Fock \notag \\
& \le  \frac{C}{\alpha^2} \int \D y\, \mathbbm{1} (|y|\le 4\alpha) (1 + \alpha y^2  ) \, \mathfrak p_\alpha(y)\, n_{ \delta, \eta }(y) \le C \alpha^{-5 + (12 + 4 \mathfrak m ) \delta},
\end{align}
where the last step follows again from Corollary \ref{cor: Gaussian for errors}.

Collecting all relevant estimates and choosing $\delta>0$ small enough completes the proof of the proposition.
\end{proof}

\subsection{Energy contribution $\mathcal E^\Gamma$}

In this section we prove an estimate for the energy contribution $\mathcal E^\Gamma$ defined in \eqref{eq: E}.

\begin{proposition} \label{prop: bound for E} Let $\mathcal V_K^{(n+1)}$ be defined as in \eqref{eq:Bogo:subspace} and $\mathbb N_1$ be defined by \eqref{eq: Bogoliubov Hamiltonian maintext}. For all $\varepsilon >0$ and $n\in \mathbb N_0$ there exist constants $C(n,\varepsilon) >0$ and $\alpha_0\ge 1$ such that 
\begin{align}
\bigg|\, \mathcal E^{\Gamma} - \frac{1}{\alpha^2}  \bigg( \lsp \Gamma | \mathbb N_1 \Gamma \rsp_\Fock - \frac{3}{2} \bigg)\, \sno \mathcal S_P \Gamma\sno_\Fock^2 \,  \bigg| \,  \le \,   C(n, \varepsilon )   \sqrt K \alpha^{-6 + \varepsilon} 
\end{align}
for all normalized $\Gamma \in \mathcal V_K^{(n+1)}$, $|P| \le \sqrt{2M^{\rm LP}}\alpha $, $K\ge K_0$ and $\alpha \ge \alpha_0$.
\end{proposition}

\begin{proof}
Since $\mathscr G_{\Gamma}^0 = \psi \otimes \Gamma$, $ h^{\rm Pek} \psi = 0$ and $\mathbb N \Gamma = \mathbb N_1 \Gamma $, one finds
\begin{align}
\mathcal E^\Gamma & \,  =\,   \int \D y\,  \lsp \mathscr G_{\Gamma}^{0} | \big(\alpha^{-2} \mathbb N_1 + \alpha^{-1} \phi(h_\cdot + \varphi_P) \big) T_y e^{A_{P,y}} W(\alpha w_{P,y})  | \mathscr G_{\Gamma}^{ 0} \rsp_\h  \, =\,  \mathcal E_1^{\Gamma} + \mathcal E_2^{\Gamma},
\end{align}
where both terms provide contributions to the energy of order $\alpha^{-2}$.\medskip

\noindent \underline{Term $\mathcal E_1^{\Gamma}$}. Recall that $H(y) = \langle \psi | T_y \psi \rangle_\2 $ and use this to write
\begin{align}
\mathcal E_1^\Gamma & \, =\,  \frac{1}{\alpha^2} \int \D y\, H(y)  \lsp \Gamma | \mathbb N_1  W(\alpha w_{P,y})  \Gamma \rsp_\Fock  \notag \\
& \quad + \frac{1}{\alpha^{2}} \int \D y\, H(y)   \lsp \Gamma | \mathbb N_1 (e^{A_{P,y}} -1 ) W(\alpha w_{P,y})  \Gamma \rsp_\Fock \,  =\,  \mathcal E_{11}^\Gamma + \mathcal E_{12}^\Gamma. \label{eq: E11 + E12}
\end{align}
Also recall $\Gamma = \mathbb U_K^\dagger \gamma$ for some $\gamma \in  \mathcal F^{(\mathfrak m)}_1$. With \eqref{eq: identity for tilde w transformation} and writing
\begin{align}
 \lsp \Gamma | \mathbb N_1  W(\alpha w_{P,y})  \Gamma \rsp_\Fock  =  \lsp   \gamma | \mathbb U_K \mathbb N_1 \mathbb U_K^\dagger  W(\alpha \widetilde w_{P,y})  \gamma \rsp_\Fock 
\end{align}
we can apply Lemma \ref{lem:X:gamma:W:chi} with $\mathbf X = \mathbb U_K \mathbb N_1 \mathbb U_K^\dagger$ (this is a quadratic operator $\mathcal F_1 \to \mathcal F_1$) and $\mathbf Y = 1$. Since $\sno \mathbb U_K \mathbb N_1 \mathbb U_K^\dagger \gamma\sno_\Fock \le C$ by Lemma \ref{lem: bounds for the number operator}, it follows that
\begin{align}
\bigg| \mathcal E_{11}^\Gamma  -  \frac{1}{\alpha^2}  \lsp \Gamma |  \mathbb N_1  \Gamma \rsp_\Fock   \int \D y\, H(y)\, n_{0,1}(y)  \bigg| \, \le \,  \frac{1}{\alpha^3}  \int \D y \, H(y) \, \mathfrak p_\alpha(y)\, n_{0,1}(y)   .
\end{align}
For the error term we find as a direct consequence of Corollary \ref{cor: Gaussian for errors} that it is bounded by $C\alpha^{-6}$. 
The second term on the l.h.s. we call $\mathcal E_{111}^{\Gamma}$ and write it as  
\begin{align}
\mathcal E_{111}^\Gamma & \, =\,  \frac{1}{\alpha^2} \lsp \Gamma   | \mathbb N_1  \Gamma  \rsp_\Fock  \int \D y\, H(y) \, e^{- \lambda \alpha^2 y^2} \notag\\
& \quad +  \frac{1}{\alpha^2} \lsp \Gamma  | \mathbb N_1 \Gamma \rsp_\Fock \int \D y\, H(y) \big( n_{0,1}(y) - e^{- \lambda \alpha^2 y^2}  \big)  \,  =\,  \mathcal E_{111}^{\Gamma,\rm lo} + \mathcal E_{111}^{\Gamma,\rm err} . \label{eq: add subtract gaussian for E111}
\end{align}
In $\mathcal E_{111}^{\Gamma,\rm lo}$ we use $| H(y) - 1 | \le C y^2$ and Corollary \ref{Cor: Upsilon estimates} to replace $H(y)$ by unity at the cost of an error of order $\alpha^{-7}$. In the term where $H(y)$ is replaced by unity, we perform the Gaussian integral and use Proposition \ref{prop:overlap} for  $ \sno \mathcal S_P \Gamma \sno_\Fock^2$, and again Corollary \ref{Cor: Upsilon estimates}. This leads to
\begin{align}\label{eq: bound for E111 lo}
\Big|\, \mathcal E_{111}^{\Gamma, \rm lo} -   \frac{1}{\alpha^2}  \lsp \Gamma | \mathbb N_1  \Gamma \rsp_\Fock \, \sno \mathcal S_P \Gamma\sno_\Fock^2  \Big| \, \le \,   C (n,\varepsilon)  \sqrt K \alpha^{-6+\varepsilon} .   
\end{align}
The error in \eqref{eq: add subtract gaussian for E111} is bounded with the help of Lemma \ref{lem: Gaussian lemma},
\begin{align}\label{eq: bound for E111 err}
| \mathcal E_{111}^{\Gamma, \rm err} |\,  \le\,  \frac{C}{\alpha^2} \int \D y\, H(y) | n_{0,1}(y) - e^{-\lambda \alpha^2 y^2 }| \,  \le \,  C \alpha^{-6} .
\end{align}

In order to bound $\mathcal E_{12}^\Gamma$ in \eqref{eq: E11 + E12}, we decompose $\Gamma = \Gamma^< + \Gamma^>$ for some $\delta>0$ as in \eqref{eq: decomposition of Upsilon}, and then follow similar steps as described below \eqref{eq: inserting identity}. This way we can estimate
%\begin{align}\label{A-Schwarzing}
%| \mathcal E_{12} | & \, \le \,  \frac{1}{\alpha^2} \int \D y\, H(y) \sno e^{\kappa \mathbb N } \mathbb U_K (e^{ - A_{P,y}} - 1)  \mathbb N_1 \Upsilon_K^< \sno_\Fock \, n_{\delta,\eta}(y) \notag \\
%& \hspace{1.5cm}  + \, \frac{2}{\alpha^2} \sno \mathbb N_1 \Upsilon_K^>\sno_\Fock \int \D y\, H(y) .
%\end{align}
\begin{equation}\label{A-Schwarzing}
| \mathcal E_{12}^\Gamma |   \le   \frac{1}{\alpha^2} \int \D y\, H(y) \sno e^{\kappa \mathbb N } \mathbb U_K (e^{ - A_{P,y}} - 1)  \mathbb N_1 \Gamma^< \sno_\Fock \, \mathfrak p_\alpha(y)\, n_{\delta,\eta}(y) 
  + \, \frac{2}{\alpha^2} \sno \mathbb N_1 \Gamma^>\sno_\Fock \int \D y\, H(y) .
\end{equation}
While the second term is bounded via \eqref{eq: exp bound for tail} by $C_\delta\, \alpha^{-12}$, in the first term we apply Lemma \ref{lem: exp N bound} and use the functional calculus for self-adjoint operators,
\begin{align}
 \sno e^{\kappa \mathbb N } \mathbb U_K (e^{ - A _{P,y} } - 1)  \mathbb N_1 \Gamma^< \sno_\Fock & \,   \le \, \sqrt 2 \sno  (P_fy + g_{P}(y) )  \mathbb N_1 \Gamma^<  \sno_\Fock.
\end{align}
Since $P_f$ changes the number of phonons in $\mathcal F_1$ at most by one, we can proceed by
\begin{align}
\hspace{-1mm}\sno  (P_fy + g_{P}(y) )  \mathbb N_1 \Gamma^<  \sno_\Fock & \le  (\alpha^{\delta}+1) \sno (P_f y  + g_{P}(y)) \Gamma^< \sno_\Fock \le   C \alpha^\delta \big( |y| \sno P_f \Gamma\sno_\Fock +   \alpha |y|^3 \big)\notag\\[1mm]
& \le C \alpha^\delta \big(  C(n) \sqrt K  \, |y|  +   \alpha |y|^3 \big),
\end{align}
where we used $1\le \alpha^\delta$ and \eqref{eq: bound for g_P} in the second step and Lemma \ref{lem: bounds for P_f:new} in the third one. We conclude via Corollary \ref{cor: Gaussian for errors} that
\begin{align}\label{eq: bound for E12}
| \mathcal E_{12}^\Gamma |& \, \le \, \frac{C}{\alpha^2} \int \D y\, H(y)\big(  C(n) \sqrt K  |y| + \alpha |y|^3 \big) \mathfrak p_\alpha(y)\, n_{\delta,\eta}(y) + C_\delta\, \alpha^{-12} \notag\\
& \, \le \, C_\delta \, C(n)  \sqrt K  \, \alpha^{- 6 + (12 + 4 \mathfrak m) \delta} .
\end{align}

\noindent \underline{Term $\mathcal E_2^{\Gamma}$}. 
Here we start with
\begin{align}
\mathcal E_2^\Gamma & \, = \, \alpha^{-1} \int \D y\,  \lsp \mathscr G_\Gamma^0 | L_{1,y} W(\alpha w_{P,y}) \mathscr G_\Gamma^0 \rsp_\Fock \notag \\
& \quad + \alpha^{-1} \int \D y\, \lsp \mathscr G_\Gamma^0 | L_{1,y}  ( e^{A_{P,y}}  -1 ) W(\alpha w_{P,y})   \mathscr G_\Gamma^0 \rsp_\Fock \, = \, \mathcal E^\Gamma_{21} + \mathcal E_{22}^\Gamma ,
\end{align}
where
\begin{align}
 L_{1,y}  \, =\,  \lsp \psi| \phi(h_\cdot + \varphi_P) T_y \psi \rsp_\2 = \phi(l_y)+\pi(j_y)
\end{align}
with 
\begin{align}\label{eq:def:ly:jy}
  l_y \,  = \,  H(y)\varphi+\lsp \psi|h_{\cdot}T_y \psi\rsp_\2, \quad j_ y \,  = \,  H(y)\xi_P ,
\end{align} 
and $\xi_P$ defined in \eqref{eq: def of varphi_P}. The proof of the following lemma is given in \cite[Lemma 3.17]{MMS22}.
\begin{lemma}\label{lemma: props_lx}
For $k=0,1$ and for all $n\in \mathbb{N}_0$,
\begin{align}\label{eq: props_lx}
  \sup_y \|\nabla^k l_y \|_\2 \, <\, \infty, \quad \int |y |^n\|\nabla^k l_y \|_\2 \, \D y \, <\, \infty.
\end{align}  
\end{lemma}
Note that, by Lemma \ref{lemma: props_peks}, $j_y$ clearly has these properties as well. In \cite[Eqs. (3.125) and (3.126)]{MMS22} it was also shown that
\begin{align}\label{eq:lys}
\sno w_{0,y}-y\nabla\varphi \sno_\2 +\sno l_y^1 \sno_\2+ \sno l_y^0 + \frac{1}{2}y\nabla\varphi \sno_\2 & \leq Cy^2.
\end{align}
We proceed by writing $\mathcal E^\Gamma_{21} =  \mathcal{E}_{21}^{\Gamma,0} +  \mathcal{E}_{21}^{\Gamma,P}$ with
\begin{subequations}
\begin{align}
  \mathcal{E}_{21}^{\Gamma,0} & \, =\, \alpha^{-1}\int \D y\, \lsp \gamma |\phi(l_y^0)W(\alpha \widetilde w_{P,y}) \mathscr \gamma \rsp_\Fock + \alpha^{-1}\int \D y\, \lsp \gamma |\phi(\underline{l_y^1})W(\alpha \widetilde w_{P,y}) \mathscr \gamma \rsp_\Fock \label{eq:E21}\\
  \mathcal{E}^{\Gamma,P}_{21} &  \, =\, \alpha^{-1}\int \D y \, \lsp \gamma |\pi(j_y)W(\alpha \widetilde w_{P,y}) \gamma \rsp_\Fock ,\label{eq:E21:P}
  \end{align} 
\end{subequations}
where we used $\Gamma = \mathbb U_K^\dagger \gamma$ and applied Lemma \ref{lem: U W U transformation} together with $j_y\in \text{Ran}(\Pi_0)$, cf. \eqref{eq: def of varphi_P} and \eqref{eq:def:ly:jy}. The three terms are estimated separately. 

For the first term, and also for later use, let us note the following identity. For any $x\mapsto g^0_x$ with $g_x^0 \in \text{Ran}(\Pi_0)$ for all $x\in \mathbb R^3$ and $\Psi, \Phi \in L^2(\mathbb R^3) \otimes \mathcal F_1$, we have
\begin{align}
 \tfrac{1}{\alpha}\lsp \Psi |\phi(g^0_\cdot )W(\alpha \widetilde w_{P,y}) \Phi  \rsp_\h \, & = \, \tfrac{1}{\alpha} \lsp \Psi |a(g^0_\cdot) e^{a^\dagger(\alpha w_{P,y}^0) }W(\alpha \widetilde w_{P,y}^1) \Phi \rsp_\h e^{-\frac{1}{2} \sno \alpha  w_{P,y}^0\sno^2_\2}  \notag\\[1.5mm]
&  = \,  \lsp \Psi | \langle g^0_\cdot |  w_{P,y}^0   \rangle_\2 W(\alpha \widetilde w_{P,y}^1) \Phi  \rsp_\h e^{-\frac{1}{2} \sno \alpha  w_{P,y}^0\sno^2_\2} \notag\\[2.5mm]
&  = \lsp \Psi | \langle g^0_\cdot | w_{P,y}^0   \rangle_\2  W(\alpha \widetilde w_{P,y} ) \Phi  \rsp_\h, \label{eq:phi(g0):term}
\end{align}
where we used $ W(\alpha \widetilde w_{P,y}) =  W(\alpha  w_{P,y}^0 )  W(\alpha \widetilde w_{P,y}^1)$, $a(g^0_\cdot )\Psi = a(  w_{P,y}^0)\Phi = 0$ and the canonical commutation relations. Invoking this identity with $\Psi = \Phi = \mathscr G_\gamma^0$ and $g_x^0 = l_y^0$ ($x$-independent), we find $\langle \gamma |\phi(l_y^0)W(\alpha \widetilde w_{P,y}) \mathscr \gamma \rangle_\Fock = \langle l_y^0 | w_{P,y}^0   \rangle_\2  \langle \gamma|  W(\alpha \widetilde w_{P,y} ) \gamma \rangle_\Fock $, and thus we can apply Lemma \ref{lem:X:gamma:W:chi} for $\mathbf X  = \langle l_y^0 |   w_{P,y}^0   \rangle_\2$ and $\mathbf Y =1$ to bound the first term in \eqref{eq:E21} as
\begin{align}\label{eq:commutator:ly0:part}
& \left|\alpha^{-1}\int \D y\, \lsp \gamma | \phi(l_y^0) W(\alpha \widetilde{w}_{P,y}) \gamma \rsp_\Fock -\int \D y\, \lsp l_y^0|w_{P,y}^0\rsp_\2 n_{0,1}(y) \right| \notag\\
& \hspace{6cm} \leq  C\int \D y\, \I \lsp l_y^0|w_{P,y}^0\rsp_\2 \I\, \mathfrak{p}_{\alpha}(y)\, n_{0,1}(y) .
\end{align}
Note that $l_{-y}^i(-z)=l_y^i(z)$, $i=0,1$, which is easily verified for $l_y(z)$ and preserved by the projections $\Pi^i$ as their kernels are even functions on $\mathbb R^6$. As discussed in Remark \ref{rem: symmetries_of_w}, $n_{0,1}(y)$ is even, and similarly, 
%$ \Theta^{-2}_K\mathrm{Im}(w_{P,y}^1)$ and 
$\mathrm{Im}(w_{P,y}^0)(z)$ is an odd function on $\mathbb{R}^6$ since $(y,z)\mapsto \mathrm{Im}(w_{P,y})(z)$ is odd on this space. Hence we can conclude that
\begin{align}
\int \D y \, \lsp l_y^0 |\mathrm{Im}(w_{P,y}^0) \rsp_\2\,  n_{0,1}(y)   \, = \, 0.
\end{align}
Moreover, we use $\langle l_y^0|\re(w_{P,y}^0)\rangle_\2 = \langle l_y^0|  w_{0,y}^0   \rangle_\2$ since $\mathrm{Re}(w_{P,y})=w_{0,y}$. The function 
\begin{align}
   v(y)\, :=\, \lsp l^0_y|w_{0,y}^0\rsp_\2
\end{align}
satisfies $ v\in L^1\cap L^{\infty}$ since $y\mapsto \sno l_y \sno_\2$ does (see \cite[Section 3.6]{MMS22}) and $\sno w_{0,y}\sno_\2\le 2 \sno \varphi \sno_\2$. Using \eqref{eq:lys} we can thus estimate
\begin{align}\label{eq: bd on v_x}
\bigg| v(y) + \frac{1}{2} \sno y\nabla\varphi\sno_\2^2 \bigg| \le C (|y|^3+y^4).
\end{align}
From this bound and from $v\in L^1\cap L^\infty$ it is also easy to deduce that $|\cdot|^{-2}v\in L^1\cap L^\infty$. Coming back to \eqref{eq:commutator:ly0:part} we write
\begin{align}
  \int \D y \, v(y) n_{0,1}(y)\, =\, \int \D y\, v(y) e^{-\alpha^2\lambda y^2}+\int \D y \,y^2\, |y|^{-2}v(y)  \big( n_{0,1}(y)-e^{-\alpha^2\lambda y^2} \big)
\end{align}
and apply Lemma \ref{lem: Gaussian lemma} with $g=|\cdot|^{-2}|v|$ to estimate
\begin{align}\label{eq: bound on the error E_21^0}
  \left|\int \D y\, y^2\, |y|^{-2}v(y) \big(  n_{0,1}(y)-e^{-\alpha^2\lambda y^2} \big) \right|\, \leq\,  C\alpha^{-6}.
\end{align}
Using \eqref{eq: bd on v_x}, the definition of $\lambda=\frac{1}{6}\sno \nabla \varphi \sno_\2^2$ and $\int y^2 e^{-y^2}\D y= \frac{3}{2}\pi^{3/2}$, we further have
\begin{align}
  \bigg| \int  \D y \, v(y) e^{-\alpha^2\lambda y^2} + \frac{3}{2\alpha^2}\left(\frac{\pi}{\lambda \alpha^2}\right)^{3/2} \bigg| \, \le\,  C\alpha^{-6}.
\end{align} 
In a similar way, we can estimate the error in \eqref{eq:commutator:ly0:part} by
\begin{align}
\int \D y \, |v(y)|\, \mathfrak{p}_{\alpha}(y)\, n_{0,1}(y) \leq C\alpha^{-6},
\end{align}
which concludes the analysis of the first term in $\mathcal E_{21}^{\Gamma,0}$.

For the second term, we apply Lemma \ref{lem:X:gamma:W:chi} with $\mathbf X  = \phi(\underline{l_y^1})$ and $\mathbf Y  = 1 $ to find
\begin{align}
& \left|\alpha^{-1}\int \D y\,   \lsp \gamma | \phi(\underline{l_y^1}) W(\alpha \widetilde{w}_{P,y}) \gamma \rsp_\Fock \right|    \le  \frac{C}{\alpha} \int \D y\, \sno \phi(\underline{l_y^1}) \gamma \sno_\Fock\, ( 1 + \mathfrak{p}_{\alpha}(y)) \, n_{0,1}(y),
\end{align}
where we used that $| \langle \gamma |  \phi(\underline{l_y^1}) \gamma \rangle _\Fock | \le C \sno l_y^1\sno_\2 $. Lemma \ref{lemma: props_lx} and Eq. \eqref{eq:lys} imply that $y\mapsto |y|^{-2} \sno \underline{l_y^1}\sno_\2 \in L^1\cap L^{\infty}$. By Corollary \ref{cor: Gaussian for errors} the expression on the r.h.s. is thus bounded by $C\alpha^{-6}$. Collecting all terms and invoking Proposition \ref{prop:overlap} gives the final estimate for \eqref{eq:E21}
\begin{align}\label{eq: bound for E_21^0}
 \bigg| \mathcal{E}_{21}^{\Gamma,0} +  \frac{3}{2\alpha^2} \, \sno \mathcal S_P \Gamma\sno_\Fock^2 \, \bigg| \, \le\,  C (n,\varepsilon) \sqrt {K } \alpha^{- 6 + \varepsilon } .
\end{align}

For the term $\mathcal{E}^{\Gamma,P}_{21}$ in \eqref{eq:E21:P} we recall $j_y \in \text{Ran}(\Pi_0)$ and thus we can use \eqref{eq:phi(g0):term} and apply Lemma \ref{lem:X:gamma:W:chi} with $\mathbf X  = \langle j_y | w_{P,y}^0   \rangle_\2 $ and $\mathbf Y  = 1$. This gives
\begin{align}\label{e_p}
  \mathcal{E}^{\Gamma,P}_{21} &  \leq     \int \D y\,  \big|  \lsp j_y^0 |  w^0_{P,y} \rsp_\2 \big  | \,  (1+\mathfrak{p}_{\alpha}(y))\, n_{0,1}(y) .
\end{align} 
Using $\sno j_y \sno_\2 \leq C H(y) |P| \alpha^{-2} \leq C H(y) \alpha^{-1}$ (since  $|P|\leq c \alpha$) and $ \sno w_{P,y}^0 \sno_\2 \leq C(|y|+|y|^3)$ as follows from Lemma \ref{lem: bound for w1 and w0}, we obtain $\mathcal{E}^{\Gamma,P}_{21}  \le C\alpha^{-6}$. 

The term $\mathcal{E}^\Gamma_{22}$ is estimated similarly as the term $\mathcal E_{22}$ in \cite[Section 3.6]{MMS22}, with the result 
\begin{align}
|\mathcal{E}^\Gamma_{22}|\leq  C(n,\varepsilon)  \sqrt K \alpha^{-6+\varepsilon }.
\end{align}
Here, the only difference compared to \cite{MMS22} is that we use Lemmas \ref{lem: bounds for P_f:new} and \ref{lem:error:terms:e-kappaN} in appropriate places; we shall skip the details. 

Combining the relevant estimates, we arrive at the statement of Proposition \ref{prop: bound for E}.
\end{proof}

\subsection{Energy contribution $\mathcal G^\Gamma$}

The energy contribution $\mathcal G^\Gamma$ defined in \eqref{eq: G} is evaluated by the following proposition.

\begin{proposition} \label{prop: bound for G} Let $\mathcal V_K^{(n+1)}$ be defined as in \eqref{eq:Bogo:subspace} and $\mathbb H_K$ be defined by \eqref{eq: Bogoliubov Hamiltonian maintext}. For all $\varepsilon >0$ and $n\in \mathbb N_0$ there exist constants $C(n,\varepsilon) >0$ and $\alpha_0\ge 1$ such that 
\begin{align}
\bigg| \, \mathcal G^\Gamma  -  \frac{2}{\alpha^{2}} \lsp \Gamma| ( \mathbb H_K - \mathbb N_1 )  \Gamma \rsp_\Fock \,  \sno \mathcal S_P \Gamma\sno_\Fock^2 \bigg| \, \le \, C (n,\varepsilon) \,  \alpha^\varepsilon \big(  \sqrt K \alpha^{-6} + K^{-1/2} \alpha^{-5} \big) 
\end{align}
for all normalized $\Gamma \in \mathcal V_K^{(n+1)}$, $| P | \le \sqrt{2M^{\rm LP}} \alpha$, $K\ge K_0$ and $\alpha \ge \alpha_0$.
\end{proposition}
\begin{proof} Using $h^{\rm Pek} \mathscr G_\Gamma^0 = 0$ and $\mathbb N \mathscr G_\Gamma^0 = \mathbb N_1 \mathscr G_\Gamma^0$ we can decompose $\mathcal G^\Gamma$ into the two terms
\begin{align}
\mathcal G^\Gamma  & \, =\,  - \frac{2}{\alpha} \int \D y\, \re \lsp \mathscr G_\Gamma^0 | (\alpha^{-2} \mathbb N_1 + \alpha^{-1} \phi (h_\cdot + \varphi_P) ) T_y e^{A_{P,y}} W(\alpha w_{P,y}) \mathscr G_\Gamma^1 \rsp_\h \notag\\
&\,  =\,  \mathcal G_1^\Gamma  + \mathcal G_2^\Gamma ,
\end{align}
where the first term will contribute to the error while the second one provides an energy contribution of order $\alpha^{-2}$.\medskip

\noindent \underline{Term $\mathcal G_1^\Gamma $}. 
Using Lemma \ref{lem:error:terms:e-kappaN} this term can be treated in close analogy to the estimation of $\mathcal G_1$ in \cite[Section 3.7]{MMS22}. This results in
\begin{align} \label{eq: bound G_1 new}
| \mathcal G^\Gamma_1 | \le C(\varepsilon) \alpha^{-6 + \varepsilon}.
\end{align}

\noindent \underline{Term $\mathcal G_2^\Gamma$}. Here we start by writing
\begin{align}
\mathcal G_{2}^\Gamma &  =   -  \frac{2}{\alpha^2} \int \D y\,  \re \lsp \mathscr G_\Gamma^0 | \phi (h_{\cdot} + \varphi_P) T_y W(\alpha w_{P,y}) \mathscr G_\Gamma^1 \rsp_\h \notag \\
& \quad  - \, \frac{2}{\alpha^2} \int \D y \, \re \lsp \mathscr G_\Gamma^0 | \phi (h_\cdot + \varphi_P ) T_y ( e^{ A_{P,y}} -1 )  W(\alpha w_{P,y}) \mathscr G_\Gamma^1 \rsp_\h = \mathcal G_{21}^\Gamma + \mathcal G_{22}^\Gamma \label{eq: def of G_22}.
\end{align}
In the first term we insert $1 = \mathbb U_K^\dagger \mathbb U_K$ next to $\mathscr G_\Gamma^0$ and bring $\mathbb U_K^\dagger$ to the right side of the inner product. With $\mathbb U_K\mathscr G^0_\Gamma = \psi \otimes \gamma = \mathscr G_\gamma^0 $ for some $\gamma \in \mathcal F_1^{(\le \mathfrak m)}$, \eqref{eq: transformation of phi} and \eqref{eq: identity for tilde w transformation} this gives
\begin{align}
\mathcal G_{21}^\Gamma & \, =\,  - \frac{2}{\alpha^2}  \int \D y \, \re \lsp \mathscr G_\gamma^0  | \phi(\underline { h_\cdot + \varphi_P } )T_y  W(\alpha \widetilde w_{P,y})  u_\alpha R  \phi(\underline{h_{K,\cdot}^1} ) \mathscr G_\gamma^0 \rsp_\h
\end{align}
where the notation $\underline{\, \cdot\, }$ is defined in \eqref{eq: underline notation}. Next we use $ h_\cdot + \varphi_P = h_\cdot^1 + \varphi + h_\cdot^0 + \xi_P $ to write
\begin{align}
\mathcal G_{21}^\Gamma & \, =\,  - \frac{2}{\alpha^2}  \int \D y \, \re \lsp \mathscr G_\gamma^0 | \phi(\underline { h_\cdot^1 + \varphi } )T_y  W(\alpha \widetilde w_{P,y})  u_\alpha R  \phi(\underline{h_{K,\cdot}^1} )\mathscr G_\gamma^0 \rsp_\h \notag\\
& \quad \,  - \frac{2}{\alpha^2}  \int \D y \, \re \lsp \mathscr G_\gamma^0 
| \phi(\underline { h_\cdot^0 + \xi_P } )T_y  W(\alpha \widetilde w_{P,y})  u_\alpha R  \phi(\underline{h_{K,\cdot}^1} )  \mathscr G_\gamma^0 \rsp_\h \notag\\[1mm]
& \, = \, \mathcal G_{211}^\Gamma  + \mathcal G_{212}^\Gamma  .
\end{align}
Since $ h_x^1 + \varphi \in \text{Ran}(\Pi_1)$ for every $x\in \mathbb R^3$, we can apply Lemma \ref{lem:X:gamma:W:chi} in the first term with 
\begin{align}\label{eq:def:X:Y:in:G21}
\mathbf X = \mathbf X_y = \PP \phi(\underline { h_{\cdot}^1 + \varphi } )T_y  u_\alpha R^{1/2} , \quad \mathbf Y  = R^{1/2} \phi(\underline{h_{K,\cdot}^1} )  \PP,
\end{align}
and thus
\begin{align}\label{eq: G211 Gamma intermediate}
& \bigg| \mathcal G_{211}^\Gamma  + \frac{2}{\alpha^2}  \int \D y \, \re \lsp \mathscr G_\gamma^0 | \mathbf X_y \mathbf Y  \mathscr  G_\gamma^0 \rsp_\h\, n_{0,1}(y) \bigg| \notag\\
&\hspace{1.5cm} \le \frac{2}{\alpha^3} \int \D y\, n_{0,1}(y)\, \mathfrak p_\alpha(y) \sno \mathbf X^\dagger_y  \mathscr G_\gamma^0 \sno_\h  \sno \mathbf Y  \mathscr G_\gamma^0 \sno_\h .
\end{align}
To estimate the error we use Lemma \ref{lem: LY CM} to bound
\begin{subequations}
\begin{align}
\sno \mathbf X^\dagger_y  \mathscr G_\gamma^0 \sno_\h & = \sno R^{1/2} \phi(\underline{ h_{\cdot-y}^1 + \varphi} ) u_\alpha T_{-y} \PP \mathscr G_\gamma^0\sno_\h \le C f_{2,\alpha}(y) \sno (\mathbb N+1)^{1/2} \gamma \sno_\Fock \\[1mm]
\sno \mathbf Y   \mathscr G_\gamma^0 \sno_\h & = \sno R^{1/2} \phi(\underline{h_{K,\cdot}^1 }) \PP \mathscr G_\gamma^0 \sno_\h \le C \sno (\mathbb N+1)^{1/2} \gamma\sno_\Fock \label{eq: bound for Y G}
\end{align}
\end{subequations}
with $f_{2,\alpha}(y) = \sno \nabla u_\alpha T_{-y} \PP \sno_\op + \sno T_{-y}u_\alpha \PP\sno_\op $. From \cite[Eq. (3.183)]{MMS22} we know that
\begin{align}\label{eq:def of f2}
\sno f_{2,\alpha} \sno_\su \le C, \quad \sno |\cdot |^n f_{2,\alpha} \sno_\1 \le C(n) \alpha^{3+n} \quad \text{for all}\ n \in \mathbb N_0,
\end{align}
and thus we can apply Corollary \ref{cor: Gaussian for errors} to estimate the r.h.s. in \eqref{eq: G211 Gamma intermediate}, the result being
\begin{align}
\bigg| \mathcal G_{211}^\Gamma  + \frac{2}{\alpha^2}  \int \D y \, \re \lsp \mathscr G_\gamma^0 | \mathbf X_y \mathbf Y    \mathscr G_\gamma^0 \rsp_\h\, n_{0,1}(y) \bigg|  \le C \alpha^{-6}.
\end{align}
In the remaining term we use $\gamma = \mathbb U_K \Gamma$ such that
\begin{align}
\lsp \mathscr G_\gamma^0 | \mathbf X_y \mathbf Y \mathscr   G_\gamma^0 \rsp_\h = \lsp \mathscr G_\Gamma^0 | \phi(  { h_\cdot^1 + \varphi } )T_y  u_\alpha R  \phi( {h_{K,\cdot}^1} )\mathscr G_\Gamma^0 \rsp_\h.
\end{align}
By the same argument as in the computation of $\mathcal G_{211}$ in \cite[Eqs. (3.172)--(3.187)]{MMS22}, one finds
\begin{align}\label{eq: bound G211 new}
\bigg| \mathcal G_{211}^\Gamma - \frac{1}{\alpha^2} \lsp \Gamma |  (\mathbb H_K -\mathbb N_1) \Gamma  \rsp_\Fock \, \sno \mathcal S_P \Gamma \sno_\Fock^2    \bigg| &   \, \le \,  C (n,\varepsilon) \alpha^{\varepsilon}   \Big( \sqrt K  \alpha^{-6} + K^{-1/2} \alpha^{-5} \Big) .
\end{align}

To estimate $\mathcal G_{212}^\Gamma $ we first note that $h_{\cdot}^0 +  \xi_P \in \text{Ran}(\Pi_0)$ (hence $\underline{h_{\cdot}^0 + \xi_P } = h_{\cdot}^0 + \xi_P $, cf. \eqref{eq: underline notation}). From \eqref{eq:phi(g0):term} for $g_x^0 = h_x^0 + \xi_P$, $\Psi = \PP \mathscr G_\gamma^0$ and $\Phi = T_y  u_\alpha R  \phi( \underline{ h_{K,\cdot}^1}  )  \mathscr G_\gamma^0  $  it follows that 
\begin{align} \label{eq:G212:Gamma}
\mathcal G_{212}^\Gamma  = - \frac{2}{\alpha }  \int \D y \, \re \lsp \mathscr G_\gamma^0 
| \langle h_\cdot^0 + \xi_P |  w_{P,y}^0 \rangle_\2  T_y  W(\alpha \widetilde w_{P,y})  u_\alpha R  \phi(\underline{h_{K,\cdot}^1} )  \mathscr G_\gamma^0 \rsp_\h.
\end{align}
Thus we can apply Lemma \ref{lem:X:gamma:W:chi} with
\begin{align}
\mathbf X =  \mathbf X_y =  \PP \lsp h_\cdot^0 + \xi_P |  w_{P,y}^0 \rsp_\2 T_y  u_\alpha R^{1/2}, \quad \mathbf Y  =  R^{1/2} \phi(\underline{h_{K,\cdot}^1} )  \PP,
\end{align}
to obtain
\begin{align}\label{eq: G212 intermediate new}
& \bigg| \mathcal G_{212}^\Gamma + \frac{2}{\alpha} \int \D y \, \lsp \mathscr G_\gamma^0 | \mathbf X_y \mathbf Y   \mathscr G_\gamma^0 \rsp_\h n_{0,1}(y) \bigg| \notag\\
& \hspace{3cm} \le \frac{C}{\alpha^2 }   \int \D y\,  \sno \mathbf X^\dagger_y \mathscr G_\gamma^0 \sno_\h  \sno \mathbf  Y  \mathscr G_\gamma^0 \sno_\h\, \mathfrak p_\alpha(y) \, n_{0,1}(y) .
\end{align}
To bound the norm
\begin{align}
\sno \mathbf X^\dagger_y \mathscr G_\gamma^0 \sno_\h = \sno R^{1/2}  \langle   w_{P,y}^0 |h_{\cdot-y}^0 +\xi_P   \rangle_\2 u_\alpha T_{-y} \PP \sno_\op
\end{align}
we write out the inner product, use the triangle inequality and Cauchy--Schwarz and then apply Corollary \ref{cor: X h Y bounds},
\begin{align}\label{eq:bound:XGg0}
\sno \mathbf X^\dagger_y \mathscr G_\gamma^0 \sno_\h & \le \int \D z \, |  w_{P,y}^0(z)| \, \big( \sno R^{1/2} h_{\cdot-y}^0(z) u_\alpha T_{-y} \PP \sno_\op  + |\xi_P(x)|  \sno R^{1/2}  u_\alpha T_{-y} \PP \sno_\op  \big)  \notag\\[1mm]
& \le  \sno   w_{P,y}^0 \sno_\2 \bigg[ \bigg( \int \D z  \sno R^{1/2} h_{\cdot-y}^0(z) u_\alpha T_{-y} \PP \sno_\op^2  \bigg)^{1/2} + \sno \xi_P\sno_2  \sno R^{1/2} u_\alpha T_{-y} \PP\sno_\op \bigg] \notag\\[3mm]
& \le C \sno   w_{P,y}^0 \sno_\2 (1+\sno \xi_P \sno_2) f_{2,\alpha}(y)
\end{align}
with $f_{2,\alpha}(y)$ as defined before \eqref{eq:def of f2}. Using $\sno \xi_P\sno_\2 \le C \alpha^{-1} $ and invoking Lemma \ref{lem: bound for w1 and w0} thus gives $\sno \mathbf X^\dagger_y \mathscr G_\gamma^0 \sno_\h \le  C   ( | y| + |y|^3  ) f_{2,\alpha}(y)$. 
Using in addition \eqref{eq: bound for Y G} and \eqref{eq:def of f2}, we can thus bound the error in \eqref{eq: G212 intermediate new} by means of Corollary \ref{cor: Gaussian for errors}. This leads to
\begin{align}\label{eq: bound G212 new0}
& \bigg|\mathcal G_{212}^\Gamma + \frac{2}{\alpha }  \int \D y \, \re \lsp \mathscr G_\gamma^0 | \mathbf X_y \mathbf Y    \mathscr G_\gamma^0 \rsp_\h\, n_{0,1}(y)  \bigg| \le \frac{C}{\alpha^{6}}.
\end{align}
The remaining term on the l.h.s. only contributes to the error. To see this, we replace the weight function by the Gaussian, i.e. we use \eqref{eq:bound:XGg0}, \eqref{eq: bound for Y G} and \eqref{eq:def of f2} in order to be able to apply Lemma \ref{lem: Gaussian lemma}
\begin{align}\label{eq:ref:real:part}
\bigg| \frac{2}{\alpha} \int \D y \, \re \lsp \mathscr G_\gamma^0 | \mathbf X_y \mathbf Y   \mathscr G_\gamma^0 \rsp_\h\, ( n_{0,1}(y) -e^{-\lambda \alpha^2 y^2 })   \bigg| \le \frac{C}{\alpha^{6}}.
\end{align}
Then we write $\mathbf X_y =   {\mathbf X}_{y}^{(1)} +  \mathbf X_y^{(2)} +  {\mathbf X}_y^{(3)}$ with
\begin{subequations}
\begin{align}
\mathbf X_y^{(1)} &   = \PP \lsp h_{K,\cdot}^0  |  w_{P,y}^0 \rsp_\2   u_\alpha R^{1/2} \\[1mm]
   \mathbf X_y^{(2)}   & =  \PP \lsp h_{\cdot}^0 - h_{K,\cdot}^0 + \xi_P |  w_{P,y}^0 \rsp_\2 T_y  u_\alpha R^{1/2} \\[1mm]
  \mathbf X_y^{(3)}   & =  \PP \lsp h_{K,\cdot}^0  |  w_{P,y}^0 \rsp_\2  (T_y - 1 ) u_\alpha R^{1/2}.
\end{align}
\end{subequations}
The contribution from $ \mathbf X_y^{(1)}$ is zero since $n_{0,1}(y)$ is an even function whereas $\re ( w_{P,-y}^0)(z) = - \re(w_{P,y}^0)(z)$, cf. Remark \ref{rem: symmetries_of_w}, and thus
\begin{align}
 \re \langle \mathscr G_\gamma^0 | \mathbf X_{-y}^{(1)} \mathbf Y   \mathscr G_\gamma^0 \rangle_\h  & =   \re \langle \mathscr G_\gamma^0 | \PP \lsp h_{K,\cdot}^0  | \re( w_{P,-y}^0) \rsp_\2   u_\alpha R^{1/2}  \mathbf Y   \mathscr G_\gamma^0 \rangle_\h \notag\\[1mm]
  & =  - \re \langle \mathscr G_\gamma^0 | \PP \lsp h_{K,\cdot}^0  | \re( w_{P,y}^0) \rsp_\2   u_\alpha R^{1/2}  \mathbf Y   \mathscr G_\gamma^0 \rangle_\h \notag\\[1mm]
   & = -\re \lsp \mathscr G_\gamma^0 | \mathbf X_y^{(1)} \mathbf Y \mathscr G_\gamma^0 \rsp_\h
\end{align}
where we used that the part involving $\im ( w_{P,\pm y}^0)$ vanishes because of the real part in front of the inner product. For the second term we proceed with Cauchy--Schwarz and \eqref{eq: bound for Y G}
\begin{align}\label{eq:bound:bfX:2}
& \bigg| \frac{2}{\alpha } \int \D y \, \re \lsp \mathscr G_\gamma^0 |  \mathbf X_y^{(2)} \mathbf Y G_\gamma^0 \rsp_\h\,  e^{-\lambda \alpha^2 y^2} \bigg|   \le  \frac{C}{\alpha } \int \D y  \, \sno ( \mathbf X_y^{(2)})^\dagger \mathscr G_\gamma^0\sno_\h   \, e^{-\lambda \alpha^2 y^2}
\end{align}
and bound the remaining norm similarly as in \eqref{eq:bound:XGg0} by 
\begin{align}
 &  \sno (\mathbf X_y^{(2)})^\dagger \mathscr G_\gamma^0\sno_\h  = \sno R^{1/2} \langle w_{P,y}^0|   ( h^0_{\cdot,-y} - h_{K,\cdot-y}^0 +   \xi_P \rangle_\2 u_\alpha T_{-y} \PP \sno_{\op} \notag\\[1.5mm]
& \quad \quad \le C \sno   w_{P,y}^0 \sno_\2 (K^{-1/2} +\sno \xi_P \sno_2) f_{2,\alpha}(y) \le C (|y|+|y|^3) (K^{-1/2} + \alpha^{-1}) 
\end{align}
where we used $f_{2,\alpha}(y)\le C$ and $\sno \xi_P \sno_\2 \le C\alpha^{-1}$. Thus by Corollary \ref{cor: Gaussian for errors}
\begin{align} 
\bigg| \frac{2}{\alpha} \int \D y \, \re \lsp \mathscr G_\gamma^0 | \mathbf X_y^{(2)} \mathbf Y   \mathscr G_\gamma^0 \rsp_\h\, e^{-\lambda \alpha^2 y^2 }    \bigg| \le C\alpha^{-5} \big( K^{-1/2} + \alpha^{-1} \big). 
\end{align}
In the last term we insert $T_y -1 = \int_0^1 \D s T_{sy} (y \nabla)$, such that
\begin{align}
& \big| \re \lsp \mathscr G_\gamma^0 | \mathbf X_y^{(3)}  \mathbf Y   \mathscr G_\gamma^0 \rsp_\h\, \big| \notag \\[1mm]
& \quad  \le \int_0^1 \D s  \big| \re \lsp \mathscr G_\gamma^0 | \PP \lsp h_{K,\cdot}^0 |  w_{P,y}^0 \rsp_\2   T_{sy} (y \nabla)  u_\alpha R^{1/2} \mathbf Y  \mathscr G_\gamma^0 \rsp_\h\, \big| \notag\\[1mm]
& \quad \le C \sqrt K \, |y| \, \sno w_{P,y}^0\sno_\2  \sno \nabla u_\alpha R^{1/2} \sno_\op \sno \mathbf Y \mathscr G_\gamma^0 \sno_\h \le C \sqrt K (y^2+y^4)
\end{align}
where we used $ \sno h_{K,\cdot}^0  \sno_\2 \le C \sqrt K$, $\sno \nabla u_\alpha R^{1/2} \sno_\op \le C$ and \eqref{eq: bound for Y G}. Again by Corollary \ref{cor: Gaussian for errors}
\begin{align} 
\bigg| \frac{2}{\alpha} \int \D y \, \re \lsp \mathscr G_\gamma^0 | \mathbf X_y^{(3)} \mathbf Y   \mathscr G_\gamma^0 \rsp_\h\, e^{-\lambda \alpha^2 y^2 }    \bigg| \le  C  \sqrt K \alpha^{-6}.
\end{align}
Combining all estimates, we arrive at the final bound for \eqref{eq:G212:Gamma}
\begin{align}\label{eq: bound G212 new}
\mathcal G_{212}^{\Gamma} \, \le \, C  \big( K^{1/2} \alpha^{-6} + K^{-1/2} \alpha^{-5} \big). 
\end{align}

The estimation for the remaining term $\mathcal G^\Gamma_{22}$ in \eqref{eq: def of G_22} is somewhat lengthy, but works analogous, with some obvious modifications, to the bound for $\mathcal G_{22}$ in \cite[Section 3.7]{MMS22}. We leave out the details and only state the result
\begin{align}\label{eq: bound for G_22^< new}
| \mathcal G_{22}^\Gamma | \, \le \, C(n,\varepsilon )\alpha^{\varepsilon} \big( K^{-1/2} \alpha^{-5} + \sqrt K  \alpha^{-6}\big).
\end{align}

In view of \eqref{eq: bound G_1 new}, \eqref{eq: bound G211 new}, \eqref{eq: bound G212 new} and \eqref{eq: bound for G_22^< new}, the proof of Proposition \ref{prop: bound for G} is now complete.
\end{proof}

\subsection{Energy contribution $\mathcal K^\Gamma$ \label{sec: energy K}}

Recall \eqref{eq: K} for the definition of $\mathcal K^\Gamma$.

\begin{proposition}\label{prop: bound for K} Let $\mathcal V_K^{(n+1)}$ be defined as in \eqref{eq:Bogo:subspace} and $\mathbb H_K$ be defined by \eqref{eq: Bogoliubov Hamiltonian maintext}. For all $\varepsilon >0$ and $n\in \mathbb N_0$ there exist constants $C(n,\varepsilon) >0$ and $\alpha_0\ge 1$ such that 
\begin{align}
\bigg| \mathcal K^\Gamma +   \frac{1}{\alpha^2}  \lsp \Gamma | ( \mathbb H_K - \mathbb N_1 ) \Gamma \rsp_\Fock \, \sno \mathcal S_P \Gamma\sno_\Fock^2\, \bigg| \le C(n,\varepsilon ) \, \alpha^\varepsilon   \big( \sqrt K \alpha^{-6} + K^{-1/2} \alpha^{-5} \big)
\end{align}
for all normalized $\Gamma \in \mathcal V_K^{(n+1)}$, $| P | \le \sqrt{2M^{\rm LP}} \alpha$, $K \ge K_0$ and $\alpha \ge \alpha_0$.
\end{proposition}
\begin{proof} We split this contribution into three terms
\begin{align}
\mathcal K^\Gamma & \, = \, \frac{1}{\alpha^2} \int \D y \,  \lsp \mathscr G_\Gamma^1 | \big( h^{\rm Pek} + \alpha^{-2} \mathbb N + \alpha^{-1} \phi(h_\cdot + \varphi_P) \big)  T_y e^{A_{P,y}} W(\alpha w_{P,y}) \mathscr G_\Gamma^1 \rsp_\h \notag\\
& \, = \,  \mathcal K_1^\Gamma + \mathcal K_2^\Gamma + \mathcal K_3^\Gamma
\end{align}
and note that $\mathcal K_1^\Gamma$ provides the energy contribution of order $\alpha^{-2}$.\medskip

\noindent \underline{Term $\mathcal K_1^\Gamma$}. We start again by separating the main term as follows,
\begin{align}
\mathcal K_1^\Gamma &\,  =\,  \frac{1}{\alpha^2} \int \D y\, \lsp  \mathscr G_\Gamma^1 | h^{\rm Pek} T_y W(\alpha w_{P,y})  \mathscr G_\Gamma^1 \rsp_\h   \notag\\
& \quad +  \frac{1}{\alpha^2} \int \D y \, \lsp \mathscr G_\Gamma^1 |h^{\rm Pek} T_y (e^{A_{P,y}}-1)   W(\alpha w_{P,y})  \mathscr G_\Gamma^1 \rsp_\h  \,  =\,  \mathcal K_{11}^\Gamma + \mathcal K_{12}^\Gamma.
\end{align}
Inserting $\mathscr G_\Gamma^1 = u_\alpha R \phi(h_{K,\cdot}^1) \PP \mathbb U_K^\dagger \gamma$ for $\gamma \in \mathcal F_1^{( \mathfrak m  )}$ we can write the first term as
\begin{align}
 \mathcal K_{11}^\Gamma & \, = \, \frac{1}{\alpha^2} \int \D y \, \lsp \mathscr G_\gamma^0 | \mathbf X_y   W(\alpha \widetilde w_{P,y}) \mathbf Y \mathscr G_\gamma^0 \rsp_\h
\end{align}
where
\begin{align}
  \mathbf X_y = \PP \phi( \underline{h_{K,\cdot}^1 }) R u_\alpha h^{\rm Pek} T_y u_\alpha R^{1/2}, \quad   \mathbf Y =  R^{1/2} \phi( \underline {h_{K,\cdot}^1} ) \PP.
\end{align}
An application of Lemma \ref{lem:X:gamma:W:chi} leads to
\begin{align}\label{eq:K11 Gamma bound intermediate}
& \bigg| 	  \mathcal K_{11}^\Gamma  -  \frac{1}{\alpha^2} \int \D y \, \lsp \mathscr G_\gamma^0 | \mathbf X_y   \mathbf Y \mathscr G_\gamma^0 \rsp_\h\, n_{0,1}(y) \bigg| \notag\\
& \hspace{2cm} \le \frac{C}{\alpha^3} \int \D y \, \sno \mathbf X^\dagger_y \mathscr G_\gamma^0 \sno_\h \sno \mathbf Y \mathscr G_\gamma^0 \sno_\h \, \mathfrak p_\alpha(y) \,  n_{0,1}(y).
\end{align}
Using \cite[Eq. (3.222)]{MMS22} and Lemma \ref{lem: LY CM}, one further verifies that
\begin{align}
\sno \mathbf  X^\dagger_y \mathscr G_\gamma^0 \sno_\h  \le C f_{3,\alpha}(y) \sno (\mathbb N + 1 )^{1/2} \gamma\sno_\Fock , \quad \sno \mathbf Y \mathscr G_\gamma^0 \sno_\h \le C \sno (\mathbb N+1)^{1/2} \gamma \sno_\Fock
\end{align}
where
\begin{align}\label{eq: def f_3}
f_{3,\alpha}(y) &  =  \sno u_\alpha T_{y} u_\alpha \sno_{\op} + \sno (\nabla u_\alpha ) T_{y} u_\alpha \sno_{\op} + \sno  u_\alpha T_{y} (\nabla u_\alpha) \sno_{\op}  + \sno (\nabla u_\alpha) T_{y} (\nabla u_\alpha) \sno_{\op}
\end{align}
satisfies (see \cite[Eq. (3.220)]{MMS22})
\begin{align} \label{eq: bound for f3}
\sno f_{3,\alpha}\sno_{\su}\le 4 \quad \text{and} \quad \sno |\cdot|^n f_{3,\alpha} \sno_\1 \le C(n) \alpha^{3+n}\quad \text{for all}\ n\in \mathbb N_0.
\end{align}
With the aid of Corollary \ref{cor: Gaussian for errors} we can thus bound the error in \eqref{eq:K11 Gamma bound intermediate} to obtain
\begin{align}
& \bigg| 	  \mathcal K_{11}^\Gamma  -  \frac{1}{\alpha^2} \int \D y \, \lsp \mathscr G_\gamma^0 | \mathbf X_y   \mathbf Y \mathscr G_\gamma^0 \rsp_\h\, n_{0,1}(y) \bigg| \le \frac{C}{\alpha^6} .
\end{align}
In the inner product we transform again with the Bogoliubov transformation $\mathbb U_K$, 
\begin{align}
\lsp \mathscr G_\gamma^0 | \mathbf X_y   \mathbf Y \mathscr G_\gamma^0 \rsp_\h = \lsp \mathscr G_\Gamma^0 | \PP \phi( {h_{K,\cdot}^1 }) R u_\alpha h^{\rm Pek} T_y u_\alpha R\phi(  {h_{K,\cdot}^1} ) \PP \mathscr G_\Gamma^0 \rsp_\h.
\end{align}
In the estimation of the remaining expression, we can follow closely the argument from \cite[Eqs. (3.210)--(3.225)]{MMS22}. We state the result without further details,
\begin{align}\label{eq: bound K11 new}
\bigg| \mathcal K_{11}^\Gamma +  \frac{1}{\alpha^2} \lsp \Gamma | ( \mathbb H_K - \mathbb N_1 ) \Gamma \rsp_\Fock \,  \sno \mathcal S_P \Gamma\sno_\Fock^2 \, \bigg| \, \le\, C(n,\varepsilon)  \sqrt K \alpha^{ - 6 + \varepsilon} . 
\end{align}

Using Lemma \ref{lem:error:terms:e-kappaN}, the term $\mathcal K_{12}^\Gamma$ is estimated following the same steps from the bound of $\mathcal K_{12}$ in \cite[Section 3.8]{MMS22}. This leads to
\begin{align}\label{eq: bound K12 new}
|\mathcal K_{12}^\Gamma | \le C(n,\varepsilon)  \sqrt K  \alpha^{-6+\varepsilon}
\end{align}
Also the contributions $\mathcal K_{2}^\Gamma$ and $\mathcal K_{3}^\Gamma$ can be estimated similarly as their counterparts in \cite[Section 3.8]{MMS22}. Here one finds
\begin{align}\label{eq: bound K2 K3 new}
| \mathcal K_{2}^\Gamma | + | \mathcal K_3^\Gamma| \le C(n,\varepsilon) \alpha^{-6 + \varepsilon} .
\end{align}

The proof of Proposition \ref{prop: bound for K} follows from combining \eqref{eq: bound K11 new}, \eqref{eq: bound K12 new} and \eqref{eq: bound K2 K3 new}.
\end{proof}

\subsection{Finishing the proof of Proposition \ref{theorem: main estimate 2}\label{Sec: concluding the proof}}

Combining Propositions \ref{prop: bound for E}, \ref{prop: bound for G} and \ref{prop: bound for K}, we have for all normalized $\Gamma\in \mathcal V_K^{(n+1)}$ that
\begin{align}\label{eq: conclusion bound}
\bigg| \frac{\mathcal E^\Gamma + \mathcal G^\Gamma + \mathcal K^\Gamma }{\sno \mathcal S_P \Gamma\sno_\Fock^2 } - \frac{\langle \Gamma | \mathbb H_K \Gamma \rangle_\Fock} {\alpha^2} + \frac{3}{2\alpha^2} \bigg| & \, \le \, C(n,\varepsilon) \, \alpha^{\varepsilon} \bigg( \frac{ K^{-1/2}\alpha^{-5} + \sqrt K \alpha^{-6} }{\sno \mathcal S_P \Gamma\sno_\Fock^2} \bigg)
\end{align}
for all $|P| \le \sqrt{2M^{\rm LP}} \alpha$, $K \ge K_0 $ and $\alpha \ge \alpha_0$. Applying Proposition \ref{prop:overlap} for $K\le \alpha$ and $\alpha \ge \alpha(n)$ with $\alpha (n)$ large enough, we further obtain $\sno \mathcal S_P \Gamma\sno_\Fock^2 \ge C \alpha^{-3}$ for some constant $C>0$. Hence the right side in \eqref{eq: conclusion bound} is bounded by $C(n,\varepsilon) \, \alpha^\varepsilon (K^{-1/2} \alpha^{-2} + K^{1/2} \alpha^{-3})$. In view of Lemma \ref{lem: energy identity} this completes the proof of Proposition \ref{theorem: main estimate 2}.

\section{Proofs of Auxiliary Lemmas \label{Sec: Remaining Proofs}}

\begin{proof}[Proof of Lemma \ref{lem: regularized Hessian}] For the proof of Items $(i)$ -- $(iv)$ see \cite[Lemma 2.2]{MMS22}. 

To show $(v)$, first note that because of $(iii)$ the operator $1-H_K^{\rm Pek}\ge 0$ is trace-class, hence compact, and thus the spectrum of $H_K^{\rm Pek}$ below one is purely discrete and possibly accumulates at one. Moreover, it follows from $\text{Ker}(R^{1/2}) = \text{Span}\{ \psi \}$, $\psi>0$ and
\begin{align}
\lsp   v | ( 1 - H_K^{\rm Pek} )  v \rsp_\2 & \, = 4 \bigg\|  R^{1/2}  \int \D y \,  v(y)h_{K,\cdot }(y) \psi \bigg\|^2_\2 \label{eq 1-H}
\end{align}
that $v\in \text{Ker}(1-H_K^{\rm Pek})$ if and only if $\int \D y \,  v(y)h_{K,x  }(y) = 0$ for all $x\in \mathbb R^3$. This is the case if and only if $\text{Supp}(\widehat v) \subset [K,\infty)$. Since $\{ v \in L^2(\mathbb R^3)\, |\, \text{Supp}(\widehat v) \subset [0,K) \}$ corresponds to an infinite dimensional subspace, the spectrum $\sigma (H_K^{\rm Pek}) \cap (0,1)$ consists in fact of infinitely many eigenvalues (of finite multiplicity) with accumulation point at one. 

The bound $|\lambda_K^{(n)} - \lambda_\infty^{(n)}| \le CK^{-1/2}$, $n\in \mathbb N_0$, follows from the min-max principle in combination with $|\langle v | (H_K^{\rm Pek}-H^{\rm Pek}) v\rangle_{\2} | \le C K^{-1/2}$ for $v\in \text{Ran}(\Pi_1)$. The latter has been shown in \cite[Proof of Lemma 2.2]{MMS22}. 

To show $(vii)$,  we write $ (1-\lambda_K^{(n)})^{1/2} \mathfrak u_K^{(n)} = (1-H_K^{\rm Pek})^{1/2} \mathfrak u_K^{(n)}$ and use $(iv)$ to obtain 
\begin{align}
(1-\lambda_K^{(n)}) \sno \nabla \mathfrak u_K^{(n)} \sno_\2^2  \le \sno (-i\nabla) (1-H_K^{\rm Pek})^{1/2} \sno_\op^2 \le \sno (- i \nabla) (1-H_K^{\rm Pek})^{1/2 } \sno_\HS^2\le C K.
\end{align}
\end{proof}

\begin{proof}[Proof of Lemma \ref{lem:linear:independence}] We argue by contradiction and assume that $\{\mathcal S_P \Gamma_K^{(j)} \}_{j=0}^n$ are linearly dependent, meaning that there is a set of complex coefficients $(c_j)_{j=0}^n$ with $\max_{0\leq j \leq n} |c_j| > 0$ such that
\begin{align}\label{lin_dep}
0 =  \sum_{j =0  }^n\, c_j \, \mathcal S_P \Gamma_K^{(j)}.
\end{align}
By Proposition \ref{prop:overlap} together with $K\leq \alpha$ and $\langle \Gamma_{K}^{(i)} |  \Gamma_{K}^{(j)} \rangle_\Fock = \delta_{ij} $, we have
\begin{align}\label{eq:overlap}
\bigg| \lsp \mathcal S_P \Gamma_K^{(i)} |\mathcal S_P \Gamma^{(j)}_K \rsp_\Fock - \delta_{ij} \bigg(\frac{\pi}{\lambda  \alpha^2}\bigg)^{3/2}  \bigg| & \, \le\,   C (n, \varepsilon )  \alpha^{-\frac{7}{2}+\varepsilon}.
\end{align}
After taking the inner product of \eqref{lin_dep} with $\mathcal{S}_P \Gamma_K^{(\ell)}$ where $\ell$ is such that $|c_\ell| \ge  |c_j| $ for all $1\le j\le n$, and dividing by $|c_\ell|$, we find
\begin{align}
 \sno \mathcal S_P \Gamma_K^{(\ell)} \sno^2_\Fock   \le  \sum_{\substack{ j =0 \\ j\neq \ell } }^n\,  \, \Big| \lsp \mathcal S_P \Gamma_K^{(\ell)}| \mathcal S_P \Gamma_K^{(j)} \rsp_\Fock \Big|.
\end{align}
Invoking \eqref{eq:overlap} it follows that
\begin{align}\label{eq:low on t}
  \bigg( \frac{\pi}{\lambda \alpha^2} \bigg)^{3/2} - C (n,\varepsilon) ~\alpha^{-\frac{7}{2}+\varepsilon }    \le  C(n,\varepsilon) n \, \alpha^{-\frac{7}{2}+\varepsilon},
\end{align}
which, for $\varepsilon < \tfrac{1}{2}$ and $\alpha \ge \alpha(n)$ with $\alpha(n)$ large enough, is a contradiction. Hence, the set $\{\mathcal S_P \Gamma_K^{(j)} \}_{j=0}^n$ is linearly independent and its linear span has dimension $n+1$.
\end{proof}

\begin{proof}[Proof of Lemma \ref{lem: bounds for P_f:new}] Let $p  := -i\nabla$ and $\Gamma_K^{(n)} $, $ n \in \mathbb N_0$, as in \eqref{eq:excited:states:Gamma}. As explained in Section \ref{sec:Bog:exc:spectrum}, the tuple $\mathfrak J_K^{(n)}$ has maximal length $|\mathfrak J_K^{(n)}|\le \mathfrak m$ and consists of numbers from $\{1,\ldots , n \}$, thus we can write $\mathfrak J_K^{(n)} =\{ j_1, \ldots j_m\}$ with $m\le \mathfrak m$ and $1\le j_\ell \le n $ for all $1\le \ell \le m$. From \cite[Lemma 3.14]{MMS22} we know that 
\begin{align}\label{eq:bound:Pf:Omega}
\sno P_f \mathbb U_K^\dagger \Omega \sno_\Fock \le C \sqrt K.
\end{align}
By unitarity of $\mathbb U_K$, we can thus estimate 
\begin{align}\label{eq:bound:Pf:Gamma}
 & \sno P_f \Gamma_K^{(n)} \sno_\Fock   = \sno ( \mathbb U_K P_f \mathbb U_K^\dagger) a^\dagger(\mathfrak u_K^{(j_1)}) \ldots a^\dagger (\mathfrak u_K^{(j_m)}) \Omega \sno_\Fock \notag \\[1mm]
  &  \quad \le  \sno  a^\dagger(\mathfrak u_K^{(j_1)}) \ldots a^\dagger (\mathfrak u_K^{(j_m)}) ( \mathbb U_K P_f \mathbb U_K^\dagger) \Omega \sno_\Fock + \sno [\mathbb U_K P_f \mathbb U_K^\dagger , a^\dagger(\mathfrak u_K^{(j_1)}) \ldots a^\dagger (\mathfrak u_K^{(j_m)})] \Omega \sno_\Fock \notag\\[1mm]
& \quad  \le C_m \sqrt K + \sno [\mathbb U_K P_f \mathbb U_K^\dagger , a^\dagger(\mathfrak u_K^{(j_1)}) \ldots a^\dagger (\mathfrak u_K^{(j_m)})] \Omega \sno_\Fock.
\end{align}
where we used \eqref{eq: standard estimates for a and a*}, $\sno \mathfrak u_K^{(j_\ell)} \sno_\2 = 1$, $\mathbb U_K P_f \mathbb U_K^\dagger \Omega \in \mathcal F^{(\le 2)}$ and \eqref{eq:bound:Pf:Omega}.
In the second term we use \eqref{eq: def of U} to compute 
\begin{align}\label{eq: U Pf Omega identity:new}
\mathbb{U}_K P_f \mathbb{U}_K^{\dagger} & = \text{Tr}_{L^2}(B_K p B_K) + \sum_{k,l}  \lsp g_k |( A_K p A_K  + B_K  p B_K)  g_l \rsp_\2  a^\dagger(g_k) a(g_l) \notag \\
  & \quad  + \sum_{k,l} \Big( \lsp g_k |A_K p B_K  g_l \rsp_\2 a^\dagger({g_k}) a^\dagger ( \overline{g_l}) +   \lsp g_k | B_K p A_K  g_l \rsp_\2 a (\overline{g_k}) a (g_l) \Big)
\end{align}
where $(g_k)_{k\in \mathbb N}$ is a suitable ONB of $L^2(\mathbb{R}^3)$. The first term is zero by rotational invariance (see also the proof of \cite[Lemma 3.14]{MMS22}). With this we evaluate the commutator applied to the vacuum state as
\begin{align}
[\mathbb U_K P_f \mathbb U_K^\dagger , a^\dagger(\mathfrak u_K^{(j_1)}) \ldots a^\dagger (\mathfrak u_K^{(j_m)})] \Omega =   \mathscr C_1 + \mathscr C_2 
\end{align}
with
\begin{align}
\mathscr C_1 & =  \sum_{k,l} \lsp g_k |(A_K p A_K  + B_K  p B_K)  g_l \rsp_\2  \Big[ a^\dagger(g_k) a(g_l)   , a^\dagger(\mathfrak u_K^{(j_1)}) \ldots a^\dagger (\mathfrak u_K^{(j_m)})\Big]\Omega \notag\\
& =  \sum_{\ell = 1}^m \bigg( \prod_{\substack {k=1 \\ k \neq \ell}}^m  a^\dagger (\mathfrak u_K^{(j_k)}) \bigg) a^\dagger((A_K p A_K  + B_K  p B_K) \mathfrak u_{K}^{(j_\ell)})  \Omega,
\end{align}
and $\mathscr C_2 = 0$ for $m\in \{0,1\}$ while
\begin{align}
\mathscr C_2 & =  \sum_{k,l}  \lsp g_k | B_K p A_K g_l \rsp_\2 \Big[  a (\overline{g_k}) a (g_l), a^\dagger(\mathfrak u_K^{(j_1)}) \ldots a^\dagger (\mathfrak u_K^{(j_m)})\Big]\Omega \notag\\
&   = \sum_{\ell = 1}^m a (A_K p B_K \mathfrak u_{K}^{(j_\ell)})  \bigg( \prod_{\substack {k=1 \\ k \neq \ell}}^m  a^\dagger (\mathfrak u_K^{(j_k)}) \bigg)  \Omega \hspace{2.5cm} \text{for}\ m\ge 2.
\end{align}
We proceed with bounding the first part of the commutator
\begin{align}
 \sno \mathscr C_1 \sno_\Fock & 
 %\le C_m \sum_{\ell=1}^m \sno a^\dagger((A_KpA_K + B_K p B_K) \mathfrak u_{K}^{(j_\ell %)}) \Omega \sno_\Fock \notag\\
 \le C_m \sum_{\ell=1}^m  \big( \sno  A_K p A_K \mathfrak u_{K}^{(j_\ell)} \sno_\2 + \sno  B_K p B_K \mathfrak u_{K}^{(j_\ell)} \sno_\2\big), 
\end{align}
where the second summand is bounded by $\sno B_K p B_K \sno_\op \le \sno B_K \sno_\HS \sno p B_K \sno_\HS \le C K^{1/2}$, cf. Lemma \ref{lem: regularized Hessian}. To bound the first summand, recall that $\mathfrak u_K^{(j_\ell)}$ is a normalized eigenfunction of $H^{\rm Pek}_K$ with eigenvalue $1 > \lambda^{(j_\ell)}_K \ge \beta $, hence also an eigenfunction of $A_K$ whose eigenvalue is uniformly bounded in $K$ and $j_\ell$, see \eqref{eq: A and B on Pi0}. By Lemma \ref{lem: regularized Hessian} we obtain
\begin{align}
\sno A_K p A_K \mathfrak u_{K}^{(j_\ell)} \sno_\2 \le C  \sno p A_K \mathfrak u_{K}^{(j_\ell)} \sno_\2  \le C \sno p  \mathfrak u_{K}^{(j_\ell)} \sno_\2 \le C  \sqrt K  (1-\lambda^{(j_\ell)}_K  )^{-1/2} .
\end{align}
With $\lambda_K^{(j_n)} \le \lambda_K^{(n)}$ as $j_\ell \le n$, we can combine the bounds to
\begin{align}
\sno \mathscr C_1 \sno_\Fock  \le C \sqrt K (1- \lambda_K^{(n)})^{-1/2}.
\end{align}
We proceed (for $m\ge 2$) with
\begin{align}
\sno \mathscr C_2 \sno_\Fock & \le \sum_{\ell = 1}^m \sno a (A_k p B_K \mathfrak u_{K}^{(j_\ell)})  \bigg( \prod_{\substack {k=1 \\ k \neq \ell}}^m  a^\dagger (\mathfrak u_K^{(j_k)}) \bigg)  \Omega \sno_\Fock \notag\\
%& \le \sum_{\ell = 1}^m \sno A_K p B_K \mathfrak u_{K}^{(j_\ell)} \sno_\2 %%%\sno (\mathbb N+1)^{1/2} \bigg( \prod_{\substack {k=1 \\ k \neq \ell}}^m  %a^\dagger (\mathfrak u_K^{(j_k)}) \bigg)  \Omega \sno_\Fock  \notag\\
& \le C_m  \sum_{\ell = 1}^m \sno A_K p B_K \mathfrak u_{K}^{(j_\ell)}  \sno_\2 \le C_m \sno p B_K \sno_{\HS} \le C_m  \sqrt K,
\end{align}
and thus, in combination with \eqref{eq:bound:Pf:Gamma}, we have
\begin{align}\label{eq:Pf:Gamma:n:bound}
 \sno P_f \Gamma_K^{(n)} \sno_\Fock \le C \sqrt K   (1-\lambda^{(n)}_K )^{-1/2}.
\end{align}
for all $n\in \mathbb N_0$ and $K\ge K_0$.

Next let $\Gamma = \sum_{ j=0}^n c_j \Gamma_K^{(j)}$ with $\sum_{j=0}^n |c_j|^2 =1$, and use \eqref{eq:Pf:Gamma:n:bound} to estimate
\begin{align}
\sno P_f \Gamma \sno_\Fock^2 \le C \sqrt K \sum_{j=0}^n |c_j | (1-\lambda^{(j)}_K )^{-1/2}    \le C \sqrt K \sqrt {n+1} ( 1-\lambda^{(n)}_K )^{-1/2}
\end{align}
where we used $\lambda_K^{(j)} \le \lambda_K^{(n)}$ for $j\le n$ and $\sum_{j=0}^n |c_j |\le \sqrt{n+1}$. The last factor is bounded by
\begin{align}
\frac{1}{ 1-\lambda_K^{(n)} } \le 
\frac{1}{ 1-\lambda_\infty^{(n)} } \bigg( 1 + \frac{\lambda_\infty^{(n)}-\lambda_K^{(n)} }{1-\lambda_\infty^{(n)}} \bigg)^{-1} \le \frac{2}{1-\lambda_\infty^{(n)}}
\end{align}
for all $n \in \mathbb N_0$ and $K\ge K(n)$ for large enough $K(n)$, where we used Lemma \ref{lem: regularized Hessian} (vii). This completes the proof of the lemma.
\end{proof}

\begin{proof}[Proof of Lemma \ref{lem:X:gamma:W:chi}] With \eqref{eq: definition of F} and \eqref{eq: def of tilde w:0} -- \eqref{eq: def of tilde w}, we have
\begin{align}
& \lsp  \mathscr G_\gamma^0 |  \mathbf  X W(\alpha \widetilde w_{P,y} ) \mathbf Y \mathscr G_\xi^0 \rsp_\Fock = \lsp \mathscr G_\gamma^0 | \mathbf  X e^{a^\dagger(\alpha \widetilde w_{P,y}^1 ) } e^{-a(\alpha \widetilde w_{P,y}^1 ) } \mathbf  Y \mathscr G_\xi^0 \rsp_\Fock n_{0,1}(y) \notag\\[1mm]
& \quad = \lsp \mathscr G_\gamma^0 | \mathbf  X  \mathbf  Y \mathscr G_\xi^0 \rsp_\Fock n_{0,1}(y) +  \lsp \mathscr G_\gamma^0 | \mathbf  X \big(  e^{a^\dagger(\alpha \widetilde w_{P,y}^{ 1 } ) } e^{-a(\alpha \widetilde w_{P,y}^{ 1 } ) } -1 \big) \mathbf Y \mathscr G_\xi^0 \rsp_\Fock n_{0,1}(y) 
\end{align}
since $\mathbf  X^\dagger \mathscr G_\gamma^0$ and $\mathbf  Y  \mathscr G_\xi^0$ correspond to the vacuum in the $\mathcal F_0$ component. To estimate the error, we write
\begin{align}
&  \lsp \mathscr G_\gamma^0 | \mathbf  X \big(  e^{a^\dagger(\alpha \widetilde w_{P,y}^{ 1 } ) } e^{-a(\alpha \widetilde w_{P,y}^{ 1 } ) } -1 \big) \mathbf  Y \mathscr G_\xi^0 \rsp_\Fock \notag\\[1mm]
&  =  \lsp e^{a(\alpha \widetilde w_{P,y}^{ 1} ) } \mathbf  X^\dagger \mathscr G_\gamma^0 |   \big(  e^{-a (\alpha \widetilde w_{P,y}^{1} ) }  -1 \big) \mathbf Y \mathscr G_\xi^0 \rsp_\Fock   +  \lsp (e^{a (\alpha \widetilde w^{1}_{P,y} ) } -1  ) \mathbf  X ^\dagger \mathscr G_\gamma^0 |  \mathbf  Y \mathscr G_\xi^0 \rsp_\Fock, 
\end{align}
truncate the Taylor expansion of the exponential series at order $\mathfrak m + 2$ (all other terms vanish since $\mathbf  X^\dagger \mathscr G_\gamma^0, \mathbf  Y  \mathscr G_\xi^0 \in  L^2(\mathbb R^3) \otimes \mathcal F_1^{(\le \mathfrak  m + 2 )})$, and apply Cauchy--Schwarz. We show the details for the first term (the second term works the same way)
\begin{align}
& \big| \lsp e^{a(\alpha \widetilde w_{P,y}^{ 1} ) } \mathbf  X^\dagger \mathscr G_\gamma^0 |   \big(  e^{-a (\alpha \widetilde w_{P,y}^{1} ) }  -1 \big) \mathbf  Y \mathscr G_\xi^0 \rsp_\Fock  \big| \notag\\
& \hspace{3cm} \le \sum_{k=0}^{\mathfrak m + 2 } \sum_{ \ell  = 1 }^{ \mathfrak m + 2  }\frac{1}{k! \ell! } \sno a(\alpha \widetilde w^{ 1 }_{P,y} )^k \mathbf X^\dagger \mathscr G_\gamma^0 \sno_\Fock \,  \sno a(\alpha \widetilde w^{ 1 }_{P,y} )^\ell  \mathbf  Y \mathscr G_\xi^0 \sno_\Fock .\label{eq:label:x}
 \end{align}
Next we use \eqref{eq: standard estimates for a and a*} together with $\sno \alpha \widetilde w_{P,y}^{ 1 }\sno_\2 \le C \alpha^{-1}(1 + \alpha^2 y^2)$, cf. Lemma \ref{lem: bound for w1 and w0}, to obtain
\begin{align}
 \eqref{eq:label:x} &  \le  \sum_{k=0}^{\mathfrak m + 2  } \sum_{\ell =1}^{\mathfrak m + 2}  (C\alpha^{-1})^{k+\ell}  (1 + \alpha^2 y^2 )^{k+\ell}   \sno \mathbb N^{k/2} \mathbf X \mathscr G_\gamma^0 \sno_\Fock \sno \mathbb N^{\ell/2} \mathbf Y_y \mathscr G_\xi^0 \sno_\Fock \notag\\
 & \le \alpha^{-1} \sum_{k=0}^{\mathfrak m + 2} \sum_{\ell =1}^{\mathfrak m + 2 }   \frac{C^{k+\ell} (\mathfrak m+2 )^{k/2 + \ell/2}}{k!\ell! } (1 + \alpha |y| )^{2(k+\ell)} \sno \mathbf X^\dagger \mathscr G_\gamma^0 \sno_\Fock \sno \mathbf Y \mathscr G_\xi^0 \sno_\Fock \notag\\[2.5mm]
 &   \le C \alpha^{-1} \, \mathfrak p_\alpha (y ) \, \sno \mathbf X^\dagger \mathscr G_\gamma^0  \sno_\Fock \sno \mathbf Y \mathscr G_\xi^0 \sno_\Fock
\end{align}
with $\mathfrak p_\alpha(y) =  1 + (\alpha |y|)^{4\mathfrak m + 8}$, where we used $\alpha \ge 1$.
\end{proof}
\begin{proof}[Proof of Lemma \ref{lem:error:terms:e-kappaN}] Using \eqref{eq: BCH for Weyl} and $e^{-\kappa \mathbb N }a^\dagger(f) e^{\kappa \mathbb N} = a^\dagger(e^{-\kappa }f)$ one finds
\begin{align}
e^{-\kappa \mathbb N } W(\alpha \widetilde w_{P,y})e^{\kappa \mathbb N} & =  e^{a^\dagger(  \alpha e^{-\kappa}  \widetilde w_{P,y}) } e^{-a(  \alpha e^{-\kappa}  \widetilde w_{P,y})} e^{-a(2 \alpha \sinh(\kappa) \widetilde  w_{P,y} )} e^{ -\tfrac{1}{2} \alpha^2 \sno \widetilde w_{P,y}\sno_\2^2 } \notag\\[1mm]
& = W ( \alpha e^{-\kappa}  \widetilde w_{P,y})   e^{-a(2 \alpha \sinh(\kappa) \widetilde  w_{P,y} )} \exp\big( \tfrac{\alpha^2 ( e^{-2\kappa} - 1 )}{2}  \sno \widetilde w_{P,y} \sno^2_\2 \big),
\end{align}
and thus
\begin{align}
& \sno e^{-\kappa \mathbb N } W(\alpha \widetilde w_{P,y}) \gamma \sno_\Fock =   \sno   e^{ - a ( 2  \alpha \sinh(\kappa) \widetilde  w_{P,y} ) } e^{-\kappa \mathbb N }  \gamma \sno_\Fock\, \exp\big( \tfrac{\alpha^2 ( e^{-2\kappa} - 1 )}{2}  \sno \widetilde w_{P,y} \sno^2_\2 \big).
\end{align}
For $\kappa = 1/ (16 e b \alpha^\delta)$ we can find an $\alpha$-independent constant $\eta>0$ such that for $\alpha$ large enough
\begin{align}
 \exp\big( \tfrac{\alpha^2 ( e^{-2\kappa} - 1 )}{2}  \sno \widetilde w_{P,y} \sno^2_\2 \big) \le \exp\big( - \tfrac{ \eta \alpha^{2(1-\delta)}}{2}  \sno \widetilde w_{P,y} \sno^2_\2  \big)  =  n_{\delta,\eta}(y).
\end{align}
For normalized $\gamma \in \mathcal F^{(\le \mathfrak m )}$, and since $\sno   \widetilde w_{P,y} \sno_\2 \le C( |y|  +  |y|^3)$, cf. Lemma \ref{lem: bound for w1 and w0}, we can further estimate 
\begin{align}
 \sno    e^{ -a ( 2  \alpha \sinh(\kappa) \widetilde  w_{P,y} ) }  e^{-\kappa \mathbb N } \gamma \sno_\Fock  &  \le    \sum_{k=0}^{\mathfrak m}  \frac{ 2^k  \sinh(\kappa)^k \mathfrak m^{k/2} } {k!}  \alpha^k \sno \widetilde  w_{P,y} \sno_\2^k  \notag\\
&   \le   \sum_{k=0}^{\mathfrak m}  \frac{C^k 2^k  \sinh(\kappa)^k  \mathfrak m^{k/2}  } {k!} \big( \alpha |y|  + \alpha |y|^3)^k   \le C \mathfrak p_\alpha(y)
\end{align}
where we used $\alpha\ge 1$ in the last step.
\end{proof}

\hspace{3mm}

\noindent\textbf{Acknowledgments}. We are grateful to  Krzysztof  My\'sliwy and Morris Brooks for valuable discussions. 

\end{spacing}


\begin{thebibliography}{D}

\makeatletter
\renewcommand{\@biblabel}[1]{[#1]\hfill}
\makeatother

\footnotesize{

\bibitem{Brooks22} M. Brooks and R. Seiringer. The Fr\"ohlich polaron at strong coupling -- part I: The quantum correction to the classical energy. Preprint. \emph{\href{https://arxiv.org/abs/2207.03156}{arXiv:2207.03156}}. (2022)

\bibitem{Brooks22b} M. Brooks and R. Seiringer. The Fr\"ohlich polaron at strong coupling -- part II: Energy-momentum relation and effective mass. In preparation. (2022)

\bibitem{JonasPHD17} J. Dahlb\ae k. Spectral analysis of large particle systems. PhD Thesis. Aarhus University. (2017)

\bibitem{Donsker1983}
M.D.\ Donsker and S.R.S.\ Varadhan. Asymptotics for the polaron. \textit{Comm. Pure Appl. Math.\ 36, 505--528}. (1983)

\bibitem{FeliciangeliRS20} 
D.\ Feliciangeli,\ S.\ Rademacher and R.\ Seiringer. Persistence of the spectral gap for the Landau--Pekar equations. \textit{Lett. Math. Phys 111, 19}. (2021)

\bibitem{FeliciangeliRS21} D. Feliciangeli, S. Rademacher and R. Seiringer. 
The effective mass problem for the Landau--Pekar equations. \textit{J. Phys. A Math. Theor. 55, 015201}. (2022)

%\bibitem{Faris72} W.G. Faris. Invariant cones and uniqueness of the ground state for fermion systems. \emph{J. Math. Phys. 13, 1285}. (1972)

%\bibitem{Feynman72} R.P. Feynman. Statistical Mechanics: A Set of Lectures. \emph{Frontiers in Physics}. (1972) 

%\bibitem{FrankS2021} R.L.\ Frank and R.\ Seiringer. Quantum corrections to the Pekar asymptotics of a strongly coupled polaron. \textit{Comm. Pure Appl. Math. 74(3), 544–588}. (2021)

%\bibitem{Froehlich1933} H. Fr\"ohlich. Theory of electrical breakdown in ionic crystals. \emph{ Proc. R. Soc. Lond. A 160, 230--241}. (1937)

%\bibitem{Froehlich1954} H.\ Fr\"ohlich. Electrons in lattice fields. \textit{Adv. in Phys. 3, 325--362}. (1954)

\bibitem{Froehlich1974} J.\ Fr\"ohlich. Existence of dressed one-electron states in a class of persistent models. \textit{Fort. Phys. 22, 159--198}. (1974) 

\bibitem{GriesemerW2016} M.\ Griesemer and A. W\"unsch. Self-adjointness and domain of the Fr\"ohlich Hamiltonian. \emph{J. Math. Phys. 57(10), 021902}. (2016)

\bibitem{LMM22} J. Lampart, D. Mitrouskas and K.  My\' sliwy. On the global minimum of the energy-momentum relation for the polaron. Preprint. \href{https://arxiv.org/abs/2206.14708}{\emph{arXiv:2206.14708}}. (2022)

\bibitem{Landau1948} L.D. Landau and S.I. Pekar. Effective mass of a polaron. \textit{J. Exp. Theor. Phys, 18, 419--423}. (1948)

\bibitem{LeeLowPines} T.D. Lee, F. Low and D. Pines. The motion of slow electrons in a polar crystal. \emph{Phys. Rev. 90, 297}. (1953)

\bibitem{LeopoldMRSS2020} 
N.\ Leopold,\ D. Mitrouskas, S.\ Rademacher,\ B.\ Schlein and R.\ Seiringer. Landau--Pekar equations and quantum fluctuations for the dynamics of a strongly coupled polaron. \emph{Pure Appl. Anal. 3(4), 653--676}. (2021)

\bibitem{Lieb1977} E.H. Lieb. Existence and uniqueness of the minimizing solution of Choquard's nonlinear equation. \textit{St. Appl. Math. 57, 93-105}. (1977)

\bibitem{Lieb1997} E.H. Lieb and L.E. Thomas. Exact ground state energy of the strong-coupling polaron. \textit{Comm. Math. Phys. 183(3), 511--519}. (1997)

\bibitem{Lieb1958} E.H. Lieb and K. Yamazaki. Ground-state energy and effective mass of the polaron. \emph{Phys. Rev. 111, 728}. (1958)

\bibitem{Prokofev2000} A.S.\ Mishchenko, N.V.\ Prokof’ev, A.\ Sakamoto and B.V.\ Svistunov. 
Diagrammatic quantum Monte Carlo study of the Fr\"ohlich polaron. \emph{Phys. Rev. B 62, 6317}. (2000)

\bibitem{Mitrouskas21} D. Mitrouskas. A note on the Fr\"ohlich dynamics in the strong coupling limit. \emph{Lett. Math. Phys. 111, 45}. (2021)

\bibitem{MMS22} D. Mitrouskas, K. My\'sliwy and R. Seiringer. Optimal parabolic upper bound for the energy-momentum relation of the strongly coupled polaron. \href{https://arxiv.org/abs/2203.02454}{\emph{arXiv:2203.02454}}. Preprint. (2022)

\bibitem{Miyao2010} T. Miyao. Nondegeneracy of ground states in nonrelativistic quantum
field theory. \emph{J. Oper. Theory 64(1)}. (2010)

\bibitem{Moeller2006}
J.S. M\o ller. The polaron revisited. \textit{Rev. Math. Phys. 18, 485--517}. (2006)

\bibitem{Pekar54} S.I. Pekar. Untersuchung \"uber die Elektronentheorie der Kristalle. \emph{Berlin, Akad. Verlag.} (1954)

\bibitem{Polzer22} S. Polzer. Renewal approach for the energy-momentum relation of the Fr\"ohlich polaron. Preprint. \href{https://arxiv.org/abs/2206.14425}{\emph{arXiv:2206.14425}}. (2022)

\bibitem{Ruijsenaars77} S.N.M. Ruijsenaars. On Bogoliubov transformations for systems of relativistic charged particles. \emph{J. Math. Phys. 18, 517}. (1977)

\bibitem{Seiringer22}
R. Seiringer. Absence of excited eigenvalues for Fr\"ohlich type polaron models at weak coupling. Preprint. \href{https://arxiv.org/abs/2210.17123}{\emph{arXiv:2210.17123}}. (2022)
 
\bibitem{Shale62} D. Shale. Linear symmetries of free boson fields. \emph{Trans. Amer. Math. Soc. 103, 149--167}. (1962)
 
\bibitem{JPS2007}
J.\,P.\ Solovej.
Many Body Quantum Mechanics. Lecture notes. \url{http://web.math.ku.dk/~solovej/MANYBODY/mbnotes-ptn-5-3-14.pdf}. (2014)

\bibitem{Spohn1988}
H. Spohn. The polaron at large momentum. \textit{J. Phys. A: Math. Gen. 21 1199--1211}. (1988)

}

\end{thebibliography}
\end{document}